\documentclass[11pt,a4paper]{article}
\usepackage{amssymb}
\usepackage{amsmath}
\usepackage{amsfonts}
\usepackage{bbm}
\usepackage{amsthm}
\usepackage{mathrsfs}
\usepackage{hyperref}
\usepackage{marvosym}
\usepackage{color}
\usepackage[margin=2.6cm]{geometry}
\usepackage[all,cmtip]{xy}
\usepackage[utf8]{inputenc}

\definecolor{darkred}{rgb}{0.8,0.1,0.1}
\hypersetup{
     colorlinks=false,         
     linkcolor=darkred,
     citecolor=blue,
}

\theoremstyle{plain}
\newtheorem{theo}{Theorem}[section]
\newtheorem{lem}[theo]{Lemma}
\newtheorem{propo}[theo]{Proposition}
\newtheorem{cor}[theo]{Corollary}

\theoremstyle{definition}
\newtheorem{defi}[theo]{Definition}
\newtheorem{ex}[theo]{Example}

\newtheorem{rem}[theo]{Remark}

\numberwithin{equation}{section}

\def\nn{\nonumber}

\def\bbR{\mathbb{R}}
\def\bbC{\mathbb{C}}
\def\bbN{\mathbb{N}}
\def\bbZ{\mathbb{Z}}
\def\bbT{\mathbb{T}}

\def\ii{{\,{\rm i}\,}}
\def\e{{\,\rm e}\,}

\def\Hom{\mathrm{Hom}}

\def\Ext{\mathrm{Ext}}
\def\Imm{\mathrm{Im}}
\def\Ker{\mathrm{Ker}}

\def\id{\mathrm{id}}

\def\dd{\mathrm{d}}

\def\1{\mathbbm{1}}

\def\ozero{\mathbf{0}}

\def\Curv{\mathrm{curv}}
\def\Char{\mathrm{char}}

\newcommand{\ip}[2]{\left\langle #1,#2 \right\rangle}
\newcommand{\Z}[3]{Z_{#1}(#2 ; #3)}
\newcommand{\C}[3]{C_{#1}(#2 ; #3)}
\newcommand{\Ho}[3]{H_{#1}(#2 ; #3)}
\newcommand{\Co}[3]{H^{#1}(#2 ; #3)}
\newcommand{\DCo}[3]{\widehat{H}^{#1}(#2 ; #3)}

\def\sk{\vspace{2mm}}

\title{%
Differential cohomology and \\ locally covariant quantum field theory
}

\author{%
Christian Becker$^{1,2,a}$, Alexander Schenkel$^{3,4,b}$ \ \ and \ \ Richard J.\ Szabo$^{3,c}$\vspace{4mm}\\
{\small $^1$ Institut f\"ur Mathematik, Universit\"at Potsdam,}\\
{\small Karl-Liebknecht-Str.\ 24--25, 14476 Potsdam, Germany.}\vspace{2mm}\\
{\small $^2$ Institut f\"ur Geometrie, Technische Universit\"at Dresden,}\\
{\small Zellescher Weg 12--14, 01069 Dresden, Germany.}\vspace{2mm}\\
{\small $^3$ Department of Mathematics, Heriot-Watt University,}\\
{\small Colin Maclaurin Building, Riccarton, Edinburgh EH14 4AS, United Kingdom.}\vspace{0.5mm}\\
{\small Maxwell Institute for Mathematical Sciences, Edinburgh, United Kingdom.}\vspace{0.5mm}\\
{\small The Higgs Centre for Theoretical Physics, Edinburgh, United Kingdom.}\vspace{2mm}\\
{\small $^4$ School of Mathematical Sciences, University of Nottingham,}\\
{\small University Park, Nottingham NG7 2RD, United Kingdom.}\vspace{3mm}\\
{\footnotesize \texttt{Email:} $^a$ \texttt{becker@math.uni-potsdam.de} ~,~~$^b$  \texttt{aschenkel83@gmail.com} ~,~~$^c$ \texttt{R.J.Szabo@hw.ac.uk} }
 }

\date{November 2016}

\begin{document}

\maketitle

\begin{abstract}
We study differential cohomology on categories of globally hyperbolic Lorentzian manifolds.
The Lorentzian metric allows us to define a natural transformation
whose kernel generalizes Maxwell's equations and fits into a restriction of the
fundamental exact sequences of differential cohomology. We consider
smooth Pontryagin duals of differential cohomology groups, which are subgroups of the character groups.
We prove that these groups fit into smooth duals of the fundamental exact sequences of differential cohomology
and equip them with a natural presymplectic structure derived from a generalized Maxwell Lagrangian.
The resulting presymplectic Abelian groups are quantized using the $\mathrm{CCR}$-functor,
which yields a covariant functor from our categories of globally hyperbolic Lorentzian manifolds 
to the category of $C^\ast$-algebras. We prove that this functor satisfies the causality
and time-slice axioms of locally covariant quantum field theory, but that it violates the locality axiom. 
We show that this violation is precisely due to the fact that our functor has topological subfunctors
describing the Pontryagin duals of certain singular cohomology groups.
As a byproduct, we develop a Fr{\'e}chet-Lie group structure on differential cohomology groups. 
\end{abstract}

\paragraph*{Report no.:} EMPG--14--11

\newpage

\setcounter{tocdepth}{1}

{\baselineskip=12pt
\tableofcontents
}

\bigskip

\section{\label{sec:intro}Introduction and summary}

In \cite{CheegerSimons}, Cheeger and Simons develop the theory of differential characters,
which can be understood as a differential refinement of singular cohomology $\Co{k}{\,\cdot\,}{\bbZ}$.
For a smooth manifold $M$, a differential character is a group homomorphism $h: \Z{k-1}{M}{\bbZ}\to\bbT$ from smooth singular
$k{-}1$-cycles to the circle group $\bbT=U(1)$, which evaluated on smooth singular 
$k{-}1$-boundaries is given by integrating a differential form, the curvature $\Curv(h)$ of $h$.
This uniquely defines the curvature map $\Curv : \DCo{k}{M}{\bbZ}\to \Omega^k_\bbZ(M)$, which  
is a natural group epimorphism from the group of differential characters to the group of $k$-forms with integral periods.
To each differential character one can assign its characteristic class via a second
natural group epimorphism $\Char : \DCo{k}{M}{\bbZ}\to
\Co{k}{M}{\bbZ}$, which is why one calls differential characters
a differential refinement of $\Co{k}{M}{\bbZ}$.
In addition to their characteristic class and curvature, differential characters carry further information
that is described by two natural group monomorphisms $\iota: \Omega^{k-1}(M)/\Omega^{k-1}_\bbZ(M)\to \DCo{k}{M}{\bbZ}$
and $\kappa : \Co{k-1}{M}{\bbT}\to \DCo{k}{M}{\bbZ}$, which map, respectively, to the kernel of 
$\Char$ and $\Curv$.  The group of differential characters $\DCo{k}{M}{\bbZ}$ together with these
group homomorphisms fits into a natural commutative diagram of exact sequences, see e.g.\ (\ref{eqn:csdiagrm}) in the main text.
It was recognized later in \cite{SimonsSullivan,BaerBecker} that this
diagram uniquely fixes
(up to a unique natural isomorphism) the functors $\DCo{k}{\,\cdot\,}{\bbZ}$. It is therefore natural to abstract these considerations
and to define a differential cohomology theory as a contravariant functor from the category of smooth manifolds to the category
of Abelian groups that fits (via four natural transformations) into
the diagram (\ref{eqn:csdiagrm}).
\sk

Differential cohomology finds its physical applications in field theory and
string theory as an efficient way to describe 
the gauge orbit spaces of generalized or higher Abelian gauge theories. The degree $k=2$ differential cohomology group
$\DCo{2}{M}{\bbZ}$ describes isomorphism classes of pairs $(P,\nabla)$ consisting of a $\bbT$-bundle $P\to M$ 
and a connection $\nabla$ on $P$. Physically this is exactly the gauge orbit space of Maxwell's theory of electromagnetism.
The characteristic class map $\Char : \DCo{2}{M}{\bbZ}\to \Co{2}{M}{\bbZ}$ assigns the first Chern class of $P$
and the curvature map $\Curv : \DCo{2}{M}{\bbZ}\to\Omega^2_\bbZ(M)$ 
assigns (up to a prefactor) the curvature of $\nabla$. The topological trivialization 
$\iota: \Omega^{1}(M)/\Omega^{1}_\bbZ(M)\to \DCo{2}{M}{\bbZ}$ identifies gauge equivalence classes 
of connections on the trivial $\bbT$-bundle $P=M\times \bbT$ and
the map $\kappa : \Co{1}{M}{\bbT}\to \DCo{2}{M}{\bbZ}$ is the inclusion of equivalence classes of
flat bundle-connection pairs $(P,\nabla)$. In degree $k=1$, the differential cohomology group
$\DCo{1}{M}{\bbZ}$ describes $\bbT$-valued smooth functions $C^\infty(M,\bbT)$, where $\Char(h)$ gives the
``winding number'' of the map $h:M\to\bbT$ around the circle and
$\Curv(h) =\frac{1}{2\pi \ii} \, \dd\log h $; this field theory
 is called the  $\sigma$-model on $M$ with target space $\bbT$. In degree $k\geq 3$, the differential cohomology groups
 $\DCo{k}{M}{\bbZ}$ describe isomorphism classes of $k{-}2$-gerbes with connection,
 which are models of relevance in string theory; see e.g.\ \cite{Szabo:2012hc} for a general introduction.
 \sk
 
 The goal of this paper is to understand the classical and quantum field theory 
 described by a differential cohomology theory. Earlier approaches 
  to this subject \cite{Freed:2006ya,Freed:2006yc} have focused on the Hamiltonian approach, which 
 required the underlying Lorentzian manifold $M$ to be ultrastatic, i.e.\ that 
 the Lorentzian metric $g$ on $M=\bbR\times \Sigma_M$ is of the form $g=-\dd t\otimes \dd t + h$, 
 where $h$ is a complete Riemannian metric on $\Sigma_M$ that is independent of time $t$. 
 Here we shall instead work in the framework of locally covariant quantum field theory
 \cite{Brunetti:2001dx,Fewster:2011pe}, which allows us to treat generic globally hyperbolic Lorentzian
 manifolds $M$ without this restriction. In addition, our construction of the quantum field theory
 is functorial in the sense that we shall obtain a covariant functor 
 $\widehat{\mathfrak{A}}^k(\,\cdot\,) : \mathsf{Loc}^m\to C^\ast\mathsf{Alg}$ from a suitable category 
 of $m$-dimensional globally hyperbolic Lorentzian manifolds to the category 
 of $C^\ast$-algebras, which describes  quantized observable 
 algebras of a degree $k$ differential cohomology theory.
 This means that in addition to obtaining for each globally hyperbolic Lorentzian manifold $M$ a
 $C^\ast$-algebra of observables $\widehat{\mathfrak{A}}^k(M)$, we get $C^\ast$-algebra morphisms 
 $\widehat{\mathfrak{A}}^k(f) : \widehat{\mathfrak{A}}^k(M)\to \widehat{\mathfrak{A}}^k(N)$ whenever
 there is a causal isometric embedding $f:M\to N$. This in particular provides a mapping of observables
 from certain subregions of $M$ to $M$ itself, which is known to encode essential physical characteristics of the quantum field 
 theory since the work of Haag and Kastler \cite{Haag:1963dh}.
 \sk
 
 Let us outline the content of this paper: In Section \ref{sec:background} we give a short introduction to
 differential cohomology, focusing both on the abstract approach and the explicit model of Cheeger-Simons differential
 characters. In Section \ref{sec:MWmaps} we restrict any (abstract) degree $k$ differential cohomology theory 
 to a suitable category of $m$-dimensional globally hyperbolic Lorentzian manifolds $\mathsf{Loc}^m$, introduce generalized Maxwell
 maps and study their solution subgroups (generalizing Maxwell's equations in degree $k=2$). The solution subgroups
 are shown to fit into a fundamental commutative
 diagram of exact sequences. We also prove that local generalized Maxwell solutions (i.e.\ solutions given solely in a suitable region
  containing a Cauchy surface) uniquely extend to global ones.
 In Section \ref{sec:duality} we study the character groups of the differential cohomology groups.
 Inspired by \cite{HLZ} we introduce the concept of smooth Pontryagin duals,
 which are certain subgroups of the character groups, and prove that
 they fit into a commutative diagram of
 fundamental exact sequences. We further show that the 
 smooth Pontryagin duals separate points of the differential cohomology groups and that they are given by a covariant
 functor from $\mathsf{Loc}^m$ to the category of Abelian groups.
 In Section \ref{sec:presymp} we equip the smooth Pontryagin duals with a natural presymplectic structure,
 which we derive from a generalized Maxwell Lagrangian by adapting Peierls' construction \cite{Peierls:1952cb}.
 This then leads to a covariant functor $\widehat{\mathfrak{G}}^k(\,\cdot\,)$ from $\mathsf{Loc}^m$
  to the category of presymplectic Abelian groups, which describes the classical field theory 
  associated to a differential cohomology theory. The generalized Maxwell equations
 are encoded by taking a quotient of this functor by the vanishing subgroups of the solution subgroups.
 Due to the fundamental commutative diagram of exact sequences for the smooth Pontryagin duals,
 we observe immediately that the functor $\widehat{\mathfrak{G}}^k(\,\cdot\,)$ has two subfunctors,
 one of which is $\Co{k}{\,\cdot\,}{\bbZ}^\star$, the Pontryagin dual of $\bbZ$-valued singular cohomology, 
 and hence is purely topological. The second subfunctor describes ``curvature observables'' and we show that it has 
 a further subfunctor $\Co{m-k}{\,\cdot\,}{\bbR}^\star$. This gives a more direct and natural
 perspective on the locally covariant topological quantum fields described in \cite{Benini:2013tra} for connections
 on a fixed $\bbT$-bundle.
In Section \ref{sec:quantization} we carry out the canonical quantization of our field theory
by using the $\mathrm{CCR}$-functor for presymplectic Abelian groups developed in
 \cite{Manuceau:1973yn} and also in \cite[Appendix A]{Benini:2013ita}.
 This yields a covariant functor $\widehat{\mathfrak{A}}^k(\,\cdot\,) : \mathsf{Loc}^m\to C^\ast\mathsf{Alg}$
 to the category of $C^\ast$-algebras. We prove that $\widehat{\mathfrak{A}}^k(\,\cdot\,) $ satisfies the causality axiom
 and the time-slice axiom, which have been proposed in \cite{Brunetti:2001dx} to single out 
 physically reasonable models for quantum field theory from all covariant functors $\mathsf{Loc}^m\to C^\ast\mathsf{Alg}$.
 The locality axiom, demanding that $\widehat{\mathfrak{A}}^k(f)$ is injective for any $\mathsf{Loc}^m$-morphism $f$,
 is in general not satisfied (except in the special case $(m,k)=(2,1)$). We prove that for a $\mathsf{Loc}^m$-morphism 
 $f:M\to N$ the morphism $\widehat{\mathfrak{A}}^k(f)$ is injective if and only if 
 the morphism $\Co{m-k}{M}{\bbR}^\star\oplus \Co{k}{M}{\bbZ}^\star \to  \Co{m-k}{N}{\bbR}^\star\oplus \Co{k}{N}{\bbZ}^\star$
 in the topological subtheories is injective, which is in general not the case. This provides
 a precise connection between the violation of the locality axiom and the presence of topological subtheories,
 which generalizes the results obtained in \cite{Benini:2013ita} for gauge theories of connections on fixed $\bbT$-bundles.
 In Appendix \ref{app:frechet} we develop a Fr{\'e}chet-Lie group structure on differential cohomology groups,
 which is required to make precise our construction of the presymplectic structure.


\section{\label{sec:background}Differential cohomology}

In this section we summarize some background material on (ordinary)
differential cohomology that will be used in this paper.
In order to fix notation we shall first give a condensed summary of singular 
homology and cohomology. We shall then briefly review the Cheeger-Simons differential characters defined in \cite{CheegerSimons}.
The group of differential characters is a particular model of (ordinary) differential cohomology.
Even though our results in the ensuing sections are formulated in a model independent way,
it is helpful to have the explicit model of differential characters in mind.
\sk

\subsection{Singular homology and cohomology}

Let $M$ be a smooth manifold. We denote by $\C{k}{M}{\bbZ}$ the free Abelian group
of {\it smooth singular $k$-chains} in $M$. There exist boundary maps
$\partial_k :  \C{k}{M}{\bbZ}\to \C{k-1}{M}{\bbZ}$, which are homomorphisms of Abelian groups
satisfying $\partial_{k-1}\circ \partial_{k} = 0$. The subgroup $\Z{k}{M}{\bbZ} := \Ker\, \partial_k$
is called the group of {\it smooth singular $k$-cycles} and it has the obvious
subgroup $B_{k}(M;\bbZ):= \Imm \, \partial_{k+1}$ of {\it smooth singular $k$-boundaries}.
The $k$-th {\it smooth singular homology group} of $M$ is defined as the quotient
\begin{flalign}
\Ho{k}{M}{\bbZ} := \frac{\Z{k}{M}{\bbZ}}{B_k(M;\bbZ)} =  \frac{\Ker\, \partial_k}{\Imm\, \partial_{k+1}}~.
\end{flalign}
Notice that $\Ho{k}{\,\cdot\,}{\bbZ}: \mathsf{Man}\to\mathsf{Ab}$ is a covariant functor from 
the category of smooth manifolds to the category of Abelian groups; for a $\mathsf{Man}$-morphism
$f:M\to N$ (i.e.\ a smooth map) the $\mathsf{Ab}$-morphism $\Ho{k}{f}{\bbZ}:\Ho{k}{M}{\bbZ}\to \Ho{k}{N}{\bbZ}$ 
is given by push-forward of smooth $k$-simplices.
In the following we shall often drop the adjective {\it smooth singular} and simply use the words
$k$-chain, $k$-cycle and $k$-boundary. Furthermore, we shall drop the label $_k$ on the boundary maps
$\partial_k$ whenever there is no risk of confusion.
\sk

Given any Abelian group $G$, the Abelian group of {\it $G$-valued $k$-cochains} is defined by
\begin{flalign}
C^{k}(M;G):= \Hom\big(\C{k}{M}{\bbZ},G\big)~,
\end{flalign} 
where $\Hom$ denotes the group homomorphisms.
The boundary maps $\partial_k$ dualize 
to the coboundary maps $\delta^k : C^{k}(M;G) \to C^{k+1}(M;G)\,,~\phi\mapsto \phi \circ \partial_{k+1}$,
which are homomorphisms of Abelian groups and satisfy $\delta^{k+1}\circ\delta^k =0$. 
Elements in $Z^{k}(M;G):=\Ker\,\delta^k$ are called {\it  $G$-valued $k$-cocycles} and elements
in $B^k(M;G):= \Imm\,\delta^{k-1}$ are called {\it $G$-valued $k$-coboundaries}.
The (smooth singular) {\it cohomology group with coefficients in $G$} is defined by
\begin{flalign}
\Co{k}{M}{G} :=\frac{Z^k(M;G)}{B^k(M;G)} =  \frac{\Ker\,\delta^k}{\Imm\,\delta^{k-1}}~.
\end{flalign}
 Notice that $\Co{k}{\,\cdot\,}{G}:\mathsf{Man}\to \mathsf{Ab}$ is a contravariant functor.
\sk

The cohomology group $\Co{k}{M}{G}$  
is in general not isomorphic to $\Hom(\Ho{k}{M}{\bbZ},G)$.
The obvious group  homomorphism $\Co{k}{M}{G} \to \Hom(\Ho{k}{M}{\bbZ},G)$ 
is in general only surjective but not injective. Its kernel 
is described by the {\it universal coefficient theorem for cohomology} (see e.g.\ \cite[Theorem 3.2]{Hatcher}),
which states that there is an exact sequence
\begin{flalign}\label{eqn:universalcoefficient}
\xymatrix{
0 \ \ar[r] & \ \Ext\big(\Ho{k-1}{M}{\bbZ},G \big) \ \ar[r] & \
\Co{k}{M}{G} \ \ar[r] & \ \Hom\big(\Ho{k}{M}{\bbZ},G \big) \ \ar[r] &
\ 0~.
}
\end{flalign}
In this paper the group $G$ will be either $\bbZ$, $\bbR$ or $\bbT = U(1)$ (the circle group).
As $\bbR$ and $\bbT$ are divisible groups, we have $\Ext(\,\cdot\,, \bbR) = \Ext(\,\cdot\, , \bbT) =0$.
Thus $\Co{k}{M}{\bbR} \simeq \Hom(\Ho{k}{M}{\bbZ},\bbR)$ and $\Co{k}{M}{\bbT} \simeq \Hom(\Ho{k}{M}{\bbZ},\bbT)$.
However $\Ext(\,\cdot\,, \bbZ)\neq 0 $ and hence in general $\Co{k}{M}{\bbZ} \not\simeq \Hom(\Ho{k}{M}{\bbZ},\bbZ)$.
Following the notations in \cite{HLZ}, we denote the image of the Abelian group
 $\Ext(\Ho{k-1}{M}{\bbZ},\bbZ) $ under the group homomorphism in (\ref{eqn:universalcoefficient})
by $H^{k}_\mathrm{tor}(M;\bbZ)\subseteq \Co{k}{M}{\bbZ}$ and call it the {\it torsion subgroup}.
We further denote by $H^{k}_\mathrm{free}(M;\bbZ) := H^{k}(M;\bbZ)/H^{k}_\mathrm{tor}(M;\bbZ)$
 the associated {\it free $k$-th cohomology group with coefficients in $\bbZ$}. By  (\ref{eqn:universalcoefficient}),
the Abelian group $H^{k}_\mathrm{free}(M;\bbZ)$ is isomorphic to $\Hom(\Ho{k}{M}{\bbZ},\bbZ)$
and, by using the inclusion $\bbZ\hookrightarrow\bbR$, we can regard $H^{k}_\mathrm{free}(M;\bbZ)$
as a lattice in $H^{k}(M;\bbR)$.

\subsection{Differential characters}

Let $M$ be a smooth manifold and denote by $\Omega^k(M)$
the $\bbR$-vector space of smooth $k$-forms on $M$.
\begin{defi}
The Abelian group of degree $k$ {\bf differential characters}\footnote{
We use the conventions in \cite{BaerBecker} for the degree $k$ of a differential character,
which is shifted by $+1$ compared to the original definition \cite{CheegerSimons}.
} on $M$, with $1\leq k\in \bbN$, is defined by
\begin{flalign}
\DCo{k}{M}{\bbZ} :=
 \Big\{ h\in \Hom\big(\Z{k-1}{M}{\bbZ},\bbT\big) \ : \ h\circ \partial \in\Omega^k(M)\Big\}~.
\end{flalign}
By the notation $h\circ \partial \in\Omega^k(M)$ we mean that there exists $\omega_h\in\Omega^k(M)$ such
that
\begin{flalign}\label{eqn:charcoboundary}
h(\partial c) = \exp\Big(2\pi \ii\int_c\, \omega_h\Big)
\end{flalign}
for all $c\in \C{k}{M}{\bbZ}$. We further define $\DCo{k}{M}{\bbZ} := \Co{k}{M}{\bbZ}$ for all $0\geq k\in\bbZ$.
\end{defi}

The Abelian group structure on $\DCo{k}{M}{\bbZ}$ is defined pointwise.  As it will simplify
the notations throughout this paper, we shall use an additive notation for the group structure
on $\DCo{k}{M}{\bbZ}$, even though this seems counterintuitive from the perspective 
of differential characters. Explicitly, we define the group operation
$+$ on $\DCo{k}{M}{\bbZ}$ by $(h + l)(z) := h(z)\,l(z)$
for all $h,l\in \DCo{k}{M}{\bbZ}$ and $z\in\Z{k-1}{M}{\bbZ} $. 
The unit element $\ozero\in \DCo{k}{M}{\bbZ}$ is the constant homomorphism $\ozero(z)= 1\in \bbT$ 
and the inverse $-h$ of $h\in \DCo{k}{M}{\bbZ}$ is defined by $(-h)(z)
:= \big( h(z)\big)^{-1}$ for all $z\in\Z{k-1}{M}{\bbZ}$.
\sk

There are various interesting group homomorphisms with the Abelian group $\DCo{k}{M}{\bbZ}$ as target or source.
The first one is obtained by observing that the form $\omega_h\in\Omega^k(M)$ in (\ref{eqn:charcoboundary}) is 
uniquely determined for any $h\in \DCo{k}{M}{\bbZ} $. Furthermore, $\omega_h$ is closed, i.e.\ $\dd \omega_h=0$ with
$\dd$ being the exterior differential, and it has integral periods, i.e.\
$\int_z\, \omega_h\in\bbZ$ for all $z\in \Z{k}{M}{\bbZ}$. We denote the Abelian group of closed $k$-forms with integral periods by
$\Omega^k_\bbZ(M)\subseteq \Omega_\dd^k(M) \subseteq \Omega^k(M)$,
where $\Omega_\dd^k(M)$ is the subspace
of closed $k$-forms. Hence we have found a group homomorphism
\begin{flalign}\label{eqn:curv}
\Curv \, : \, \DCo{k}{M}{\bbZ} \ \longrightarrow \ \Omega^k_\bbZ(M)
\  , \qquad h \ \longmapsto \ \Curv(h) = \omega_h~,
\end{flalign}
which we call the {\it curvature}.
\sk

We can also associate to each differential character its characteristic class, which is an element in $\Co{k}{M}{\bbZ}$.
There exists a group homomorphism 
\begin{flalign}\label{eqn:char}
\Char \, : \, \DCo{k}{M}{\bbZ} \ \longrightarrow \ \Co{k}{M}{\bbZ}
\end{flalign}
called the {\it characteristic class}, which is constructed as follows:
Since $\Z{k-1}{M}{\bbZ}$ is a free $\bbZ$-module, any $h\in\DCo{k}{M}{\bbZ}$ has a real lift
$\tilde h \in \Hom(\Z{k-1}{M}{\bbZ},\bbR)$ such that $h(z) = \exp (2\pi\ii \tilde h(z) )$ for all
$z\in \Z{k-1}{M}{\bbZ}$. We define a real valued $k$-cochain by
$\mu^{\tilde h}: \C{k}{M}{\bbZ} \to \bbR\,,~c\mapsto \int_c\, \Curv(h) - \tilde h(\partial c)$.
It can be easily checked that $\mu^{\tilde h}$ is a $k$-cocycle, i.e.\ $\delta \mu^{\tilde h} =0$,
and that it takes values in $\bbZ$. We define the class
$\Char(h) := [\mu^{\tilde h}] \in \Co{k}{M}{\bbZ}$ and note that it is independent of the
choice of lift $\tilde h$ of $h$.
\sk

It can be shown that the curvature and characteristic class maps are surjective, however, 
in general they are not injective \cite{CheegerSimons}. This means that differential characters have further properties
besides their curvature and characteristic class. In order to characterize these properties
we shall define two further homomorphisms of Abelian groups with $\DCo{k}{M}{\bbZ}$ as target:
Firstly, the {\it topological trivialization} is the group homomorphism
\begin{flalign}
\iota \, : \, \frac{\Omega^{k-1}(M)}{\Omega^{k-1}_\bbZ(M) } \
\longrightarrow \ \DCo{k}{M}{\bbZ}
\end{flalign}
defined by
$\iota([\eta])(z) := \exp(2\pi\ii\int_z\, \eta)$ for all $[\eta]\in \Omega^{k-1}(M)/ \Omega^{k-1}_\bbZ(M)$ and $z\in \Z{k-1}{M}{\bbZ}$. This expression is well-defined since by definition
$\int_z\, \eta\in\bbZ$ for all $\eta\in \Omega^{k-1}_\bbZ(M)$ and $z\in \Z{k-1}{M}{\bbZ}$.
Secondly, the {\it inclusion of flat classes} is the group homomorphism
\begin{flalign}
\kappa \, : \, \Co{k-1}{M}{\bbT} \ \longrightarrow \ \DCo{k}{M}{\bbZ}
\end{flalign}
defined by
$\kappa(u)(z) := \ip{u}{[z]}$ for all $u\in  \Co{k-1}{M}{\bbT}$ and $z\in \Z{k-1}{M}{\bbZ}$, where $[z]\in \Ho{k-1}{M}{\bbZ}$ is the homology class of $z$ 
and $\ip{\,\cdot\,}{\,\cdot\,}$ is the pairing induced by the isomorphism
$\Co{k-1}{M}{\bbT}\simeq \Hom(\Ho{k-1}{M}{\bbZ},\bbT)$ given by the universal coefficient theorem
(\ref{eqn:universalcoefficient}). (Recall that $\bbT$ is divisible.)
\sk

As shown in \cite{CheegerSimons,BaerBecker}, the various group homomorphisms defined above fit into
a commutative diagram of short exact sequences.
\begin{theo}\label{theo:csdiagram}
The following diagram of homomorphisms of Abelian groups commutes and its rows and columns are exact sequences:
\begin{flalign}\label{eqn:csdiagrm}
\xymatrix{
 & 0 \ar[d]& 0\ar[d] & 0\ar[d] & \\
 0 \ \ar[r] & \
 \frac{\Co{k-1}{M}{\bbR}}{H^{k-1}_\mathrm{free}(M;\bbZ)} \ \ar[d]
 \ar[r] & \ \frac{\Omega^{k-1}(M)}{\Omega^{k-1}_\bbZ(M)} \
 \ar[d]^-{\iota} \ar[r]^-{\dd} & \ \dd \Omega^{k-1}(M) \ \ar[r]\ar[d]
 & \ 0\\
 0 \ \ar[r] & \ \Co{k-1}{M}{\bbT} \ \ar[d] \ar[r]^-{\kappa} & \
 \DCo{k}{M}{\bbZ} \ \ar[d]^-{\Char} \ar[r]^-{\Curv} & \ar[d] \
 \Omega^k_\bbZ(M) \ \ar[r] & \ 0\\
\ 0 \ \ar[r] & \ H^{k}_\mathrm{tor}(M;\bbZ) \ \ar[d]\ar[r] & \
\Co{k}{M}{\bbZ} \ \ar[d] \ar[r] & \ H^k_\mathrm{free}(M;\bbZ) \
\ar[d]\ar[r] & \ 0\\
 & 0 & 0 & 0 &
}
\end{flalign}
\end{theo}
\sk

The Abelian group of differential characters $\DCo{k}{M}{\bbZ}$, as well as all other Abelian groups 
appearing in the diagram (\ref{eqn:csdiagrm}), are given by contravariant functors from the 
category of smooth manifolds $\mathsf{Man}$ to the category of Abelian groups $\mathsf{Ab}$. 
The morphisms appearing in the diagram (\ref{eqn:csdiagrm}) are natural transformations.

\begin{ex}
The Abelian groups of differential characters $\DCo{k}{M}{\bbZ}$ can be interpreted as gauge orbit spaces 
of (higher) Abelian gauge theories, see e.g.\ \cite[Examples 5.6--5.8]{BaerBecker} for mathematical 
details and \cite{Szabo:2012hc} for a discussion of physical applications.
\begin{enumerate}
\item In degree $k=1$, the differential characters $\DCo{1}{M}{\bbZ}$
  describe smooth $\bbT$-valued functions on $M$,
i.e.\ $\DCo{1}{M}{\bbZ}\simeq C^\infty(M,\bbT)$. The characteristic class in this case is the ``winding number''
of a smooth map $h\in C^\infty(M,\bbT)$  around the circle, while the curvature is $\Curv(h)=\frac{1}{2\pi\ii}\,\dd\log h$.
Physically the group $\DCo{1}{M}{\bbZ}$ describes the $\sigma$-model
on $M$ with target space the circle $\bbT$.

\item In degree $k=2$, the differential characters $\DCo{2}{M}{\bbZ}$ describe isomorphism classes
of $\bbT$-bundles with connections $(P,\nabla)$ on $M$.
The holonomy map associates to any one-cycle $z\in\Z{1}{M}{\bbZ}$ a group element $h(z)\in\bbT$.
This defines a differential character $h\in \DCo{2}{M}{\bbZ}$, whose
curvature is $\Curv(h) = -\frac{1}{2\pi \ii} \, F_\nabla$ 
and whose characteristic class is the first Chern class of $P$.
The topological trivialization $\iota:\Omega^{1}(M)/ \Omega^{1}_\bbZ(M) \to \DCo{2}{M}{\bbZ}$ assigns to
gauge equivalence classes of connections on the trivial $\bbT$-bundle their holonomy maps.
The inclusion of flat classes $\kappa: \Co{1}{M}{\bbT}\to \DCo{2}{M}{\bbZ}$ assigns to isomorphism classes of
flat bundle-connection pairs $(P,\nabla)$ their holonomy
maps. Physically the group $\DCo{2}{M}{\bbZ}$ describes the
 ordinary Maxwell theory of electromagnetism.

\item In degree $k\geq 3$, the differential characters $\DCo{k}{M}{\bbZ}$ describe
isomorphism classes of $k{-}2$-gerbes with connections, see e.g.\
\cite{Hitchin} for the case of usual gerbes, i.e.\ $k=3$. 
These models are examples of higher Abelian gauge theories where the
curvature is given by a $k$-form, and they physically arise in string theory, see e.g.\ \cite{Szabo:2012hc}.
\end{enumerate}
\end{ex}

\subsection{Differential cohomology theories}

The functor describing Cheeger-Simons differential characters is a specific 
model of what is called a differential cohomology theory. There are also other explicit models
 for differential cohomology, as for example those obtained in smooth Deligne cohomology (see e.g. \cite{Brylinski}), the 
de Rham-Federer approach \cite{HLZ} making use of de Rham currents (i.e.\ distributional differential forms),
 and the seminal Hopkins-Singer model \cite{HS} which is based on differential cocycles and 
 the homotopy theory of differential function spaces. 
 These models also fit into the commutative diagram of exact sequences in (\ref{eqn:csdiagrm}). 
 The extent to which (\ref{eqn:csdiagrm}) determines the functors $\DCo{k}{\,\cdot\,}{\bbZ}$
has been addressed in \cite{SimonsSullivan,BaerBecker} and it turns
out that they are uniquely determined (up to a unique natural isomorphism).
This motivates the following
\begin{defi}[\cite{BaerBecker}]
A {\bf differential cohomology theory} is a contravariant functor $\widetilde{H}^\ast(\,\cdot\,,\bbZ): \mathsf{Man} \to \mathsf{Ab}^\bbZ$
from the category of smooth manifolds to the category of $\bbZ$-graded Abelian groups, together with four natural transformations
\begin{itemize}
\item $\widetilde{\Curv} : \widetilde{H}^\ast(\,\cdot\,;\bbZ) \Rightarrow \Omega^\ast_\bbZ(\,\cdot\,)$ (called {\it curvature})
\item $\widetilde{\Char} : \widetilde{H}^\ast(\,\cdot\,;\bbZ) \Rightarrow H^\ast(\,\cdot\,;\bbZ) $ (called {\it characteristic class})
\item $\widetilde{\iota} : \Omega^{\ast-1}(\,\cdot\,)/ \Omega^{\ast-1}_\bbZ(\,\cdot\,) \Rightarrow \widetilde{H}^\ast(\,\cdot\,;\bbZ)  $
(called {\it topological trivialization})
\item $\widetilde{\kappa} : H^{\ast-1}(\,\cdot\,;\bbT) \Rightarrow \widetilde{H}^\ast(\,\cdot\,;\bbZ)$ (called {\it inclusion of flat classes})
\end{itemize}
such that for any smooth manifold $M$ the following diagram commutes and has exact rows and columns:
\begin{flalign}\label{eqn:absdiagrm}
\xymatrix{
 & 0 \ar[d]& 0\ar[d] & 0\ar[d] & \\
 0 \ \ar[r] & \
 \frac{\Co{\ast-1}{M}{\bbR}}{H^{\ast-1}_\mathrm{free}(M;\bbZ)} \
 \ar[d] \ar[r] & \ \frac{\Omega^{\ast-1}(M)}{\Omega^{\ast-1}_\bbZ(M)}
 \ \ar[d]^-{\widetilde{\iota}} \ar[r]^-{\dd} & \ \dd
 \Omega^{\ast-1}(M) \ \ar[r]\ar[d] & \ 0\\
 0 \ \ar[r] & \ \Co{\ast-1}{M}{\bbT} \ \ar[d]
 \ar[r]^-{\widetilde{\kappa}} & \ \widetilde{H}^{\ast}(M;\bbZ) \
 \ar[d]^-{\widetilde{\Char}} \ar[r]^-{\widetilde{\Curv}} & \ar[d] \
 \Omega^\ast_\bbZ(M) \ \ar[r] & \ 0\\
 0 \ar[r] & \ H^\ast_\mathrm{tor}(M;\bbZ) \ \ar[d]\ar[r] & \
 \Co{\ast}{M}{\bbZ} \ \ar[d] \ar[r] & \ H^\ast_\mathrm{free}(M;\bbZ) \
 \ar[d]\ar[r] & \ 0\\
 & 0 & 0 & 0 &
}
\end{flalign}
\end{defi}

\begin{theo}[{\cite[Theorems 5.11 and 5.14]{BaerBecker}}]~\\
For any differential cohomology theory
 $(\widetilde{H}^\ast(\,\cdot\,;\bbZ),\widetilde{\Curv},\widetilde{\Char},\widetilde{\iota},\widetilde{\kappa})$ 
 there exists a unique natural isomorphism $\Xi : \widetilde{H}^\ast(\,\cdot\,;\bbZ) \Rightarrow \DCo{\ast}{\,\cdot\,}{\bbZ}$
 to differential characters such that
 \begin{flalign}
 \Xi\circ \widetilde{\iota} = \iota \  , \qquad \Xi\circ
 \widetilde{\kappa} = \kappa \  , \qquad \Curv\circ \Xi =
 \widetilde{\Curv} \  , \qquad \Char\circ \Xi = \widetilde{\Char}~.
 \end{flalign}
\end{theo}
\begin{rem}
In order to simplify the notation we shall denote in the following any differential co\-ho\-mo\-lo\-gy
theory by $( \DCo{\ast}{\,\cdot\,}{\bbZ},\Curv,\Char,\iota,\kappa)$.
\end{rem}


\section{\label{sec:MWmaps}Generalized Maxwell maps}

Our main interest lies in understanding the classical and quantum field theory described
by a differential cohomology theory $\DCo{\ast}{\,\cdot\,}{\bbZ}: \mathsf{Man} \to \mathsf{Ab}^\bbZ$.
For a clearer presentation we shall fix $1\leq k\in \bbZ$ and study the differential cohomology groups of degree $k$,
i.e.\ the contravariant functor $\DCo{k}{\,\cdot\,}{\bbZ} : \mathsf{Man} \to \mathsf{Ab}$.
Furthermore, in order to formulate relativistic field equations which generalize Maxwell's equations in degree $k=2$,
we shall restrict the category of smooth manifolds to a suitable category of globally hyperbolic spacetimes.
A natural choice, see e.g.\ \cite{Brunetti:2001dx,Fewster:2011pe,Bar:2011iu,Bar:2007zz},  is the following
\begin{defi}
The category $\mathsf{Loc}^m$ consists of the following objects and morphisms:
\begin{itemize}
\item The objects $M$ in $\mathsf{Loc}^m$ are 
oriented and time-oriented globally hyperbolic Lorentzian manifolds, which are of dimension $m\geq 2$ and
of finite-type.\footnote{
A manifold is of finite-type if it has a finite good cover, i.e.\ a finite cover by contractible 
open subsets such that all (multiple) overlaps are also contractible. This condition is not part
of the original definition in \cite{Brunetti:2001dx,Fewster:2011pe}, however it is very useful for
studying gauge theories as it  makes available Poincar{\'e} duality. 
See also \cite{Benini:2013ita,Benini:2013tra} for similar issues.
} (For ease of notation we shall always suppress the orientation, time-orientation and Lorentzian metric.)
\item The morphisms $f: M\to N$ in $\mathsf{Loc}^m$ are orientation and 
time-orientation preserving isometric embeddings, such that the image $f[M]\subseteq N$ is causally compatible and open.
\end{itemize}
\end{defi}
\begin{rem}
The curvature $\Curv : \DCo{k}{M}{\bbZ}\to \Omega^k_\bbZ(M)$ is only non-trivial if the
degree $k$ is less than or equal to the dimension $m$ of $M$.
Hence when restricting the contravariant functor $\DCo{k}{\,\cdot\,}{\bbZ}$ to the category $\mathsf{Loc}^m$
we shall always assume that $k\leq m$.
\end{rem}

When working on the category $\mathsf{Loc}^m$ we have available a further natural transformation given by the 
{\it codifferential} $\delta : \Omega^\ast(\,\cdot\,)\Rightarrow \Omega^{\ast-1}(\,\cdot\,)$.\footnote{We have denoted the codifferential by the same symbol as the coboundary maps in singular cohomology. It should be
clear from the context to which of these maps the symbol $\delta$ refers to.}
Our conventions for the codifferential $\delta$ are as follows: Denoting by $\ast$ the Hodge operator, we define $\delta$ on $p$-forms by
\begin{flalign}
\delta \,:\, \Omega^p(M) \ \longrightarrow \ \Omega^{p-1}(M) \ , \qquad \omega \ 
\longmapsto \ (-1)^{m\, (p+1)} \ast \dd \ast \omega \ .
\end{flalign}
For any two forms $\omega,\omega^\prime\in \Omega^p(M)$ with compactly
overlapping support we have a natural indefinite inner product defined by
\begin{flalign}\label{eqn:pairing}
\ip{\omega}{\omega^\prime\,} := \int_M \, \omega\wedge \ast \, \omega^\prime .
\end{flalign}
Then the codifferential $\delta$ is the formal adjoint of the differential $\dd$ with respect to this inner product, 
i.e.\ $\ip{\delta \omega}{\omega^\prime\,} = \ip{\omega}{\dd\omega^\prime\,}$ for all $\omega\in \Omega^p(M)$ 
and $\omega^\prime\in\Omega^{p-1}(M)$ with compactly overlapping support.
\begin{defi}
The {\bf (generalized) Maxwell map} is the natural transformation
\begin{flalign}
\mathrm{MW}:=\delta\circ \Curv \,:\, \DCo{k}{\,\cdot\,}{\bbZ} \
\Longrightarrow \ \Omega^{k-1}(\,\cdot\,)~.
\end{flalign} 
For any object $M$ in $\mathsf{Loc}^m$, the {\bf solution subgroup} in $\DCo{k}{M}{\bbZ}$ is defined as
the kernel of the Maxwell map, 
\begin{flalign}\label{eqn:solutionsubgroup}
\widehat{\mathfrak{Sol}}{}^k(M):= \left\{h\in\DCo{k}{M}{\bbZ} \ : \ \mathrm{MW}(h) = \delta\big(\Curv(h)\big)=0 \right\}~.
\end{flalign}
\end{defi}
\begin{lem}
$\widehat{\mathfrak{Sol}}{}^k(\,\cdot\,) : \mathsf{Loc}^m \to \mathsf{Ab}$ is a subfunctor of 
$\DCo{k}{\,\cdot\,}{\bbZ}: \mathsf{Loc}^m \to \mathsf{Ab}$.
\end{lem}
\begin{proof}
Let $M$ be any object in $\mathsf{Loc}^m$. Then clearly $\widehat{\mathfrak{Sol}}{}^k(M)$
is a subgroup of $\DCo{k}{M}{\bbZ}$, since $\mathrm{MW}$ is a homomorphism of Abelian groups.
Let now $f: M\to N$ be any $\mathsf{Loc}^m$-morphism.  We have to show that
$\DCo{k}{f}{\bbZ} : \DCo{k}{N}{\bbZ}\to \DCo{k}{M}{\bbZ}$ restricts to an $\mathsf{Ab}$-morphism
$\widehat{\mathfrak{Sol}}{}^k(N)\to \widehat{\mathfrak{Sol}}{}^k(M)$.
This follows from naturality of $\mathrm{MW}$: for any $h\in \widehat{\mathfrak{Sol}}{}^k(N)$
we have $\mathrm{MW}\big(\DCo{k}{f}{\bbZ}(h)\big) = \Omega^{k-1}(f)\big(\mathrm{MW}(h)\big) = 0$, hence
$\DCo{k}{f}{\bbZ}(h)\in \widehat{\mathfrak{Sol}}{}^k(M)$.
\end{proof}
\begin{rem}
For any $\mathsf{Loc}^m$-morphism $f: M\to N$ we shall denote the
restriction of $\DCo{k}{f}{\bbZ}$ to $\widehat{\mathfrak{Sol}}{}^k(N)$ by
$\widehat{\mathfrak{Sol}}{}^k(f) : \widehat{\mathfrak{Sol}}{}^k(N)\to \widehat{\mathfrak{Sol}}{}^k(M)$.
\end{rem}

The next goal is to restrict the diagram (\ref{eqn:absdiagrm}) to the solution subgroup 
$\widehat{\mathfrak{Sol}}{}^k(M)\subseteq \DCo{k}{M}{\bbZ}$. 
Let us denote by $\Omega^k_{\bbZ,\,\delta}(M)$ the Abelian group of closed and coclosed $k$-forms
with integral periods. From the definition of the solution subgroups (\ref{eqn:solutionsubgroup})
it is clear that the middle horizontal sequence in  (\ref{eqn:absdiagrm}) restricts to the exact sequence
\begin{flalign}
\xymatrix{
0 \ \ar[r] & \ \Co{k-1}{M}{\bbT} \ \ar[r]^{ \ \ \kappa} & \
\widehat{\mathfrak{Sol}}{}^k(M) \ \ar[r]^-{\Curv} & \
\Omega^k_{\bbZ,\,\delta}(M) \ \ar[r] & \ 0 \ .
}
\end{flalign}
In order to restrict the complete diagram  (\ref{eqn:absdiagrm}) to the solution subgroups
we need the following
\begin{lem}
The inverse image of $\widehat{\mathfrak{Sol}}{}^k(M) $ under the topological trivialization $\iota$
is given by
\begin{flalign}
\mathfrak{Sol}^k(M) := \iota^{-1}\big(\,
\widehat{\mathfrak{Sol}}{}^k(M)\, \big) =
\Big\{[\eta]\in\frac{\Omega^{k-1}(M)}{\Omega^{k-1}_\bbZ(M)} \ : \ 
\delta\dd \eta =0\Big\}~.
\end{flalign}
\end{lem}
\begin{proof}
This follows immediately from the commutative square in the upper right corner of the diagram (\ref{eqn:absdiagrm}): the equivalence class $[\eta]\in \Omega^{k-1}(M)/\Omega^{k-1}_\bbZ(M)$ maps under $\iota$ to $\widehat{\mathfrak{Sol}}{}^k(M)$
if and only if $\dd \eta$ is coclosed.
\end{proof}
Denoting by $(\dd\Omega^{k-1})_\delta(M)$ the space of exact $k$-forms which are also coclosed,
we obtain
\begin{theo}\label{theo:soldiagram}
The following diagram commutes and has exact rows and columns:
\begin{flalign}\label{eqn:soldiagrm}
\xymatrix{
 & 0 \ar[d]& 0\ar[d] & 0\ar[d] & \\
 0 \ \ar[r] & \
 \frac{\Co{k-1}{M}{\bbR}}{H^{k-1}_\mathrm{free}(M;\bbZ)} \ \ar[d]
 \ar[r] & \ \mathfrak{Sol}^k(M) \ \ar[d]^-{\iota} \ar[r]^-{\dd} & \
 (\dd \Omega^{k-1})_\delta (M) \ \ar[r]\ar[d] & \ 0\\
 0 \ \ar[r] & \ \Co{k-1}{M}{\bbT} \ \ar[d] \ar[r]^-{\kappa} & \
 \widehat{\mathfrak{Sol}}{}^k(M) \ \ar[d]^-{\Char} \ar[r]^-{\Curv} &
 \ar[d] \ \Omega^k_{\bbZ,\,\delta}(M) \ \ar[r] & \ 0\\
 0 \ \ar[r] & \ H^k_\mathrm{tor}(M;\bbZ) \ \ar[d]\ar[r] & \
 \Co{k}{M}{\bbZ} \ \ar[d] \ar[r] & \ H^k_\mathrm{free}(M;\bbZ) \
 \ar[d]\ar[r] & \ 0\\
 & 0 & 0 & 0 &
}
\end{flalign}
\end{theo}
\begin{proof}
The only nontrivial step is to show that $\Char :  \widehat{\mathfrak{Sol}}{}^k(M) \to \Co{k}{M}{\bbZ}$
is surjective. Let $u\in \Co{k}{M}{\bbZ}$ be any cohomology class. By the middle vertical exact sequence in (\ref{eqn:absdiagrm})
there exists $h\in \DCo{k}{M}{\bbZ}$ such that $\Char(h) = u$. Note that $h$ is not necessarily an
element in $\widehat{\mathfrak{Sol}}{}^k(M)$, i.e.\ in general $0\neq \mathrm{MW}(h)\in\Omega^{k-1}(M)$. 
Let us now take $[\eta]\in \Omega^{k-1}(M)/\Omega^{k-1}_\bbZ(M)$ and note that
the characteristic class of  $h^\prime := h + \iota\big([\eta]\big)\in\DCo{k}{M}{\bbZ}$ 
is again $u$ as $\iota$ maps to the kernel of $\Char$. We now show that $[\eta]$ can be chosen such that
$\mathrm{MW}(h^\prime\, ) =0$, which completes the proof. By posing
$\mathrm{MW}(h^\prime\, ) =0$ as a condition we obtain the partial differential equation 
\begin{flalign}\label{eqn:inhompde}
0= \mathrm{MW}(h) + \mathrm{MW}\big(\iota([\eta])\big) =
\mathrm{MW}(h) + \delta \dd\eta~,
\end{flalign}
where $\eta\in\Omega^{k-1}(M)$ is any representative of the class $[\eta]$. 
As the inhomogeneity $\mathrm{MW}(h)=\delta\big(\Curv(h)\big)$ is
coexact,
there always exists a solution $\eta$ to the equation (\ref{eqn:inhompde}), see e.g.\ \cite[Section 2.3]{Sanders:2012sf}.
\end{proof}
\begin{rem}
In the context of compact {\it Riemannian} manifolds, a result similar to Theorem \ref{theo:soldiagram} 
is proven in \cite{Green:2008ub}. They consider harmonic differential characters 
on a compact Riemannian manifold, i.e.\ differential characters with harmonic curvature forms, 
and prove that these fit into exact sequences similar to the ones in (\ref{eqn:soldiagrm}). However,
the proof  in \cite{Green:2008ub} relies on the theory of elliptic
partial differential equations and therefore differs from our proof of 
Theorem \ref{theo:soldiagram}, which uses the theory of hyperbolic
partial differential equations. In particular,
the results of \cite{Green:2008ub} do not imply our results.
\end{rem}
\sk

We say that a $\mathsf{Loc}^m$-morphism $f: M\to N$ is a {\it Cauchy morphism} if
its image $f[M]$ contains a Cauchy surface of $N$. 
The following statement proves that local solutions
 to the generalized Maxwell equation 
 (or more precisely solutions local in time)
 uniquely extend to global solutions, i.e.\ that $\mathrm{MW}$ imposes a deterministic
 dynamical law on $\DCo{k}{M}{\bbZ}$.
\begin{theo}\label{theo:soltimeslice}
If $f:M\to N$ is a Cauchy morphism, then
 $\widehat{\mathfrak{Sol}}{}^k(f): \widehat{\mathfrak{Sol}}{}^k(N)\to \widehat{\mathfrak{Sol}}{}^k(M)$
 is an $\mathsf{Ab}$-isomorphism.
\end{theo}
\begin{proof}
Let us start with a simple observation:
Any $\mathsf{Loc}^m$-morphism $f:M\to N$ can be factorized as
$f= \iota_{N;f[M]}\circ \underline{f}$,  where
$\underline{f} : M\to f[M]$ is the $\mathsf{Loc}^m$-isomorphism given by
restricting $f$ to its image 
and $\iota_{N;f[M]}: f[M] \to N$ is the $\mathsf{Loc}^m$-morphism given by the 
canonical inclusion of subsets. As functors map isomorphisms to isomorphisms,
it is sufficient to prove that for any object $N$ in $\mathsf{Loc}^m$ and any causally compatible, open and 
globally hyperbolic subset $O\subseteq N$ that contains a Cauchy surface of $N$, the
canonical inclusion $\iota_{N;O} : O\to N$ is mapped to an $\mathsf{Ab}$-isomorphism 
$\widehat{\mathfrak{Sol}}{}^k(\iota_{N;O}) : \widehat{\mathfrak{Sol}}{}^k(N)\to \widehat{\mathfrak{Sol}}{}^k(O)$.
\sk

We first prove injectivity. Let $h\in \widehat{\mathfrak{Sol}}{}^k(N)$ be any element 
in the kernel of $\widehat{\mathfrak{Sol}}{}^k(\iota_{N;O})$. Applying $\Char$ 
implies that $\Char(h)$ lies in the kernel of $\Co{k}{\iota_{N;O}}{\bbZ} : \Co{k}{N}{\bbZ} \to \Co{k}{O}{\bbZ}$,
which is an isomorphism since $O$ and $N$ are both homotopic to their common Cauchy surface.
As a consequence $\Char(h)=0$ and by Theorem \ref{theo:soldiagram} there exists $[\eta]\in\mathfrak{Sol}^k(N)$ such that
$h=\iota\big([\eta]\big)$. Since $\iota$ is natural and injective, this implies that
$[\eta]$ lies in the kernel of  $\mathfrak{Sol}^k(\iota_{N;O})$. We can always choose a coclosed representative
$\eta\in\Omega_\delta^{k-1}(N)$ of the class $[\eta]$ (by going to Lorenz gauge, cf.\ \cite[Section 2.3]{Sanders:2012sf})
and the condition that $[\eta]$ lies in the kernel of $\mathfrak{Sol}^k(\iota_{N;O})$ implies that
the restriction $\eta\vert_O\in\Omega_\delta^{k-1}(O)$ of $\eta$ to $O$ is an element of $\Omega^{k-1}_\bbZ(O)$.
In particular, $\delta \eta\vert_O =0$ and $\dd\eta\vert_O=0$, so by \cite[Theorem 7.8]{Benini:2014vsa} there 
exist forms  $\alpha\in\Omega^{k}_{\mathrm{tc},\,\dd }(O)$ and $\beta\in\Omega^{k-2}_{\mathrm{tc},\,\delta}(O)$
of timelike compact support such that $\eta\vert_O = G(\delta\alpha + \dd \beta)$. 
Here $G:=G^+ - G^-: \Omega^{k-1}_\mathrm{tc}(O)\to\Omega^{k-1}(O)$ is the unique retarded-minus-advanced
Green's operator for the Hodge-d'Alembert operator 
$\square := \delta \circ \dd + \dd \circ \delta : \Omega^{k-1}(O)\to\Omega^{k-1}(O)$, see \cite{Bar:2007zz,Baernew}.
Since $O\subseteq N$ contains a Cauchy
surface of $N$, any form of timelike compact support on $O$ can be extended by zero 
to a form of timelike compact support on $N$ (denoted with a slight abuse of notation by the same symbol).
Hence there exists $\rho\in\Omega^{k-1}(N)$ satisfying $\rho\vert_O =0$ such that
$\eta = G(\delta\alpha + \dd\beta) +\rho$ on all of $N$.
As $\eta$ satisfies $\delta \dd \eta=0$ and the Lorenz gauge condition $\delta\eta=0$, it also satisfies
$\square\eta=0$.  Since also $\square G(\delta\alpha + \dd\beta) =0$, we obtain $\square\rho =0$, which together
with the support condition $\rho\vert_O=0$ implies $\rho=0$. So $\eta = G(\delta\alpha + \dd\beta) $ on all of $N$
and it remains to prove that $\eta\in\Omega^{k-1}_\bbZ(N)$. As $\eta$ is obviously closed, 
the integral $\int_z\, \eta$  depends only on the homology class $[z]\in \Ho{k-1}{N}{\bbZ}$.
The inclusion $\iota_{N;O}:O \to N$ induces isomorphisms $\Ho{k-1}{\iota_{N;O}}{\bbZ}:\Ho{k-1}{O}{\bbZ}\to \Ho{k-1}{N}{\bbZ}$ 
and $\Co{k-1}{\iota_{N;O}}{\bbZ}:\Co{k-1}{N}{\bbZ}\to \Co{k-1}{O}{\bbZ}$.
Thus any $z \in Z_{k-1}(N;\bbZ)$ is homologous to a cycle of the form ${\iota_{N;O}}_\ast(z^\prime\, )$ for some $z^\prime \in Z_{k-1}(O;\bbZ)$,
where ${\iota_{N;O}}_\ast$ is the push-forward.
This yields $\int_z\, \eta = \int_{{\iota_{N;O}}_\ast (z^\prime\, )} \, \eta = \int_{z^\prime} \, \iota_{N;O}^\ast (\eta) = \int_{z^\prime} \,
\eta\vert_O \in \bbZ$, since $\eta\vert_O\in\Omega^{k-1}_{\bbZ}(O)$.
Thus $\eta\in \Omega^{k-1}_\bbZ(N)$ and hence $h = \iota\big([\eta]\big)=\ozero$.
\sk

We now prove surjectivity. Let $l\in\widehat{\mathfrak{Sol}}{}^k(O)$ be arbitrary and consider its characteristic class
$\Char(l)\in \Co{k}{O}{\bbZ}$. As we have explained above, $\Co{k}{\iota_{N;O}}{\bbZ} : \Co{k}{N}{\bbZ} \to \Co{k}{O}{\bbZ}$
is an isomorphism, hence by using also Theorem \ref{theo:soldiagram} we can find $h\in\widehat{\mathfrak{Sol}}{}^k(N)$
such that $\Char(l)= \Co{k}{\iota_{N;O}}{\bbZ} \big(\Char(h)\big)  =
\Char\big(\, \widehat{\mathfrak{Sol}}{}^k(\iota_{N;O})(h)\, \big)$.
Again by Theorem \ref{theo:soldiagram}, there exists $[\eta]\in \mathfrak{Sol}^k(O)$ such that
$l =\widehat{\mathfrak{Sol}}{}^k(\iota_{N;O})(h) + \iota\big([\eta]\big)$. By \cite[Theorem 7.4]{Benini:2014vsa} 
the equivalence class $[\eta]$ has a representative $\eta \in \Omega^{k-1}(O)$ which is of the form $\eta = G(\alpha)$ for some
$\alpha\in \Omega^{k-1}_{\mathrm{tc},\,\delta}(O)$. We can extend $\alpha$ by zero (denoted with a slight abuse of notation
 by the same symbol) and define $[G(\alpha)]\in \mathfrak{Sol}^k(N)$. Since 
 $[\eta] = \mathfrak{Sol}^k(\iota_{N;O})\big([G(\alpha)]\big)$ we have
 $l = \widehat{\mathfrak{Sol}}{}^k(\iota_{N;O})(h) + \iota\big(\mathfrak{Sol}^k(\iota_{N;O})\big([G(\alpha)]\big)\big)
 = \widehat{\mathfrak{Sol}}{}^k(\iota_{N;O})\big(h +
 \iota\big([G(\alpha)]\big)\big)$, which gives the assertion.
\end{proof} 


\section{\label{sec:duality}Smooth Pontryagin duality}

Let $\big(\DCo{\ast}{\,\cdot\,}{\bbZ},\Curv,\Char,\iota,\kappa\big)$ be a differential cohomology theory
and let us consider its restriction $\DCo{k}{\,\cdot\,}{\bbZ} : \mathsf{Loc}^m \to\mathsf{Ab}$ 
to degree $k\geq 1$ and to the category $\mathsf{Loc}^m$ with $m\geq k$.
For an Abelian group $G$, the {\it character group} is defined by $G^\star :=\Hom(G,\bbT)$,
where $\Hom$ denotes the homomorphisms of Abelian groups.
Since the circle group $\bbT$ is divisible, the $\Hom$-functor $\Hom(\,\cdot\, , \bbT)$ preserves exact sequences.
Hence we can dualize the degree $k$ component of the diagram (\ref{eqn:absdiagrm}) and
obtain the following commutative diagram with exact rows and columns:
\begin{flalign}\label{eqn:absdiagrmdual}
\xymatrix{
 & 0 \ar[d]& 0\ar[d] & 0\ar[d] & \\
 0 \ \ar[r] & \ H^k_\mathrm{free}(M;\bbZ)^\star \ \ar[d] \ar[r]& \
 \Co{k}{M}{\bbZ}^\star \ \ar[r] \ar[d]^-{\Char^\star}& \ar[d] \
 H^k_\mathrm{tor}(M;\bbZ)^\star \ \ar[r]& \ 0\\
 0 \ \ar[r] & \ar[d] \ \Omega^k_\bbZ(M)^\star \ \ar[r]^-{\Curv^\star}
 &\ar[d]^-{\iota^\star} \ \DCo{k}{M}{\bbZ}^\star \
 \ar[r]^-{\kappa^\star} & \ar[d] \ \Co{k-1}{M}{\bbT}^\star \ \ar[r]& \
 0\\
 0 \ \ar[r] & \ar[d] \ \left(\dd\Omega^{k-1}(M)\right)^\star \
 \ar[r]^-{\dd^\star} & \ar[d] \ \Big(\,
 \frac{\Omega^{k-1}(M)}{\Omega^{k-1}_\bbZ(M)}\, \Big)^\star \ 
 \ar[r] &\ar[d] \ \Big(\,
 \frac{\Co{k-1}{M}{\bbR}}{H^{k-1}_\mathrm{free}(M;\bbZ)}\, \Big)^\star
 \ \ar[r] & \ 0\\
 & 0 & 0 & 0 &
}
\end{flalign}

The diagram (\ref{eqn:absdiagrmdual}) contains the character groups of 
$\DCo{k}{M}{\bbZ}$ and of various groups of differential forms,
whose generic elements are too singular for our purposes.
We shall use a strategy similar to \cite{HLZ} (called {\it smooth Pontryagin duality}) in order to identify suitable subgroups of 
such character groups, which describe regular group characters.
In order to explain the construction of the smooth Pontryagin duals of the Abelian group $\DCo{k}{M}{\bbZ}$ 
and the various groups of differential forms, let us first notice that there exists an injective 
homomorphism of Abelian groups $\mathcal{W}: \Omega^{p}_0(M) \to \Omega^p(M)^\star\,,~\varphi\mapsto \mathcal{W}_\varphi$, from the space of compactly supported $p$-forms 
$\Omega^{p}_0(M)$ to the character group of the $p$-forms $\Omega^p(M)$.
For any $\varphi \in \Omega^{p}_0(M)$, the character $\mathcal{W}_\varphi$ is defined as
\begin{flalign}\label{eqn:calWmap}
\mathcal{W}_\varphi \,:\, \Omega^p(M) \ \longrightarrow \ \bbT~,
\qquad \omega \ \longmapsto \  
 \exp\big(2\pi\ii\ip{\varphi}{\omega}\big) = \exp\Big(2\pi\ii\int_M
\, \varphi\wedge \ast\,\omega\Big)~.
\end{flalign}
With this homomorphism we can regard $\Omega^{p}_0(M)$ as a subgroup of $\Omega^p(M)^\star$
and we shall simply write $\Omega^{p}_0(M)\subseteq \Omega^p(M)^\star$, suppressing the map $\mathcal{W}$
when there is no risk of confusion. We say that $\Omega^p(M)^\star_\infty := \Omega^{p}_0(M)\subseteq \Omega^p(M)^\star$
 is the {\it smooth Pontryagin dual} of the $p$-forms. It is important to notice that
 the smooth Pontryagin dual separates points of $\Omega^p(M)$ since $\mathcal{W}_{\varphi}(\omega) =1$ for all $\varphi\in\Omega^p_0(M)$
 if and only if $\omega=0$.
 \sk
 
We next come to the smooth Pontryagin dual of the Abelian group 
$\Omega^{k-1}(M)/\Omega^{k-1}_\bbZ(M)$. 
A smooth group character $\varphi\in \Omega^{k-1}_0(M) = \Omega^{k-1}(M)^\star_\infty$
induces to this quotient if and only if $\mathcal{W}_\varphi(\omega) =1$ for all
$\omega\in \Omega^{k-1}_\bbZ(M)$. Hence we have to understand the vanishing subgroups
of differential forms with integral periods,
\begin{flalign}\label{eqn:integralvanishing}
\mathcal{V}^p(M):= \left\{\varphi \in \Omega^{p}_0(M) \ : \
  \mathcal{W}_\varphi(\omega) =1 \quad \forall \omega\in \Omega^p_\bbZ(M)\right\}~.
\end{flalign}
To give an explicit characterization of the subgroups $\mathcal{V}^p(M)$ we require some prerequisites:
Any coclosed and compactly supported $p$-form $\varphi\in \Omega^p_{0,\,\delta}(M)$
defines via the pairing (\ref{eqn:pairing}) a linear map $\ip{\varphi }{\,\cdot\,} : \Omega_\dd^p(M)\to\bbR\,,~\omega \mapsto
\ip{\varphi}{\omega} = \ip{[\varphi]}{[\omega]}$, which depends only on the de~Rham class
 $[\omega]\in \Omega^p_{\dd}(M)/\dd\Omega^{p-1}(M)$ of $\omega$ 
and the compactly supported dual de Rham class $[\varphi]\in \Omega^p_{0,\,\delta}(M)/\delta\Omega^{p+1}_0(M)$ of $\varphi$.
By Poincar{\'e} duality and de Rham's theorem, we can naturally identify the compactly supported dual de 
Rham cohomology group $\Omega^p_{0,\,\delta}(M)/\delta\Omega^{p+1}_0(M)$ with the dual vector space
$\Hom_\bbR(H^p(M;\bbR),\bbR)$. As $H^p_\mathrm{free}(M;\bbZ)$ is a lattice in $H^p(M;\bbR)$,
we have the subgroup $H^p_\mathrm{free}(M;\bbZ)^\prime := \Hom(H^p_\mathrm{free}(M;\bbZ),\bbZ) 
$ of $ \Hom_\bbR(H^p(M;\bbR),\bbR)$. (By $^\prime$ we denote the dual $\bbZ$-module.)
We then write $[\varphi]\in H^p_\mathrm{free}(M;\bbZ)^\prime $, for some $\varphi\in \Omega^p_{0,\,\delta}(M)$,
if and only if  $\ip{[\varphi]}{\,\cdot\,}$ restricts (under the isomorphisms above) 
to a homomorphism  of Abelian groups $H^p_\mathrm{free}(M;\bbZ) \to \bbZ$.

\begin{lem}\label{lem:vanishinggroup}
$\mathcal{V}^p(M) = \big\{\varphi\in \Omega^{p}_{0,\,\delta}(M) \ : \ [\varphi] \in H^p_\mathrm{free}(M;\bbZ)^\prime  \big\}$.
\end{lem}
\begin{proof}
We first show the inclusion ``$\supseteq$'':
Let $\varphi\in \Omega^{p}_0(M)$ be coclosed, i.e.\ $\delta\varphi=0$, and such that 
$\ip{[\varphi]}{\,\cdot\,}$ restricts to a homomorphism  of 
Abelian groups $H^p_\mathrm{free}(M;\bbZ) \to \bbZ$.
 For any $\omega\in \Omega^p_\bbZ(M)$ we have
 \begin{flalign}
 \mathcal{W}_\varphi(\omega) = \exp\big(2\pi\ii\ip{\varphi}{\omega}\big) = \exp\big(2\pi\ii\ip{[\varphi]}{[\omega]}\big) =1~,
 \end{flalign}
 where in the second equality we have used the fact that the pairing depends only on the equivalence classes
 and in the last equality we have used $[\omega]\in H_\mathrm{free}^p(M;\bbZ)$
 via the de Rham isomorphism.
 \sk
 
Let us now show the inclusion ``$\subseteq$'': Let $\varphi\in \mathcal{V}^p(M)$.
As $\dd\Omega^{p-1}(M)\subseteq \Omega^p_\bbZ(M)$ we obtain $\mathcal{W}_\varphi(\dd\eta) = \exp\left(2\pi\ii \ip{\varphi}{\dd\eta} \right) = \exp \left(2\pi\ii\ip{\delta\varphi}{\eta}\right) =1$ for all $\eta\in\Omega^{p-1}(M)$,
which implies that $\delta\varphi=0$ and hence $\varphi\in \Omega^{p}_{0,\, \delta} (M)$. For any 
$\omega\in\Omega^p_\bbZ(M)$ we obtain
$\mathcal{W}_\varphi(\omega) = \exp\left(2\pi\ii
  \ip{[\varphi]}{[\omega]}\right)=1$, and hence 
$[\varphi]\in H^p_\mathrm{free}(M;\bbZ)^\prime$.
\end{proof}

Motivated by the definition (\ref{eqn:integralvanishing}) we define the smooth Pontryagin dual of the
 quotient group $\Omega^{k-1}(M)/\Omega^{k-1}_\bbZ(M)$ by
\begin{flalign}
\Big(\, \frac{\Omega^{k-1}(M)}{\Omega^{k-1}_\bbZ(M)}\, \Big)^\star_\infty := \mathcal{V}^{k-1}(M)~.
\end{flalign}
\begin{lem}\label{lem:separatingtoptriv}
The smooth Pontryagin dual $\mathcal{V}^{k-1}(M)$ 
separates points of $\Omega^{k-1}(M)/\Omega^{k-1}_\bbZ(M)$.
\end{lem}
\begin{proof}
Let $\eta\in\Omega^{k-1}(M)$ be such that $\mathcal{W}_\varphi( \eta ) = \exp\big(2\pi\ii\ip{\varphi}{\eta}\big) =1$ for all $\varphi\in \mathcal{V}^{k-1}(M) $.
We need to prove that $\eta$ is closed and has integral periods, which implies that $\mathcal{V}^{k-1}(M)$
separates points of the quotient $\Omega^{k-1}(M)/\Omega^{k-1}_\bbZ(M)$.
Since by Lemma \ref{lem:vanishinggroup} 
we have $\delta\Omega^k_0(M)\subseteq \mathcal{V}^{k-1}(M)$, we obtain in particular the
condition $1= \exp\big(2\pi\ii\ip{\delta\zeta}{\eta}\big) = \exp\big(2\pi\ii\ip{\zeta}{\dd\eta}\big)$ for all
$\zeta\in\Omega^k_0(M)$, which implies $\dd\eta=0$. For any $\varphi\in \mathcal{V}^{k-1}(M)$ we then
get the condition
\begin{flalign}
1=\exp\big(2\pi\ii\ip{\varphi}{\eta}\big) = \exp\big(2\pi\ii\ip{[\varphi]}{[\eta]}\big)
\end{flalign}
for all $[\varphi]\in H^{k-1}_\mathrm{free}(M;\bbZ)^\prime$, which implies that 
the de Rham class $[\eta]$ defines an element in the double dual $\bbZ$-module $H^{k-1}_\mathrm{free}(M;\bbZ)^{\prime\prime}$
of $H^{k-1}_\mathrm{free}(M;\bbZ)$.
As $H^{k-1}_\mathrm{free}(M;\bbZ)$ is finitely generated (by our assumption that $M$ is of finite-type) 
and free, its double dual $\bbZ$-module is isomorphic to itself,
hence the class $[\eta]$ defines an element in $H^{k-1}_\mathrm{free}(M;\bbZ)$
and as a consequence $\eta$ has integral periods.
\end{proof}

Using further the natural isomorphism (see e.g.\ \cite[Lemma 5.1]{HLZ})
\begin{flalign}
\Big(\, \frac{\Co{k-1}{M}{\bbR}}{H^{k-1}_\mathrm{free}(M;\bbZ) }\,
\Big)^\star \simeq  \Hom\big(H^{k-1}_\mathrm{free}(M;\bbZ),\bbZ \big)= H^{k-1}_\mathrm{free}(M;\bbZ)^\prime~,
\end{flalign}
we observe that the restriction of the lowest row of the diagram (\ref{eqn:absdiagrmdual}) 
to smooth Pontryagin duals reads as
\begin{flalign}\label{eqn:lowersmooth}
\xymatrix{
0 \ \ar[r] & \ \delta\Omega^{k}_0(M) \ \ar[r] & \ \mathcal{V}^{k-1}(M)
\ \ar[r] & \ H^{k-1}_\mathrm{free}(M;\bbZ)^\prime \ \ar[r] & \ 0
}
\end{flalign}
with the dual group homomorphisms 
\begin{subequations}
\begin{flalign}
\delta\Omega^{k}_0(M) & \ \longrightarrow \ \mathcal{V}^{k-1}(M)  ~, \qquad
\delta\eta \ \longmapsto \ \delta\eta~,\\[4pt]
\mathcal{V}^{k-1}(M) & \ \longrightarrow \
H^{k-1}_\mathrm{free}(M;\bbZ)^\prime ~, \qquad \varphi \ \longmapsto \
[\varphi]~.
\end{flalign}
\end{subequations}
Here we have implicitly used the injective homomorphism of Abelian groups 
$\mathcal{W} : \delta\Omega^{k}_0(M)\to\big(\dd\Omega^{k-1}(M)\big)^\star$, $\delta\zeta \mapsto \mathcal{W}_{\delta\zeta}$,
defined by
\begin{flalign}
\mathcal{W}_{\delta\zeta} \,:\, \dd\Omega^{k-1}(M) \ \longrightarrow \
\bbT~, \qquad \dd \eta \ \longmapsto \ \exp\big( 2\pi\ii \ip{\zeta}{\dd\eta}\big)~.
\end{flalign}
We suppress this group homomorphism and call $\big(\dd\Omega^{k-1}(M)\big)^\star_\infty:= \delta\Omega^{k}_0(M)\subseteq 
\big(\dd\Omega^{k-1}(M)\big)^\star$ the smooth Pontryagin dual of $\dd\Omega^{k-1}(M)$.
The smooth Pontryagin dual $\delta\Omega^{k}_0(M)$
separates points of $\dd\Omega^{k-1}(M)$: if $\exp(2\pi\ii \ip{\zeta}{\dd\eta}) =1$ for all $\zeta\in \Omega^{k}_0(M)$, then $\dd\eta=0$.
Exactness of the sequence (\ref{eqn:lowersmooth}) is an easy check. 
\sk

We now define the smooth Pontryagin dual of the differential cohomology group $\DCo{k}{M}{\bbZ}$ 
by the inverse image
\begin{flalign}\label{eqn:smoothdualdc}
\DCo{k}{M}{\bbZ}^\star_\infty := {\iota^{\star}}^{-1}\big(\mathcal{V}^{k-1}(M)\big)~.
\end{flalign}
Furthermore, by (\ref{eqn:integralvanishing})  it is natural to set 
$\big(\Omega^k_\bbZ(M)\big)^\star_\infty := \Omega^k_0(M)/\mathcal{V}^k(M)$, as in this way we divide out from the smooth
 group characters on $\Omega^k(M)$ exactly those which are trivial on $\Omega^k_\bbZ(M)$.
 The diagram (\ref{eqn:absdiagrmdual}) restricts as follows to the smooth Pontryagin duals.
\begin{theo}\label{theo:smoothdual}
The following diagram commutes and has exact rows and columns:
\begin{flalign}\label{eqn:absdiagrmdualsmooth}
\xymatrix{
 & 0 \ar[d]& 0\ar[d] & 0\ar[d] & \\
 0 \ \ar[r] & \ H^k_\mathrm{free}(M;\bbZ)^\star \ \ar[d] \ar[r]& \
 \Co{k}{M}{\bbZ}^\star \ \ar[r] \ar[d]^-{\Char^\star}& \ar[d] \
 H^k_\mathrm{tor}(M;\bbZ)^\star \ \ar[r]& \ 0\\
 0 \ \ar[r] & \ar[d]^\delta \frac{\Omega^{k}_0(M)}{\mathcal{V}^{k}(M)}
 \ \ar[r]^-{\Curv^\star} &\ar[d]^-{\iota^\star} \
 \DCo{k}{M}{\bbZ}_\infty^\star \ \ar[r]^-{\kappa^\star} & \ar[d] \
 \Co{k-1}{M}{\bbT}^\star \ \ar[r]& \ 0 \\
 0 \ \ar[r] & \ar[d] \ \delta\Omega^{k}_0(M) \ \ar[r] & \ar[d]  \
 \mathcal{V}^{k-1}(M) \ 
 \ar[r] &\ar[d] \ H^{k-1}_\mathrm{free}(M;\bbZ)^\prime \ \ar[r] & \ 0\\
 & 0 & 0 & 0 &
}
\end{flalign}
\end{theo}
\begin{proof}
By the constructions above and (\ref{eqn:absdiagrmdual}), we have the following commutative subdiagram with exact rows and columns:
\begin{flalign}\label{eqn:partialdualdiagram}
\xymatrix{
 & & 0\ar[d] & 0\ar[d] & \\
 0 \ \ar[r] &  \ H^k_\mathrm{free}(M;\bbZ)^\star \ \ar[r]& \
 \Co{k}{M}{\bbZ}^\star \ \ar[r] \ar[d]^-{\Char^\star}& \ar[d] \
 H^k_\mathrm{tor}(M;\bbZ)^\star \ \ar[r]& \ 0\\
  &  &\ar[d]^-{\iota^\star} \ \DCo{k}{M}{\bbZ}_\infty^\star \
  \ar[r]^-{\kappa^\star} & \ar[d] \ \Co{k-1}{M}{\bbT}^\star &\\
0 \ \ar[r] & \ \delta\Omega^{k}_0(M) \ \ar[r] & \ar[d] \
\mathcal{V}^{k-1}(M) \ 
 \ar[r] &\ar[d] \ H^{k-1}_\mathrm{free}(M;\bbZ)^\prime \ \ar[r] & \ 0\\
 & & 0 & 0 &
}
\end{flalign}
and it remains to prove that it extends to the diagram of exact sequences in (\ref{eqn:absdiagrmdualsmooth}). 
\sk

Let us first focus on the left column in  (\ref{eqn:absdiagrmdualsmooth}). By  (\ref{eqn:absdiagrmdual}), there exists
an injective  group homomorphism $H^k_\mathrm{free}(M;\bbZ)^\star \to \Omega^k_\bbZ(M)^\star $ and we
have to show that its image lies in the smooth Pontryagin dual $\Omega^k_0(M)/\mathcal{V}^k(M)$ of $\Omega^k_\bbZ(M)$.
The character group $H^k_\mathrm{free}(M;\bbZ)^\star$ is isomorphic to the quotient
$\Hom_\bbR(H^k(M;\bbR),\bbR)/H^k_\mathrm{free}(M;\bbZ)^\prime$.\footnote{
The isomorphism $\Hom_\bbR(H^k(M;\bbR),\bbR)/H^k_\mathrm{free}(M;\bbZ)^\prime \to H^k_\mathrm{free}(M;\bbZ)^\star$
is constructed as follows:  Given any $\bbR$-linear map $\psi : H^k(M;\bbR) \to \bbR$, we define a group character on
$H^k_\mathrm{free}(M;\bbZ)\subseteq H^k(M;\bbR)$ by $\exp(2\pi\ii\psi(\,\cdot\,))$. This association is
surjective since, as $H^k_\mathrm{free}(M;\bbZ)$ is a free Abelian group, 
any character $\phi : H^k_\mathrm{free}(M;\bbZ)\to \bbT$ has a real lift $\widetilde{\phi}:  
H^k_\mathrm{free}(M;\bbZ)\to \bbR$, i.e.\ $\phi(\,\cdot\,) = \exp(2\pi\ii \widetilde{\phi}(\,\cdot\,)) $,
which further has an $\bbR$-linear extension to $H^k(M;\bbR)$. The kernel of this association
is exactly $H^k_\mathrm{free}(M;\bbZ)^\prime = \Hom(H^k_\mathrm{free}(M;\bbZ),\bbZ)$. 
}
Under this identification, the group homomorphism $H^k_\mathrm{free}(M;\bbZ)^\star \to \Omega^k_\bbZ(M)^\star $
maps $\psi\in \Hom_\bbR(H^k(M;\bbR),\bbR)$ to the group character on
$\Omega^k_\bbZ(M)$ given by
\begin{flalign}
\Omega^k_\bbZ(M) \ \longrightarrow \ \bbT~, \qquad \omega \
\longmapsto \ \exp\big(2\pi\ii \psi([\omega])\big) 
= \exp\big(2\pi\ii\ip{\varphi_\psi}{\omega}\big)~.
\end{flalign}
In the last equality we have used the fact that, by Poincar{\'e} duality and de Rham's theorem,
there exists $\varphi_\psi\in \Omega^k_{0,\,\delta}(M)$ such that $\psi([\omega]) = \ip{\varphi_\psi}{\omega}$
for all $\omega\in\Omega^k_{\dd}(M)$. Hence the image of $H^k_\mathrm{free}(M;\bbZ)^\star \to \Omega^k_\bbZ(M)^\star $
lies in the smooth Pontryagin dual $\Omega^k_0(M)/\mathcal{V}^k(M)$ of $\Omega^k_\bbZ(M)$. Exactness of the corresponding
sequence (the left column in (\ref{eqn:absdiagrmdualsmooth})) is an easy check.
\sk

It remains to understand the middle horizontal sequence in (\ref{eqn:absdiagrmdualsmooth}). From the commutative square in the
lower left corner of (\ref{eqn:absdiagrmdual}) and the definition (\ref{eqn:smoothdualdc}), we find
that $\Curv^\star: \Omega^{k}_\bbZ(M)^\star \to \DCo{k}{M}{\bbZ}^\star$
restricts to the smooth Pontryagin duals: by commutativity of this square, 
$\iota^\star\circ \Curv^\star $ maps the smooth Pontryagin dual $\Omega^k_0(M)/\mathcal{V}^k(M)$
of $\Omega^k_\bbZ(M)$ into the smooth Pontryagin dual $\mathcal{V}^{k-1}(M)$ of $\Omega^{k-1}(M)/\Omega^{k-1}_\bbZ(M)$,
thus $\Curv^\star$ maps $\Omega^k_0(M)/\mathcal{V}^k(M)$ to
$\DCo{k}{M}{\bbZ}^\star_\infty$ by the definition (\ref{eqn:smoothdualdc}). 
 We therefore get the middle horizontal sequence in
(\ref{eqn:absdiagrmdualsmooth}) and it remains to prove that it is exact everywhere. As the restriction of an injective group homomorphism,
$\Curv^\star : \Omega^k_0(M)/\mathcal{V}^k(M) \to \DCo{k}{M}{\bbZ}^\star_\infty$ is injective. 
Next, we prove exactness of the middle part of this sequence by 
using what we already know about the diagram (\ref{eqn:absdiagrmdualsmooth}):
Let $\mathrm{w} \in \DCo{k}{M}{\bbZ}^\star_\infty$ be such that $\kappa^\star(\mathrm{w})=0$. As a consequence of the commutative square
in the lower right corner  and exactness of the lower horizontal sequence in this diagram,
there exists $\varphi\in \Omega^k_0(M)$ such that $\iota^\star(\mathrm{w}) = \delta\varphi$. We can use
$\varphi$ to define an element $[\varphi]\in \Omega^k_0(M)/\mathcal{V}^k(M)$. By the commutative square in the lower left corner 
we have $\iota^\star\big(\mathrm{w} - \Curv^\star([\varphi])\big)=0$, 
hence by exactness of the middle vertical sequence there exists $\phi\in H^k(M;\bbZ)^\star$ such that
$\mathrm{w} = \Curv^\star([\varphi]) + \Char^\star(\phi)$. Applying $\kappa^\star$ yields
$0=\kappa^\star\big(\Char^\star(\phi)\big)$, which by the commutative square in the upper right corner 
and exactness of the right vertical and upper horizontal sequences implies that $\phi$ 
has a preimage $\tilde\phi\in H^k_\mathrm{free}(M;\bbZ)^\star$.
Finally, the commutative square in the upper left corner implies that $\Char^\star(\phi)$
is in fact in the image of $\Curv^\star$ (restricted to the smooth
Pontryagin dual), and hence so is
$\mathrm{w}$. It remains to prove that $\kappa^\star: \DCo{k}{M}{\bbZ}^\star_\infty \to H^{k-1}(M;\bbT)^\star$ is surjective,
which follows from a similar argument based on what we already know about the diagram (\ref{eqn:absdiagrmdualsmooth}):
Let $\phi\in H^{k-1}(M;\bbT)^\star$ and consider its image $\tilde\phi\in H^{k-1}_\mathrm{free}(M;\bbZ)^\prime$
under the group homomorphism in the right column in this diagram. Since $\mathcal{V}^{k-1}(M)\to H^{k-1}_\mathrm{free}(M;\bbZ)^\prime$
and $\iota^\star:\DCo{k}{M}{\bbZ}^\star_\infty\to\mathcal{V}^{k-1}(M)$ are surjective, there exists
 $\mathrm{w}\in\DCo{k}{M}{\bbZ}^\star_\infty$ which maps under the composition of these morphisms to $\tilde\phi$.
 Hence by the commutative square in the lower right corner we have $\widetilde{\kappa^\star(\mathrm{w})} -\tilde{\phi}=0$,
 which by exactness of the right vertical sequence implies that there exists $\psi\in H^k_\mathrm{tor}(M;\bbZ)^\star$
 such that $\phi = \kappa^\star(\mathrm{w}) + \hat\psi$, where $\hat \psi \in H^{k-1}(M;\bbT)^\star$ is the image of $\psi$
 under the group homomorphism $H^k_\mathrm{tor}(M;\bbZ)^\star\to  H^{k-1}(M;\bbT)^\star$.
  By exactness of the upper horizontal sequence,
 $\psi$ has a preimage  $\underline{\psi}\in H^k(M;\bbZ)^\star$, and by the commutative square in the upper right corner 
 we get $\hat\psi = \kappa^\star \big(\Char^\star(\, \underline{\psi}\, )\big)$. This proves surjectivity since
 $\phi = \kappa^\star(\mathrm{w}^\prime\, )$ with 
 $\mathrm{w}^\prime =\mathrm{w} + \Char^\star(\, \underline{\psi}\, )\in \DCo{k}{M}{\bbZ}^\star_\infty$.
\end{proof}

It remains to study two important points: Firstly, we may ask whether the association of the Abelian groups $\DCo{k}{M}{\bbZ}_\infty^\star$
to objects in $\mathsf{Loc}^m$ is functorial and, secondly, we still have to prove that $\DCo{k}{M}{\bbZ}_\infty^\star$
separates points of $\DCo{k}{M}{\bbZ}$.
Let us start with the second point:
\begin{propo}
The smooth Pontryagin dual $\DCo{k}{M}{\bbZ}_\infty^\star$ separates points of $\DCo{k}{M}{\bbZ}$.
\end{propo}
\begin{proof}
Let $h\in\DCo{k}{M}{\bbZ}$ be such that $\mathrm{w}(h) =1$ for all $\mathrm{w}\in\DCo{k}{M}{\bbZ}^\star_\infty$.
Due to the group homomorphism
$\Char^\star : \Co{k}{M}{\bbZ}^\star \to
\DCo{k}{M}{\bbZ}^\star_\infty$ we have in particular
\begin{flalign}
1=\Char^\star(\phi)(h) = \phi\big(\Char(h)\big)
\end{flalign}
for all $\phi\in \Co{k}{M}{\bbZ}^\star$. As the character group $\Co{k}{M}{\bbZ}^\star$ separates points of $\Co{k}{M}{\bbZ}$ we obtain $\Char(h)=0$
and hence by (\ref{eqn:absdiagrm}) there exists $[\eta]\in \Omega^{k-1}(M)/\Omega^{k-1}_\bbZ(M)$
such that $h=\iota\big([\eta]\big)$. The original condition $\mathrm{w}(h) =1$ for all $\mathrm{w}\in\DCo{k}{M}{\bbZ}^\star_\infty$
now reduces to 
\begin{flalign}
1= \mathrm{w}\big(\iota\big([\eta]\big)\big)= \exp\big(2\pi\ii \ip{\iota^\star(\mathrm{w})}{[\eta]}\big)
\end{flalign} 
for all $\mathrm{w}\in\DCo{k}{M}{\bbZ}^\star_\infty$. Using (\ref{eqn:absdiagrmdualsmooth}) the homomorphism
$\iota^\star :\DCo{k}{M}{\bbZ}^\star_\infty\to  \mathcal{V}^{k-1}(M)$ is surjective, and hence by Lemma 
\ref{lem:separatingtoptriv} we find $[\eta]=0$. As a consequence, $h=\iota(0) =\ozero$ and the result follows.
\end{proof}

We shall now  address functoriality: First, recall that the compactly supported $p$-forms are given by a covariant
functor $\Omega^p_0(\,\cdot\,): \mathsf{Loc}^m \to \mathsf{Ab}$. (In fact, this functor maps to the category of real vector spaces.
We shall however forget the multiplication by scalars and only consider the Abelian group structure given by $+$
 on compactly supported $p$-forms.) Explicitly, to any object $M$ in $\mathsf{Loc}^m$ the functor associates the Abelian group
 $\Omega^p_0(M)$ and to any $\mathsf{Loc}^m$-morphism $f:M\to N$ the functor associates
 the $\mathsf{Ab}$-morphism given by the push-forward (i.e.\ extension by zero) $\Omega^p_0(f) :=f_\ast :\Omega^p_0(M)
 \to \Omega^p_0(N)$. Notice that $\mathcal{V}^p(\,\cdot\,) : \mathsf{Loc}^m\to\mathsf{Ab}$ is a subfunctor of $\Omega_0^p(\,\cdot\,)$: using the definition (\ref{eqn:integralvanishing}) we find
 \begin{flalign}
 \mathcal{W}_{f_\ast(\varphi)}(\omega) = \exp\big(2\pi\ii \ip{f_\ast(\varphi)}{\omega}\big) = 
 \exp\big(2\pi\ii \ip{\varphi}{f^\ast(\omega)}\big)=1
 \end{flalign}
for any $\mathsf{Loc}^m$-morphism
 $f:M\to N$, and for all $\varphi\in \mathcal{V}^p(M)$ and
 $\omega\in\Omega^p_\bbZ(N)$, where $f^\ast$ denotes the pull-back of differential forms. In the last equality we have used the fact that closed $p$-forms
with integral periods on $N$ are pulled-back under $f$ to such forms on $M$.
Thus in the diagram (\ref{eqn:absdiagrmdualsmooth}) we can regard
$\Omega^k_0(\,\cdot\,)/\mathcal{V}^k(\,\cdot\,)$, $\delta\Omega^k_0(\,\cdot\,)$ and $\mathcal{V}^{k-1}(\,\cdot\,)$
as covariant functors from $\mathsf{Loc}^m$ to $\mathsf{Ab}$. 
Furthermore, as a consequence
of being the character groups (or dual $\bbZ$-modules) 
of Abelian groups given by {\it contravariant} functors from $\mathsf{Loc}^m$ to $\mathsf{Ab}$,
we can also regard $H^k_\mathrm{free}(\,\cdot\,;\bbZ)^\star$, $\Co{k}{\,\cdot\,}{\bbZ}^\star$, 
$H^k_\mathrm{tor}(\,\cdot\,;\bbZ)^\star$, $\Co{k-1}{\,\cdot\,}{\bbT}^\star$  and 
$H^{k-1}_\mathrm{free}(\,\cdot\,;\bbZ)^\prime$ as covariant functors from $\mathsf{Loc}^m$ to $\mathsf{Ab}$. 
(Indeed, they are just given by composing the corresponding contravariant functors of
degree $k$ in (\ref{eqn:absdiagrm}) with the contravariant $\Hom$-functor $\Hom(\,\cdot\,,\bbT)$ in case of the character
groups or with $\Hom(\,\cdot\,,\bbZ)$ in case of the dual $\bbZ$-modules.)
By the same argument, the full character groups $\DCo{k}{\,\cdot\,}{\bbZ}^\star$ of $\DCo{k}{\,\cdot\,}{\bbZ}$ are given by a covariant functor
$\DCo{k}{\,\cdot\,}{\bbZ}^\star:\mathsf{Loc}^m\to\mathsf{Ab}$.
\begin{propo}
The smooth Pontryagin dual $\DCo{k}{\,\cdot\,}{\bbZ}^\star_\infty$ is a subfunctor of $\DCo{k}{\,\cdot\,}{\bbZ}^\star:\mathsf{Loc}^m\to\mathsf{Ab}$. Furthermore,
(\ref{eqn:absdiagrmdualsmooth}) is a diagram of natural transformations.
\end{propo}
\begin{proof}
Let $f:M\to N$ be any $\mathsf{Loc}^m$-morphism.
Restricting $\DCo{k}{f}{\bbZ}^\star : \DCo{k}{M}{\bbZ}^\star\to \DCo{k}{N}{\bbZ}^\star$ to the smooth Pontryagin
dual $\DCo{k}{M}{\bbZ}^\star_\infty$, we obtain by naturality of the
(unrestricted) morphism $\iota^\star$ the commutative diagram
\begin{flalign}
\xymatrix{
\ar[d]^-{\iota^\star}\DCo{k}{M}{\bbZ}_\infty^\star \
\ar[rr]^-{\DCo{k}{f}{\bbZ}^\star} & & \ \DCo{k}{N}{\bbZ}^\star\ar[d]^-{\iota^\star}\\
\mathcal{V}^{k-1}(M) \ \ar[rr]^-{\mathcal{V}^{k-1}(f)} && \ \mathcal{V}^{k-1}(N)
}
\end{flalign}
Hence the image of $\DCo{k}{M}{\bbZ}_\infty^\star$ under $\DCo{k}{f}{\bbZ}^\star$ is contained in the inverse image
of $\mathcal{V}^{k-1}(N)$ under $\iota^\star$, which is by the definition (\ref{eqn:smoothdualdc}) the smooth Pontryagin dual $\DCo{k}{N}{\bbZ}^\star_\infty$. Thus $\DCo{k}{\,\cdot\,}{\bbZ}^\star_\infty$ is a subfunctor of $\DCo{k}{\,\cdot\,}{\bbZ}^\star$.
 \sk
 
Finally, (\ref{eqn:absdiagrmdualsmooth}) is a diagram of natural transformations since 
it is the restriction to smooth Pontryagin duals of the diagram (\ref{eqn:absdiagrmdual})
of natural transformations, which is given by acting with the $\Hom$-functor $\Hom(\,\cdot\,,\bbT)$ 
on the degree $k$ component of the natural diagram (\ref{eqn:absdiagrm}).
\end{proof}
\begin{rem}
For any $\mathsf{Loc}^m$-morphism $f: M\to N$ we shall denote the
restriction of $\DCo{k}{f}{\bbZ}^\star$ to $\DCo{k}{M}{\bbZ}_\infty^\star$ by
$\DCo{k}{f}{\bbZ}_\infty^\star : \DCo{k}{M}{\bbZ}_\infty^\star\to \DCo{k}{N}{\bbZ}_\infty^\star$ .
\end{rem}


\section{\label{sec:presymp}Presymplectic Abelian group functors}

As a preparatory step towards the quantization of the smooth Pontryagin dual
$\DCo{k}{\,\cdot\,}{\bbZ}^\star_\infty : \mathsf{Loc}^m\to\mathsf{Ab}$ of a 
degree $k$ differential cohomology theory we have to equip the
Abelian groups $\DCo{k}{M}{\bbZ}^\star_\infty$ with a natural presymplectic structure $\widehat{\tau} : \DCo{k}{M}{\bbZ}^\star_\infty\times
\DCo{k}{M}{\bbZ}^\star_\infty \to\bbR$.
A useful selection criterion for these structures
is given by Peierls' construction \cite{Peierls:1952cb} that allows us to derive a
Poisson bracket which can be used as a presymplectic structure on 
$\DCo{k}{M}{\bbZ}^\star_\infty$. We shall now explain this
construction in some detail, referring to \cite[Remark 3.5]{Benini:2013tra} where a similar construction is done for 
connections on a fixed $\bbT$-bundle.
\sk

Let $M$ be any object in $\mathsf{Loc}^m$. Recall that any element $\mathrm{w}\in \DCo{k}{M}{\bbZ}^\star_\infty$
is a group character, i.e.\ a homomorphism of Abelian groups $\mathrm{w}: \DCo{k}{M}{\bbZ} \to\bbT$ to the circle
group $\bbT$. Using the inclusion $\bbT\hookrightarrow \bbC$ of the circle group into the complex numbers
of modulus one, we may regard $\mathrm{w}$ as a complex-valued functional, i.e.\ $\mathrm{w}: \DCo{k}{M}{\bbZ} \to\bbC$.
We use the following notion of functional derivative, which we derive in Appendix \ref{app:frechet} from 
a Fr{\'e}chet-Lie group structure on $\DCo{k}{M}{\bbZ} $.
\begin{defi}\label{def:funcder}
For any $\mathrm{w}\in \DCo{k}{M}{\bbZ}^\star_\infty$ considered as a complex-valued functional 
$\mathrm{w}: \DCo{k}{M}{\bbZ} \to\bbC$, the {\bf functional derivative} of $\mathrm{w}$ at $h\in\DCo{k}{M}{\bbZ}$ 
along the vector $[\eta]\in\Omega^{k-1}(M)/\dd\Omega^{k-2}(M)$ (if it exists) is defined by
\begin{flalign}\label{eqn:funcder}
\mathrm{w}^{(1)}(h)\big([\eta]\big) := \lim_{\epsilon\to 0} \,
\frac{\mathrm{w}\big(h + \iota\big([\epsilon \, \eta]\big)\big) -\mathrm{w}(h)}{\epsilon}~,
\end{flalign}
where we have suppressed the projection
$\Omega^{k-1}(M)/\dd\Omega^{k-2}(M)\to \Omega^{k-1}(M)/\Omega^{k-1}_{\bbZ}(M)$  
that is induced by the identity on $\Omega^{k-1}(M)$.
\end{defi}
\begin{propo}
For any $\mathrm{w}\in \DCo{k}{M}{\bbZ}^\star_\infty$,  $h\in\DCo{k}{M}{\bbZ}$  and  $[\eta]\in\Omega^{k-1}(M)/\dd\Omega^{k-2}(M)$
the functional derivative exists and reads as
\begin{flalign}
\mathrm{w}^{(1)}(h)\big([\eta]\big) = 2\pi\ii\, \mathrm{w}(h) \, \ip{\iota^{\star}(\mathrm{w})}{[\eta]}~.
\end{flalign}
\end{propo}
\begin{proof}
We compute (\ref{eqn:funcder}) explicitly to get
\begin{flalign}
\nn \mathrm{w}^{(1)}(h)\big([\eta]\big)  &= \lim_{\epsilon\to 0} \, \frac{\mathrm{w}(h)~ \mathrm{w}\big(\iota\big([\epsilon\, \eta]\big)\big) -\mathrm{w}(h)}{\epsilon}\\[4pt]
 &= \mathrm{w}(h)~\lim_{\epsilon\to 0} \,
\frac{\exp\big(2\pi\ii\ip{\iota^\star(\mathrm{w})}{[\epsilon\, \eta]}\big) -1}{\epsilon} =
\mathrm{w}(h) ~2\pi\ii \ip{\iota^{\star}(\mathrm{w})}{[\eta]}~.
\end{flalign}
In the first equality we have used the fact that $\mathrm{w}$ is a homomorphism of Abelian groups 
and in the second equality the group homomorphism (\ref{eqn:calWmap}).
\end{proof}

To work out Peierls' construction we need a Lagrangian, which we
take to be
the generalized Maxwell Lagrangian
\begin{flalign}\label{eqn:lagrangian}
L(h) = \frac{\lambda}{2} \,\Curv(h)\wedge\ast\, \big(\Curv(h)\big)~, 
\end{flalign}
where $\lambda>0$ is a ``coupling'' constant and the factor $\frac12$ is purely conventional. The corresponding Euler-Lagrange equation coincides (up to the factor $\lambda$) with the Maxwell map
defined in Section \ref{sec:MWmaps}, so they have the same solution subgroups. Given any solution
$h\in\widehat{\mathfrak{Sol}}{}^k(M)$ of the Euler-Lagrange equation $\lambda\,\delta\big(\Curv(h)\big)=0$,
Peierls' proposal is to study the retarded/advanced effect of a functional $\mathrm{w}$ on this solution.
Adapted to our setting, we shall introduce a formal parameter $\varepsilon$ and search
for $\eta_\mathrm{w}^\pm \in \Omega^{k-1}(M)$ such that $h_\mathrm{w}^\pm := h + \iota\big([\varepsilon\, \eta^\pm_\mathrm{w}]\big) $ 
solves the partial differential equation
\begin{flalign}\label{eqn:tmpPeierlsinhom}
\lambda \,\delta\big(\Curv(h_\mathrm{w}^\pm)\big) +  \varepsilon\, \mathrm{w}^{(1)}(h^\pm_\mathrm{w})=0
\end{flalign}
up to first order in $\varepsilon$ and such that $\eta_\mathrm{w}^\pm$ satisfies a suitable asymptotic condition to
be stated below. Expanding (\ref{eqn:tmpPeierlsinhom}) to first order in $\varepsilon$ (and using $\delta\big(\Curv(h)\big)=0$) yields
the inhomogeneous equation
\begin{flalign}\label{eqn:tmp:Peierlsinhom2}
\lambda\,\delta \dd\eta^{\pm}_\mathrm{w} + 2\pi\ii \mathrm{w}(h) \, \iota^\star(\mathrm{w})=0~.
\end{flalign}
The requisite asymptotic condition on $\eta_\mathrm{w}^\pm$ is as follows: 
There exist small gauge transformations $\dd \chi_\mathrm{w}^\pm \in \dd \Omega^{k-2}(M)$ and Cauchy surfaces 
$\Sigma_\mathrm{w}^\pm$ in $M$ such that
\begin{flalign}
\big(\eta^\pm_\mathrm{w} + \dd\chi^\pm_{\mathrm{w}}\big)\big\vert_{J_M^{\mp}(\Sigma_\mathrm{w}^\pm)} = 0~,
\end{flalign}
where $J^{\pm}_M(A)$ denotes the causal future/past of a subset
$A\subseteq M$. In simple terms, this requires
$\eta^+_\mathrm{w}$ to be pure gauge in the far past and $\eta^-_\mathrm{w}$ to be pure gauge in the far future.
Under these assumptions, the unique (up to small gauge invariance) solution to (\ref{eqn:tmp:Peierlsinhom2})
is given by $\eta^\pm_\mathrm{w} =  -\frac{2\pi\ii}{\lambda}\, 
\mathrm{w}(h) \, G^\pm\big(\iota^\star(\mathrm{w})\big)$,
where $G^\pm : \Omega^{k-1}_0(M)\to \Omega^{k-1}(M)$ denote the unique retarded/advanced Green's operators
of the Hodge-d'Alembert operator $\square := \delta \circ \dd + \dd \circ \delta : \Omega^{k-1}(M)\to\Omega^{k-1}(M)$.
Following further the construction of Peierls we define the retarded/advanced effect of $\mathrm{w}$ on
 $\mathrm{v}\in\DCo{k}{M}{\bbZ}^\star_\infty$ (considered
also as a functional $\mathrm{v}: \DCo{k}{M}{\bbZ} \to\bbC$) by taking the functional derivative of $\mathrm{v}$ at $h$ along
$[\eta^{\pm}_\mathrm{w}]$, i.e.\ 
\begin{flalign}
\big(E^\pm_\mathrm{w}(\mathrm{v}) \big)(h) := \mathrm{v}^{(1)}(h)\big([\eta^\pm_\mathrm{w}]\big) = 
\frac{4\pi^2 }{\lambda}
\,\ip{\iota^\star(\mathrm{v})}{G^\pm\big(\iota^\star(\mathrm{w})\big)}
\,\mathrm{v}(h)\,\mathrm{w}(h)~.
\end{flalign}
The difference between the retarded and advanced effects defines a Poisson bracket on the associative, commutative and unital
$\ast$-algebra generated by $\DCo{k}{M}{\bbZ}^\star_\infty$. For two generators $\mathrm{v},\mathrm{w}\in \DCo{k}{M}{\bbZ}^\star_\infty$
the Poisson bracket reads as
\begin{flalign}
\{\mathrm{v},\mathrm{w}\} =  4\pi^2 \, \widehat{\tau}(\mathrm{v},\mathrm{w}) ~\mathrm{v}\,\mathrm{w}~,
\end{flalign}
with the antisymmetric bihomomorphism of Abelian groups
\begin{flalign}
\widehat{\tau} \,:\, \DCo{k}{M}{\bbZ}_\infty^\star\times
\DCo{k}{M}{\bbZ}_\infty^\star \ & \longrightarrow \ \bbR ~, \nonumber \\
(\mathrm{v},\mathrm{w}) \ & \longmapsto \ 
\widehat{\tau}(\mathrm{v},\mathrm{w}) =\lambda^{-1}\,\ip{\iota^\star(\mathrm{v})}{G\big(\iota^\star(\mathrm{w})\big)}~.
\label{eqn:presymp}\end{flalign}
In this expression $G:= G^+ - G^- : \Omega^{k-1}_0(M)\to \Omega^{k-1}(M)$ is the retarded-minus-advanced Green's operator.
Antisymmetry of $\widehat{\tau}$ follows from the fact that $G$ is
formally skew-adjoint
as a consequence of $\square$ being formally self-adjoint with respect to the inner product on forms $\ip{\,\cdot\,}{\,\cdot\,}$.
\sk

By naturality of the Green's operators $G^\pm$ and the inner product $\ip{\,\cdot\,}{\,\cdot\,}$,
the presymplectic structure $\widehat{\tau}$ is also natural.
This allows us to promote the covariant functor $\DCo{k}{\,\cdot\,}{\bbZ}^\star_\infty : \mathsf{Loc}^m\to \mathsf{Ab}$
to a functor with values in the category of presymplectic Abelian groups $\mathsf{PAb}$ defined as follows:
The objects in $\mathsf{PAb}$ are pairs $(G,\sigma)$, where $G$ is an Abelian group and $\sigma : G\times G\to \bbR$ 
is an antisymmetric bihomomorphism of Abelian groups (called a presymplectic structure), 
i.e.\ for any $g\in G$, the maps $\sigma(\,\cdot\,,g)$, $\sigma(g,\,\cdot\,): G\to \bbR$
are both homomorphisms of Abelian groups. The morphisms in $\mathsf{PAb}$ are group homomorphisms 
$\phi : G\to G^\prime$ that preserve the presymplectic structures, i.e.\ $\sigma^\prime\circ ( \phi\times\phi) = \sigma$.
\begin{defi}
The {\bf off-shell presymplectic Abelian group functor} $\widehat{\mathfrak{G}}^k_\mathrm{o}(\,\cdot\, ) : \mathsf{Loc}^m\to \mathsf{PAb}$
for a degree $k$ differential cohomology theory is defined as follows: To an object $M$ in $\mathsf{Loc}^m$ 
it associates the presymplectic Abelian group 
$\widehat{\mathfrak{G}}^k_\mathrm{o}(M) := \big(\DCo{k}{M}{\bbZ}^\star_\infty , \widehat{\tau} \big)$ 
with $\widehat{\tau}$ given in (\ref{eqn:presymp}).
To a $\mathsf{Loc}^m$-morphism $f: M\to N$ it associates the $\mathsf{PAb}$-morphism
$\widehat{\mathfrak{G}}^k_\mathrm{o}(f): \widehat{\mathfrak{G}}^k_\mathrm{o}(M)\to\widehat{\mathfrak{G}}^k_\mathrm{o}(N)$
 that is induced by the $\mathsf{Ab}$-morphism 
  $\DCo{k}{f}{\bbZ}^\star_\infty : \DCo{k}{M}{\bbZ}^\star_\infty\to  \DCo{k}{N}{\bbZ}^\star_\infty$.
\end{defi}
\sk

The terminology {\it off-shell} comes from the physics literature and it means that the Abelian groups underlying $\widehat{\mathfrak{G}}^k_\mathrm{o}(\,\cdot\,)$ are (subgroups of) the character groups
of $\DCo{k}{\,\cdot\,}{\bbZ}$. In contrast, the Abelian groups underlying the {\it on-shell} presymplectic 
Abelian group functor should be (subgroups of) the character groups of the subfunctor $\widehat{\mathfrak{Sol}}{}^k(\,\cdot\,)$
of $\DCo{k}{\,\cdot\,}{\bbZ}$, see Section \ref{sec:MWmaps}. We shall discuss the on-shell presymplectic Abelian group functor
later in this section after making some remarks on $\widehat{\mathfrak{G}}^k_\mathrm{o}(\,\cdot\,)$.
\sk

Our first remark is concerned with the presymplectic structure (\ref{eqn:presymp}). Notice that 
$\widehat{\tau}$ is the pull-back under $\iota^\star$ of the presymplectic structure
on the Abelian group $\mathcal{V}^{k-1}(M)$ given by
\begin{flalign}\label{eqn:trivpresymp}
\tau \,:\, \mathcal{V}^{k-1}(M)\times\mathcal{V}^{k-1}(M) \
\longrightarrow \ \bbR~, \qquad (\varphi,\psi) \ \longmapsto \ \tau(\varphi,\psi) = \lambda^{-1}\, \ip{\varphi}{G(\psi)}~.
\end{flalign}
By using this presymplectic structure, the covariant functor $\mathcal{V}^{k-1}(\,\cdot\,):\mathsf{Loc}^m\to\mathsf{Ab}$ 
can be promoted to a functor $\mathfrak{G}^k_{\mathrm{o}}(\,\cdot\,) : \mathsf{Loc}^m\to\mathsf{PAb}$ 
taking values in the category $\mathsf{PAb}$: For any object $M$ in $\mathsf{Loc}^m$ we set $\mathfrak{G}^k_{\mathrm{o}}(M) := 
\big(\mathcal{V}^{k-1}(M),\tau\big)$ with $\tau$ given in (\ref{eqn:trivpresymp}) and for any
$\mathsf{Loc}^m$-morphism $f:M\to N$ we set 
$\mathfrak{G}^k_{\mathrm{o}}(f) : \mathfrak{G}^k_{\mathrm{o}}(M)\to \mathfrak{G}^k_{\mathrm{o}}(N)$
to be the $\mathsf{PAb}$-morphism induced by the 
$\mathsf{Ab}$-morphism $\mathcal{V}^{k-1}(f): \mathcal{V}^{k-1}(M)\to \mathcal{V}^{k-1}(N)$.
By Theorem \ref{theo:smoothdual} and (\ref{eqn:presymp}), we have a surjective natural transformation
$\iota^\star: \widehat{\mathfrak{G}}^k_\mathrm{o}(\,\cdot\,)\Rightarrow \mathfrak{G}^k_\mathrm{o}(\,\cdot\,)$
 between functors from $\mathsf{Loc}^m$ to $\mathsf{PAb}$.  
Furthermore, by equipping the Abelian groups $H^k(M;\bbZ)^\star$ with
the trivial presymplectic structure, we may regard $H^k(\,\cdot\,;\bbZ)^\star : \mathsf{Loc}^m\to\mathsf{Ab}$ as a covariant
functor with values in $\mathsf{PAb}$.
Theorem \ref{theo:smoothdual} then provides us with a
natural exact sequence in the category $\mathsf{PAb}$ given by
\begin{flalign}
\xymatrix{
0 \ \ar[r] & \ H^k(M;\bbZ)^\star \ \ar[r]^-{\Char^\star} & \
\widehat{\mathfrak{G}}^k_\mathrm{o}(M) \ \ar[r]^-{\iota^\star} & \
\mathfrak{G}^k_\mathrm{o}(M) \ \ar[r] & \ 0  \ .
}
\end{flalign}
If we pull back $\widehat{\tau}$ under the
 natural transformation $\Curv^\star : \Omega^k_0(\,\cdot\,)/\mathcal{V}^k(\,\cdot\,) \Rightarrow \DCo{k}{\,\cdot\,}{\bbZ}^\star_\infty$,
 we can promote the covariant functor $\Omega^k_0(\,\cdot\,)/\mathcal{V}^k(\,\cdot\,) :\mathsf{Loc}^m\to \mathsf{Ab}$
  to a functor with values  in the category 
  $\mathsf{PAb}$, which we denote by $\mathfrak{F}_\mathrm{o}^k(\,\cdot\,): \mathsf{Loc}^m\to\mathsf{PAb}$.
  Using again Theorem \ref{theo:smoothdual}, we obtain a natural
  diagram in the category $\mathsf{PAb}$ given by
  \begin{flalign}\label{eqn:subsystemclassical}
  \xymatrix{
  &&0\ar[d]&&\\
  &&\mathfrak{F}^k_\mathrm{o}(M)\ar[d]^-{\Curv^\star} & &\\
  0 \ \ar[r] & \ H^k(M;\bbZ)^\star \ \ar[r]^-{\Char^\star} & \
  \widehat{\mathfrak{G}}^k_\mathrm{o}(M) \ \ar[r]^-{\iota^\star} & \
  \mathfrak{G}^k_\mathrm{o}(M) \ \ar[r] & \ 0 
  }
  \end{flalign}
  where the horizontal and vertical sequences are exact. 

  \begin{rem}\label{rem:physinterpretation}
  The diagram (\ref{eqn:subsystemclassical}) has the following physical interpretation.
  If we think of the covariant functor $\widehat{\mathfrak{G}}^k_\mathrm{o}(\,\cdot\,)$ as a field theory
   describing classical observables on the differential cohomology groups $\DCo{k}{\,\cdot\,}{\bbZ}$,
   the diagram shows that this field theory has two (faithful) subtheories: The first subtheory $H^k(\,\cdot\,;\bbZ)^\star$
    is purely topological and it describes observables on the cohomology groups $H^k(\,\cdot\,;\bbZ)$. The second 
   subtheory $\mathfrak{F}^k_\mathrm{o}(\,\cdot\,)$ describes only the ``field strength observables'',
   i.e.\ classical observables measuring the curvature of elements in $\DCo{k}{\,\cdot\,}{\bbZ}$.
   In addition to $\widehat{\mathfrak{G}}^k_\mathrm{o}(\,\cdot\,)$ having two subtheories,
   it also projects onto the field theory $\mathfrak{G}^k_\mathrm{o}(\,\cdot\,) $ describing classical observables of topologically trivial fields.
  \end{rem}
    \begin{rem}
  In the $\mathsf{PAb}$-diagram (\ref{eqn:subsystemclassical}) the
  character group $\Co{k-1}{M}{\bbT}^\star$ (cf.\ the $\mathsf{Ab}$-diagram (\ref{eqn:absdiagrmdualsmooth})) does not appear.
  The reason is that there is no presymplectic
  structure on $\Co{k-1}{\,\cdot\,}{\bbT}^\star$ such that the components of both $\Curv^\star$ and $\kappa^\star$ are $\mathsf{PAb}$-morphisms:
  if such a presymplectic structure $\sigma$ would exist, then the presymplectic structure on
  $\mathfrak{F}_\mathrm{o}^k(\,\cdot\,)$ would have to be trivial
   as it would be given by the pull-back of $\sigma$ along $\kappa^\star\circ \Curv^\star =0$.
This is not the case.
We expect that the role of the flat classes $\Co{k-1}{\,\cdot\,}{\bbT}$ is that of a local 
symmetry group of the field theory $\widehat{\mathfrak{G}}^k_\mathrm{o}(\,\cdot\,)$.
This claim is strengthened by noting that adding flat classes does not change the generalized Maxwell 
Lagrangian (\ref{eqn:lagrangian}). In future work we plan to study this local symmetry group in detail and 
also try to understand its role in Abelian S-duality.
\end{rem}
\sk

We shall now discuss the on-shell presymplectic Abelian group functor for a degree $k$ differential cohomology theory.
Recall that for any object $M$ in $\mathsf{Loc}^m$ we 
have that $\widehat{\mathfrak{Sol}}{}^k(M)\subseteq \DCo{k}{M}{\bbZ}$ is a subgroup. Thus any element 
$\mathrm{w}\in \DCo{k}{M}{\bbZ}^\star_\infty$ defines a group character on $\widehat{\mathfrak{Sol}}{}^k(M)$.
However, there are elements $\mathrm{w} \in \DCo{k}{M}{\bbZ}^\star_\infty$ which give rise to a trivial group
 character on $\widehat{\mathfrak{Sol}}{}^k(M)$, i.e.\ $\mathrm{w}(h) =1$ for all $h\in \widehat{\mathfrak{Sol}}{}^k(M)$.
 We collect all of these elements in the vanishing subgroups
 \begin{flalign}
 \widehat{\mathfrak{I}}^k(M) := \big\{\mathrm{w}\in
 \DCo{k}{M}{\bbZ}^\star_\infty \ : \ \mathrm{w}(h) =1 \quad \forall h\in  \widehat{\mathfrak{Sol}}{}^k(M)\big\}~.
 \end{flalign}
Notice that $\widehat{\mathfrak{I}}^k(\,\cdot\,): \mathsf{Loc}^m \to \mathsf{Ab}$ is a subfunctor of $\DCo{k}{\,\cdot\,}{\bbZ}^\star_\infty$.
In order to characterize this subfunctor, let us first dualize the Maxwell maps 
$\mathrm{MW} = \delta\circ \Curv  : \DCo{k}{M}{\bbZ}\to \Omega^{k-1}(M)$ to the smooth Pontryagin duals.
This yields the group homomorphisms
\begin{flalign}
\mathrm{MW}^\star = \Curv^\star \circ \dd \,:\, \Omega^{k-1}_0(M) \
\longrightarrow \ \DCo{k}{M}{\bbZ}^\star_\infty~.
\end{flalign}
It is immediate to see that the image
$\mathrm{MW}^\star[\Omega^{k-1}_0(M) ]$ is a subgroup of
$\widehat{\mathfrak{I}}^k(M)$: for any $\rho\in \Omega^{k-1}_0(M) $ we have
\begin{flalign}
\mathrm{MW}^\star(\rho)(h) = \exp\big(2\pi\ii\ip{\rho}{\mathrm{MW}(h)}\big) =1
\end{flalign}
for all $h\in\widehat{\mathfrak{Sol}}{}^k(M)$.
\begin{propo}\label{propo:vanishingsolution}
$\widehat{\mathfrak{I}}^k(M) = \mathrm{MW}^\star[\Omega^{k-1}_0(M)]$.
\end{propo}
\begin{proof}
The inclusion ``$\supseteq$'' was shown above. To prove the inclusion ``$\subseteq$'', let us take an arbitrary 
$\mathrm{w}\in\widehat{\mathfrak{I}}^k(M) $. By (\ref{eqn:soldiagrm}) we have 
$\kappa[\Co{k-1}{M}{\bbT}]\subseteq \widehat{\mathfrak{Sol}}{}^k(M)$, which implies that 
$\kappa^\star(\mathrm{w})=0$ is the trivial group character on $\Co{k-1}{M}{\bbT}$.
As a consequence of (\ref{eqn:absdiagrmdualsmooth}) (exactness of the middle horizontal sequence),
we have $\mathrm{w} = \Curv^\star\big([\varphi]\big)$ for some $[\varphi]\in \Omega^k_0(M)/\mathcal{V}^k(M)$.
Furthermore, applying $\iota^\star$ on $\mathrm{w}$ and using again
the commutative diagram (\ref{eqn:absdiagrmdualsmooth}) we find $\iota^\star(\mathrm{w}) = \delta\varphi\in\mathcal{V}^{k-1}(M)$.
Due to the injective group homomorphism $\iota$ (cf.\ (\ref{eqn:soldiagrm})), the group character $ \iota^\star(\mathrm{w}) $
has to be trivial on $\mathfrak{Sol}^k(M)$, i.e.\ 
\begin{flalign}
\exp\big(2\pi\ii\ip{\delta\varphi}{[\eta]}\big)=1
\end{flalign}
for all $[\eta]\in\mathfrak{Sol}^k(M)$.
It is immediate to see that  $[G(\alpha)]\in \mathfrak{Sol}^k(M)$, 
for any $\alpha \in\Omega^{k-1}_{\mathrm{tc},\,\delta}(M)$ of timelike compact support,
and by \cite[Theorem 7.4]{Benini:2014vsa} any $[\eta]\in \mathfrak{Sol}^k(M)$ has a representative of the 
form $\eta = G(\alpha)$, with $\alpha \in\Omega^{k-1}_{\mathrm{tc},\,\delta}(M)$.  Hence, we obtain the equivalent condition
\begin{flalign}
\ip{\delta\varphi}{ G(\alpha)} =0
\end{flalign}
for all $\alpha\in\Omega^{k-1}_{\mathrm{tc},\,\delta}(M)$.
For $k=1$ this condition implies $G(\delta\varphi) =0$ and hence
$\delta\varphi = \square \rho = \delta \dd\rho$ for some $\rho\in\Omega^{0}_0(M)$
by standard properties of normally hyperbolic operators \cite{Bar:2007zz,Baernew}.
For $k>1$ we use Poincar{\'e} duality 
between spacelike and timelike compact de Rham cohomology groups
\cite[Theorem 6.2]{Benini:2014vsa} to find $G(\delta\varphi) = \dd \chi$ for some 
$\chi\in \Omega^{k-2}_\mathrm{sc}(M)$ of spacelike compact support.
Applying $\delta$ to this equation  yields $0=\delta \dd\chi$ and, again by
\cite[Theorem 7.4]{Benini:2014vsa}, there exists $\beta\in\Omega^{k-2}_{0,\,\delta}(M)$ such that $\dd\chi = \dd G(\beta) $.
Plugging this into the equation above yields $G(\delta \varphi - \dd \beta) =0$ and hence
$\delta\varphi- \dd\beta = \square\rho$ for some $\rho\in\Omega^{k-1}_0(M)$.
Applying $\delta$ and using the fact that $\beta$ is coclosed implies
$-\square\beta = \square\, \delta\rho$, thus
$\beta = -\delta\rho$. Plugging this back into the original equation leads to
$\delta \varphi = \delta \dd\rho$, just as in the case $k=1$.
As a consequence we obtain
\begin{flalign}
\iota^\star\big(\mathrm{MW}^\star(\rho) \big)= \delta \dd\rho = \delta\varphi = \iota^\star(\mathrm{w})~.
\end{flalign}
By exactness of the middle vertical sequence in the diagram (\ref{eqn:absdiagrmdualsmooth})
there exists $\phi \in \Co{k}{M}{\bbZ}^\star$ such that $\mathrm{w} = \mathrm{MW}^\star(\rho) + \Char^\star(\phi)$.
Using the fact that both $\mathrm{w}$ and $\mathrm{MW}^\star(\rho)$ are trivial group
 characters on the solution subgroups we find
 \begin{flalign}
 1 = \Char^\star(\phi)(h) = \phi\big(\Char(h)\big)
 \end{flalign}
for all $h\in\widehat{\mathfrak{Sol}}{}^k(M)$.
 In order to finish the proof we just have to notice that $\Char: \widehat{\mathfrak{Sol}}{}^k(M) \to \Co{k}{M}{\bbZ}$
 is surjective (cf.\ Theorem \ref{theo:soldiagram}), so $\phi$ is the trivial group character 
 and $\mathrm{w} = \mathrm{MW}^\star(\rho)$ with $\rho\in \Omega^{k-1}_0(M)$.
\end{proof}

For any object $M$ in $\mathsf{Loc}^m$, the vanishing subgroup $\widehat{\mathfrak{I}}^k(M)$
is a subgroup of the radical in the presymplectic Abelian group $\widehat{\mathfrak{G}}_\mathrm{o}^k(M)$: given any $\rho\in\Omega^{k-1}_0(M)$ and any $\mathrm{w}\in \DCo{k}{M}{\bbZ}^\star_\infty$, we find 
\begin{flalign}
\nn \widehat{\tau}\big(\mathrm{MW}^\star(\rho), \mathrm{w}\big) &= \lambda^{-1}\,\ip{\iota^\star\big(\mathrm{MW}^\star(\rho)\big)}{G\big(\iota^\star(\mathrm{w})\big)} \\[4pt] &=
\lambda^{-1}\,\ip{\delta \dd \rho}{G\big(\iota^\star(\mathrm{w})\big)} \nn \\[4pt]
&= \lambda^{-1}\,\ip{\rho}{G\big(\delta \dd\,
  \iota^\star(\mathrm{w})\big)} \nn \\[4pt]
&=\lambda^{-1}\,\ip{\rho}{G\big(\square \, \iota^\star(\mathrm{w})\big)} \nn \\[4pt]
&=0~,
\end{flalign}
where in the third equality we have integrated by parts (which is possible since $\rho$ is of compact support) and
in the fourth equality we have used the fact that $\iota^\star(\mathrm{w})\in\mathcal{V}^{k-1}(M)$ is coclosed (cf.\ Lemma 
\ref{lem:vanishinggroup}), hence $\delta \dd \,
\iota^\star(\mathrm{w}) = \square \, \iota^\star(\mathrm{w})$.
The last equality follows from $G\circ\square =0$ on compactly supported forms.
As a consequence we can take the quotient of the covariant functor
 $\widehat{\mathfrak{G}}_\mathrm{o}^k(\,\cdot\,):\mathsf{Loc}^m\to\mathsf{PAb}$ by the subfunctor
$\widehat{\mathfrak{I}}^k(\,\cdot\,)$, which yields another functor to the category $\mathsf{PAb}$.
\begin{defi}
The {\bf on-shell presymplectic Abelian group functor} 
$\widehat{\mathfrak{G}}^k(\,\cdot\, ) : \mathsf{Loc}^m\to \mathsf{PAb}$ for a degree $k$ differential cohomology theory
is defined as the quotient
 $\widehat{\mathfrak{G}}^k(\,\cdot\, ):=
 \widehat{\mathfrak{G}}_\mathrm{o}^k(\,\cdot\, )/\, \widehat{\mathfrak{I}}^k(\,\cdot\,)$.
 Explicitly, it associates to an object $M$ in $\mathsf{Loc}^m$ the presymplectic Abelian group 
$\widehat{\mathfrak{G}}^k(M) :=
\widehat{\mathfrak{G}}_\mathrm{o}^k(M)/\, \widehat{\mathfrak{I}}^k(M)$ with the presymplectic
structure induced from  $\widehat{\mathfrak{G}}_\mathrm{o}^k(M)$.\footnote{
The induced presymplectic structure is well-defined, since, as we have shown above, 
$\widehat{\mathfrak{I}}^k(M)$ is a subgroup of the radical in $\widehat{\mathfrak{G}}_\mathrm{o}^k(M)$.
}
To a $\mathsf{Loc}^m$-morphism $f: M\to N$ it associates the $\mathsf{PAb}$-morphism
$\widehat{\mathfrak{G}}^k(f): \widehat{\mathfrak{G}}^k(M)\to\widehat{\mathfrak{G}}^k(N)$
 that is induced by $\widehat{\mathfrak{G}}_\mathrm{o}^k(f):
  \widehat{\mathfrak{G}}_\mathrm{o}^k(M)\to\widehat{\mathfrak{G}}_\mathrm{o}^k(N)$.
\end{defi}
\sk

We shall now derive the analog of the diagram (\ref{eqn:subsystemclassical}) for on-shell functors.
By construction, it is clear that we have two natural transformations 
$\Char^\star : H^k(\,\cdot\,;\bbZ)^\star \Rightarrow \widehat{\mathfrak{G}}^k(\,\cdot\,)$
and $\Curv^\star : \mathfrak{F}^k_\mathrm{o}(\,\cdot\,)\Rightarrow \widehat{\mathfrak{G}}^k(\,\cdot\,)$,
which however do not have to be injective. To make them injective we have to take a quotient by the kernel subfunctors,
which we now characterize.
\begin{lem}\label{lem:onshellcurv}
For any object $M$ in $\mathsf{Loc}^m$, the kernel of $\Char^\star: H^k(M;\bbZ)^\star \to \widehat{\mathfrak{G}}^k(M)$
is trivial while the kernel $\mathfrak{K}^k(M)$ of $\Curv^\star : \mathfrak{F}^k_\mathrm{o}(M)\to \widehat{\mathfrak{G}}^k(M)$
is given by
\begin{flalign}
\mathfrak{K}^k(M) = \big[\dd\Omega^{k-1}_{0}(M)\big] \ \subseteq \ \frac{\Omega^k_0(M)}{\mathcal{V}^k(M)}~.
\end{flalign}
\end{lem}
\begin{proof}
To prove the first statement, let $\phi\in H^k(M;\bbZ)^\star $ be such that $\Char^\star(\phi)\in\widehat{\mathfrak{I}}^k(M)$.
This means that $\Char^\star(\phi)(h) = \phi\big(\Char(h)\big)=1$ for all $h\in\widehat{\mathfrak{Sol}}{}^k(M)$.
As $\Char: \widehat{\mathfrak{Sol}}{}^k(M) \to H^k(M;\bbZ)$ is surjective (cf.\ Theorem \ref{theo:soldiagram}),
we obtain $\phi =0$ and hence the kernel of $\Char^\star: H^k(M;\bbZ)^\star \to \widehat{\mathfrak{G}}^k(M)$ is trivial.
\sk

Next, let us notice that by Theorem \ref{theo:soldiagram} an element 
$[\varphi]\in \Omega^k_0(M)/\mathcal{V}^k(M)$ is in the kernel of 
$\Curv^\star :\mathfrak{F}^k_\mathrm{o}(M)\to \widehat{\mathfrak{G}}^k(M)$ 
if and only if
\begin{flalign}\label{eqn:tmpcurvsol}
1= \exp\big(2\pi\ii \,\ip{\varphi}{\omega}\big)
\end{flalign}
for all $\omega\in \Omega^k_{\bbZ,\,\delta}(M)$ and 
for any choice of representative $\varphi$ of the class $[\varphi]$. It is clear that for 
$\varphi = \dd\rho \in \dd\Omega^{k-1}_{0}(M)$ the condition (\ref{eqn:tmpcurvsol}) is fulfilled: just integrate $\dd$ by parts
and use $\delta\omega=0$.
To show that any $[\varphi]$ satisfying (\ref{eqn:tmpcurvsol}) has a representative $\varphi = \dd\rho\in \dd\Omega^{k-1}_{0}(M) $,
we first use the fact that the vector space $\{G(\dd\beta) : \beta\in\Omega^{k-1}_{\mathrm{tc},\,\delta}(M) \}$ 
is a subgroup of $(\dd\Omega^{k-1})_\delta(M)\subseteq \Omega^{k}_{\bbZ,\,\delta}(M)$. Then
the condition (\ref{eqn:tmpcurvsol}) in particular implies that
\begin{flalign}
0 = \ip{\varphi}{G(\dd\beta)} = -\ip{G(\delta\varphi)}{\beta}
\end{flalign}
for all $\beta \in \Omega^{k-1}_{\mathrm{tc},\,\delta}(M)$.
Arguing as in the proof of Proposition \ref{propo:vanishingsolution},
this condition implies that $\delta\varphi = \delta  \dd \rho$ for some $\rho\in \Omega^{k-1}_0(M)$.
Hence $\varphi$ is of the form $\varphi = \dd\rho + \varphi^\prime $ 
with $\delta\varphi^\prime =0$. Plugging this into (\ref{eqn:tmpcurvsol})
we obtain the condition
\begin{flalign}\label{eqn:tmpcurvsol2}
1 =\exp\big(2\pi\ii \,\big(\ip{\varphi^\prime}{\omega}
+\ip{\rho}{\delta\omega}\big) \big) = \exp\big(2\pi\ii
\,\ip{[\varphi^\prime\, ]}{[\omega]}\big)
\end{flalign}
for all $\omega\in\Omega^k_{\bbZ,\,\delta}(M)$, 
where $[\varphi^\prime\, ]$ is a dual de~Rham class in
$\Omega^{k}_{0,\,\delta}(M)/\delta\Omega^{k+1}_{0}(M)$ and $[\omega]$ 
is a de~Rham class in $\Omega^k_{\dd}(M)/\dd\Omega^{k-1}(M)$.
As by Theorem \ref{theo:soldiagram} the de Rham class mapping $\Omega^k_{\bbZ,\,\delta}(M) \to H^k_\mathrm{free}(M;\bbZ)$
is surjective, (\ref{eqn:tmpcurvsol2}) implies that $[\varphi^\prime\,
]\in H^k_\mathrm{free}(M;\bbZ)^\prime$ and hence
that $\varphi$ is equivalent to $\dd\rho$ in $\Omega^k_0(M)/\mathcal{V}^k(M)$.
\end{proof}
\sk

Taking now the quotient of the covariant functor $\mathfrak{F}^k_\mathrm{o}(\,\cdot\,) : \mathsf{Loc}^m\to \mathsf{PAb}$
by its subfunctor $\mathfrak{K}^k(\,\cdot\,)$, we get another covariant functor 
$\mathfrak{F}^k(\,\cdot\,) :=\mathfrak{F}^k_\mathrm{o}(\,\cdot\,)/\mathfrak{K}^k(\,\cdot\,) : \mathsf{Loc}^m\to \mathsf{PAb}$.
By Lemma \ref{lem:onshellcurv} there are now two {\it injective} natural transformations 
$\Char^\star : H^k(\,\cdot\,;\bbZ)^\star \Rightarrow \widehat{\mathfrak{G}}^k(\,\cdot\,)$
and $\Curv^\star : \mathfrak{F}^k(\,\cdot\,) \Rightarrow \widehat{\mathfrak{G}}^k(\,\cdot\,)$
to the on-shell presymplectic Abelian group functor $\widehat{\mathfrak{G}}^k(\,\cdot\,)$.
To obtain the on-shell analog of the diagram (\ref{eqn:subsystemclassical}) we just have to notice that,
by a proof similar to that of Proposition \ref{propo:vanishingsolution},
the vanishing subgroups of the topologically trivial field theory are given by
\begin{flalign}
\mathfrak{I}^k(M) := \big\{\varphi \in\mathcal{V}^{k-1}(M) \ : \
\exp\big(2\pi\ii\ip{\varphi}{[\eta]}\big)=1 \quad \forall \, [\eta]\in
\mathfrak{Sol}^k(M)\big\} = \delta \dd \Omega^{k-1}_0(M)~.
\end{flalign}
We define the corresponding quotient $\mathfrak{G}^k(\,\cdot\,) 
:= \mathfrak{G}_\mathrm{o}^k(\,\cdot\,)/\mathfrak{I}^k(\,\cdot\,): \mathsf{Loc}^m\to \mathsf{PAb}$
and notice that the natural transformation $\iota^\star : \widehat{\mathfrak{G}}^k(\,\cdot\,)\Rightarrow 
\mathfrak{G}^k(\,\cdot\,)$ is well-defined and surjective. Hence we obtain a natural diagram 
in the category $\mathsf{PAb}$ given by
\begin{flalign}
  \xymatrix{
  &&0\ar[d]&&\\
  &&\mathfrak{F}^k(M)\ar[d]^-{\Curv^\star} & &\\
  0 \ \ar[r] & \ H^k(M;\bbZ)^\star \ \ar[r]^-{\Char^\star} & \
  \widehat{\mathfrak{G}}^k(M) \ \ar[r]^-{\iota^\star} & \ 
  \mathfrak{G}^k(M) \ \ar[r] & \ 0 
  }
  \end{flalign}
  where the horizontal and vertical sequences are exact. The physical interpretation given in Remark
  \ref{rem:physinterpretation} applies to this diagram as well.
  \sk
  
We conclude this section by pointing out that the subtheory $\mathfrak{F}^k(\,\cdot\,)$ of $\widehat{\mathfrak{G}}^k(\,\cdot\,)$
has a further purely topological subtheory. Let $M$ be any object in $\mathsf{Loc}^m$.
Recall that the Abelian group underlying $\mathfrak{F}^k(M)$ is given by the double quotient
 $\big(\Omega^k_0(M)/\mathcal{V}^k(M)\big)/\big[\dd\Omega^{k-1}_0(M)\big]$, classes of which we denote by double brackets, e.g.\
 $[[\varphi]]$. There is a natural group homomorphism $\Omega^k_{0,\,\dd }(M) \to \mathfrak{F}^k(M)\,,~\varphi\mapsto [[\varphi]]$,
 which induces to the quotient $\Omega^k_{0,\,\dd}(M)/\dd\Omega^{k-1}_0(M) \to \mathfrak{F}^{k}(M)$ since
 $[[\dd\rho]] =0 $ for any $\rho\in\Omega^{k-1}_0(M)$. This group homomorphism is injective: if $[[\varphi]]=0$ for some $\varphi\in\Omega^{k}_{0,\,\dd}(M)$, then $\varphi = \dd\rho + \varphi^\prime$ for some
 $\rho\in \Omega^{k-1}_0(M)$ and $\varphi^\prime \in \mathcal{V}^k(M)$ (in particular $\delta\varphi^\prime=0$). 
 Since $\dd\varphi=0$ we also have
 $\dd\varphi^\prime=0$ and hence $\square \varphi^\prime =0 $, which implies $\varphi^\prime=0$.
 By Poincar{\'e} duality and de Rham's theorem, 
  the quotient $\Omega^k_{0,\,\dd}(M)/\dd\Omega^{k-1}_0(M) $ can be canonically identified with 
  $\Hom_\bbR(H^{m-k}(M;\bbR),\bbR)\simeq H^{m-k}(M;\bbR)^\star$, the character group of the $(m{-}k)$-th
  singular cohomology group with coefficients in $\bbR$. We denote the injective natural 
  transformation constructed above by $\mathrm{q}^\star : H^{m-k}(\,\cdot\,;\bbR)^\star\Rightarrow \mathfrak{F}^k(\,\cdot\,)$.
The pull-back of the presymplectic structure on $\mathfrak{F}^k(\,\cdot\,)$ to $H^{m-k}(\,\cdot\,;\bbR)^\star$
 is trivial. In fact, the presymplectic structure on $H^{m-k}(M;\bbR)^\star$ is then given by pulling back
 $\tau$ in (\ref{eqn:trivpresymp}) via $\iota^\star\circ\Curv^\star\circ \mathrm{q}^\star$. 
 For any two elements in $H^{m-k}(M;\bbR)^\star$, which we represent by two classes
  $[\varphi],[\varphi^\prime\, ]\in \Omega^k_{0,\,\dd}(M)/\dd\Omega^{k-1}_0(M)$ with the isomorphism above,
  we find
 \begin{flalign}
 \tau\big(\iota^\star\circ\Curv^\star\circ
 \mathrm{q}^\star\big([\varphi]\big), \iota^\star\circ\Curv^\star\circ
 \mathrm{q}^\star\big([\varphi^\prime\, ]\big)\big)
 =\tau(\delta\varphi,\delta\varphi^\prime\, ) = \lambda^{-1}\,
 \ip{\varphi}{G\big(\dd\delta\varphi^\prime\, \big)} =0
 \end{flalign}
 since $ \varphi^\prime$ is closed and hence $G(\dd
 \delta\varphi^\prime\, ) = G(\square\varphi^\prime\, ) =0$.
 In order to get a better understanding of $\mathrm{q}^\star$, let us compute how
elements in the image of $\Curv^\star \circ \mathrm{q}^\star: H^{m-k}(M;\bbR)^\star\to \widehat{\mathfrak{G}}^k(M)$
act on the solution subgroup $\widehat{\mathfrak{Sol}}{}^k(M)$. For any element in
$H^{m-k}(M;\bbR)^\star$, which we represent by a class
  $[\varphi]\in \Omega^k_{0,\,\dd}(M)/\dd\Omega^{k-1}_0(M)$ with the isomorphism above, 
 and any $h\in \widehat{\mathfrak{Sol}}{}^k(M)$ we find
 \begin{flalign}\label{eqn:electric}
 \big(\Curv^\star \circ \mathrm{q}^\star([\varphi])\big)(h) = 
 \exp\big(2\pi\ii \ip{\varphi}{\Curv(h)}\big) = \exp\Big(2\pi\ii
 \int_M\, \varphi \wedge \ast\, \big(\Curv(h)\big)\Big)~.
 \end{flalign}
 Since $h$ lies in the kernel of the Maxwell map, the forms $\ast\big(\Curv(h)\big)\in\Omega^{m-k}(M)$ are closed
 and the integral in (\ref{eqn:electric}) depends only on the de Rham classes
 $[\varphi] \in H^k_0(M;\bbR)$ and $\big[\ast\big(\Curv(h)\big) \big] \in H^{m-k}(M;\bbR)$. So the observables described by the subtheory $H^{m-k}(\,\cdot\,;\bbR)^\star$
 of $\widehat{\mathfrak{G}}^k(\,\cdot\,)$ are exactly those measuring the de~Rham class of the dual curvature
 of a solution. 
 \begin{rem}\label{rem:subtheoryonshell}
 Following the terminology used in ordinary Maxwell theory (given in degree $k=2$) we may 
 call the subtheory $H^{m-k}(\,\cdot\,;\bbR)^\star$ {\it electric} and the subtheory $H^k(\,\cdot\,;\bbZ)^\star$ {\it magnetic}.
The structures we have found for the on-shell field theory 
can be summarized by the following diagram with all horizontal and vertical sequences exact:
\begin{flalign}\label{eqn:subtheoryonshell}
  \xymatrix{
  &&\ar[d]0&&\\
 0 \ \ar[r] & \ {\underbrace{H^{m-k}(M;\bbR)^\star}_{\text{\rm
       electric}}} \ \ar[r]^-{\mathrm{q}^\star}& \ \mathfrak{F}^k(M)\ar[d]^-{\Curv^\star} && \\
  0 \ \ar[r] & \ {\underbrace{H^k(M;\bbZ)^\star}_{\text{\rm
        magnetic}}} \ \ar[r]^-{\Char^\star} & \
  \widehat{\mathfrak{G}}^k(M) \ \ar[r]^-{\iota^\star} & \
  \mathfrak{G}^k(M) \ \ar[r]& \ 0
  }
\end{flalign}
 \end{rem}


\section{\label{sec:quantization}Quantum field theory}

In the previous section we have derived various functors from the category
$\mathsf{Loc}^m$ to the category $\mathsf{PAb}$ of presymplectic Abelian groups.
In particular, the functor $\widehat{\mathfrak{S}}^k(\,\cdot\,)$ describes the association of the
smooth Pontryagin duals (equipped with a natural presymplectic structure)  
of the solution subgroups of a degree $k$ differential cohomology theory. 
To quantize this field theory, we shall make use of the $\mathrm{CCR}$-functor for presymplectic Abelian groups,
see \cite{Manuceau:1973yn} and \cite[Appendix A]{Benini:2013ita} for details.
In short, canonical quantization is a covariant functor $\mathfrak{CCR}(\, \cdot\,): \mathsf{PAb}\to C^\ast\mathsf{Alg}$
to the category of unital $C^\ast$-algebras with morphisms given by unital $C^\ast$-algebra homomorphisms (not necessarily injective).
To any presymplectic Abelian group $(G,\sigma)$ this functor associates the unital $C^\ast$-algebra $\mathfrak{CCR}(G,\sigma)$,
which is generated by the symbols $W(g)$, $g\in G$, satisfying the
Weyl relations $W(g)\,W(\tilde g) = \e^{-\ii \sigma(g,\tilde g)/2}\,W(g+\tilde g)$ and the $\ast$-involution property
$W(g)^\ast = W(-g)$. This unital $\ast$-algebra is then equipped and completed with respect to a suitable $C^\ast$-norm.
To any $\mathsf{PAb}$-morphism $\phi : (G,\sigma) \to
(G^\prime,\sigma^\prime\, )$ the functor associates
the $C^\ast\mathsf{Alg}$-morphism $\mathfrak{CCR}(\phi):
\mathfrak{CCR}(G,\sigma) \to \mathfrak{CCR}(G^\prime,\sigma^\prime\, )$,
which is obtained as the unique continuous extension of the unital $\ast$-algebra homomorphism defined on
generators by $W(g) \mapsto W(\phi(g))$. 

\begin{defi}
The {\bf quantum field theory functor} $\widehat{\mathfrak{A}}^k(\,\cdot\,) : \mathsf{Loc}^m\to C^\ast\mathsf{Alg}$ 
for a degree $k$ differential cohomology theory is defined as the composition of the on-shell
 presymplectic Abelian group functor $\widehat{\mathfrak{G}}^k(\,\cdot\,) :
\mathsf{Loc}^m \to\mathsf{PAb}$ with the $\mathrm{CCR}$-functor $\mathfrak{CCR}(\,\cdot\,) :\mathsf{PAb}\to C^\ast\mathsf{Alg} $, i.e.\
\begin{flalign}
\widehat{\mathfrak{A}}^k(\,\cdot\,) := \mathfrak{CCR}(\,\cdot\,)\circ \widehat{\mathfrak{G}}^k(\,\cdot\,) ~.
\end{flalign}
\end{defi}
\begin{rem}\label{rem:subsystemquantum}
The subtheory structure of the classical on-shell field theory explained in Remark \ref{rem:subtheoryonshell} 
is also present, with a slight caveat, in the quantum field theory. Acting with the functor 
$\mathfrak{CCR}(\,\cdot\,)$ on the diagram (\ref{eqn:subtheoryonshell}) we obtain a similar
 diagram in the category $C^\ast\mathsf{Alg}$ (with $\mathfrak{CCR}(0) = \bbC$, the trivial unital $C^\ast$-algebra).
 However, the sequences in this diagram will in general not be exact, as the $\mathrm{CCR}$-functor
 is not an exact functor.
 This will not be of major concern to us, since by \cite[Corollary A.7]{Benini:2013ita} the functor
 $\mathfrak{CCR}(\,\cdot\,)$ does map injective $\mathsf{PAb}$-morphisms to injective
 $C^\ast\mathsf{Alg}$-morphisms. Thus our statements in Remark \ref{rem:subtheoryonshell} about
  the (faithful) subtheories remain valid after quantization. Explicitly,
 the quantum field theory $\widehat{\mathfrak{A}}^k(\,\cdot\,)$ has three faithful subtheories
 $\mathfrak{A}^k_{\mathrm{el}}(\,\cdot\,) := \mathfrak{CCR}(\,\cdot\,)\circ H^{m-k}(\,\cdot\,;\bbR)^\star$,
  $\mathfrak{A}^k_{\mathrm{mag}}(\,\cdot\,) := \mathfrak{CCR}(\,\cdot\,)\circ H^k(\,\cdot\,;\bbZ)^\star$ 
 and $\mathfrak{A}_\mathfrak{F}^k(\,\cdot\,):= \mathfrak{CCR}(\,\cdot\,)\circ\mathfrak{F}^k(\,\cdot\,) $,
  with the first two being purely topological and the third being a theory of quantized curvature observables.
\end{rem}
\sk

We shall now address the problem of whether or not our functor $\widehat{\mathfrak{A}}^k(\,\cdot\,)$ satisfies the axioms of locally covariant quantum field theory,
which have been proposed in \cite{Brunetti:2001dx} to single out physically reasonable models
 for quantum field theory from all possible covariant functors from $\mathsf{Loc}^m$ to $C^\ast\mathsf{Alg}$.
 The first axiom formalizes the concept of Einstein causality.
\begin{theo}
The functor $\widehat{\mathfrak{A}}^k(\,\cdot\,) : \mathsf{Loc}^m\to C^\ast\mathsf{Alg}$ satisfies the 
causality axiom: For any pair of $\mathsf{Loc}^m$-morphisms $f_1 : M_1\to M$ and $f_2: M_2\to M$ such that
$f_1[M_1]$ and $f_2[M_2]$ are causally disjoint subsets of $M$, the subalgebras 
$\widehat{\mathfrak{A}}^k(f_1)\big[\widehat{\mathfrak{A}}^k(M_1)\big]$ and $\widehat{\mathfrak{A}}^k(f_2)\big[\widehat{\mathfrak{A}}^k(M_2)\big]$
of $\widehat{\mathfrak{A}}^k(M)$ commute.
\end{theo}
\begin{proof}
For any two generators $W(\mathrm{w})\in \widehat{\mathfrak{A}}^k(M_1)$ and $W(\mathrm{v})\in \widehat{\mathfrak{A}}^k(M_2)$
we have
\begin{multline}
\big[\,
\widehat{\mathfrak{A}}^k(f_1)\big(W(\mathrm{w})\big),\widehat{\mathfrak{A}}^k(f_2)\big(W(\mathrm{v})\big)\,
\big] = 
\big[W\big(\widehat{\mathfrak{G}}^k(f_1)(\mathrm{w})\big),W\big(\widehat{\mathfrak{G}}^k(f_2)(\mathrm{v})\big)\big] \\[4pt]
=-2\ii
W\big(\widehat{\mathfrak{G}}^k(f_1)(\mathrm{w})+\widehat{\mathfrak{G}}^k(f_2)(\mathrm{v})\big)
\,\sin \Big(\mbox{$\frac12$}\, \tau\big(f_{1\ast}( \iota^\star(\mathrm{w})) , f_{2\ast} (\iota^\star(\mathrm{v}))\big) \Big)=0~,
\end{multline}
where we have used the Weyl relations and the fact that, by hypothesis, the push-forwards $f_{1\ast}( \iota^\star(\mathrm{w}))$
and $f_{2\ast} (\iota^\star(\mathrm{v}))$ are differential forms of causally disjoint support,
 for which the presymplectic structure (\ref{eqn:trivpresymp})
vanishes. The result now follows by approximating generic elements in $\widehat{\mathfrak{A}}^k(M_1)$
and $\widehat{\mathfrak{A}}^k(M_2)$ by linear combinations of generators and using continuity of $\widehat{\mathfrak{A}}^k(f_1)$
and $\widehat{\mathfrak{A}}^k(f_2)$.
\end{proof}

The second axiom formalizes the concept of a dynamical law. 
Recall that a $\mathsf{Loc}^m$-morphism $f:M\to N$ 
is called a Cauchy morphism if the image $f[M]$ contains a Cauchy surface of $N$.
\begin{theo}
The functor $\widehat{\mathfrak{A}}^k(\,\cdot\,) : \mathsf{Loc}^m\to C^\ast\mathsf{Alg}$ satisfies the 
time-slice axiom: If $f: M\to N$ is a Cauchy morphism, then $\widehat{\mathfrak{A}}^k(f) : \widehat{\mathfrak{A}}^k(M)
\to \widehat{\mathfrak{A}}^k(N)$ is a $C^\ast\mathsf{Alg}$-isomorphism.
\end{theo}
\begin{proof}
Recall that the Abelian groups underlying $\widehat{\mathfrak{G}}^k(\,\cdot\,)$ are subgroups of the
character groups of $\widehat{\mathfrak{Sol}}{}^k(\,\cdot\,)$ and that by definition
$\widehat{\mathfrak{G}}^k(f) : \widehat{\mathfrak{G}}^k(M) \to \widehat{\mathfrak{G}}^k(N)\,,
~\mathrm{w} \mapsto \mathrm{w}\circ \widehat{\mathfrak{Sol}}{}^k(f)$. Using Theorem \ref{theo:soltimeslice}
we have that $\widehat{\mathfrak{Sol}}{}^k(f)$ is an $\mathsf{Ab}$-isomorphism for any Cauchy morphism $f:M\to N$,
hence $\widehat{\mathfrak{G}}^k(f)$ is a $\mathsf{PAb}$-isomorphism and
as a consequence of functoriality $\widehat{\mathfrak{A}}^k(f)
=\mathfrak{CCR}\big(\, \widehat{\mathfrak{G}}^k(f)\, \big) $
is a $C^\ast\mathsf{Alg}$-isomorphism.
\end{proof}

In addition to the causality and time-slice axioms, \cite{Brunetti:2001dx} proposed the locality axiom which demands
that the functor $\widehat{\mathfrak{A}}^k(\,\cdot\,):\mathsf{Loc}^m\to C^\ast\mathsf{Alg}$ should map
any $\mathsf{Loc}^m$-morphism $f:M\to N$ to an {\it injective} $C^\ast\mathsf{Alg}$-morphism
$\widehat{\mathfrak{A}}^k(f):\widehat{\mathfrak{A}}^k(M)\to \widehat{\mathfrak{A}}^k(N)$. The physical
idea behind this axiom is that any observable quantity on a sub-spacetime $M$ should also be an observable
quantity on the full spacetime $N$ into which it embeds via $f:M\to N$. It is known that 
this axiom is not satisfied in various formulations of Maxwell's theory, see e.g.\
 \cite{Dappiaggi:2011zs,Benini:2013tra,Benini:2013ita,Sanders:2012sf,Fewster:2014hba}.
The violation of the locality axiom is shown in most of these works by giving an example of a
 $\mathsf{Loc}^m$-morphism $f:M\to N$ such that the induced $C^\ast$-algebra morphism is not injective.
A detailed characterization and understanding of which $\mathsf{Loc}^m$-morphisms violate the
locality axiom is given in \cite{Benini:2013ita} for a theory of connections on fixed $\bbT$-bundles.
It is shown there that a morphism violates the locality axiom if and only if the induced morphism
between the compactly supported de Rham cohomology groups of degree
$2$ is not injective. 
Thus the locality axiom is violated due to topological obstructions.
Our present theory under consideration has a much richer topological structure than
a theory of connections on a fixed $\bbT$-bundle, see Remark \ref{rem:subsystemquantum}. 
It is therefore important to extend the analysis of \cite{Benini:2013ita} to our functor
$\widehat{\mathfrak{A}}^k(\,\cdot\,) : \mathsf{Loc}^m\to C^\ast\mathsf{Alg}$ in order to characterize
exactly those $\mathsf{Loc}^m$-morphisms which violate the locality axiom.
\sk

We collect some results which will simplify our analysis:
Let $\phi : (G,\sigma)\to (G^\prime,\sigma^\prime\, )$ be any $\mathsf{PAb}$-morphism.
Then $\mathfrak{CCR}(\phi)$ is injective if and only if $\phi$ is injective:
the direction ``$\Leftarrow$'' is shown in \cite[Corollary A.7]{Benini:2013ita} and the 
 direction ``$\Rightarrow$'' is an obvious proof by contraposition (which is spelled out in
 \cite[Theorem 5.2]{Benini:2013ita}).  Hence our problem of characterizing all $\mathsf{Loc}^m$-morphisms
 $f:M\to N$ for which $\widehat{\mathfrak{A}}^k(f)$ is injective is equivalent 
 to the classical problem of characterizing all $\mathsf{Loc}^m$-morphisms
  $f:M\to N$ for which $\widehat{\mathfrak{G}}^k(f)$ is injective.
 Furthermore, the kernel of any $\mathsf{PAb}$-morphism $\phi :
 (G,\sigma)\to (G^\prime,\sigma^\prime\, )$
 is a subgroup of the radical in $(G,\sigma)$: if $g\in G$ with $\phi(g) =0$ then
  $0 = \sigma^\prime\big(\phi(g),\phi(\tilde g)\big) = \sigma(g,\tilde g)$ for all $\tilde g\in G$,
 which shows that $g$ is an element the radical of $(G,\sigma)$.
For any object $M$ in $\mathsf{Loc}^m$ the radical of $\widehat{\mathfrak{G}}^k(M)$ is easily 
computed:
\begin{lem}\label{lem:radical}
The radical of $\widehat{\mathfrak{G}}^k(M)$ is the subgroup
\begin{flalign}\label{eqn:rad}
\mathrm{Rad}\big(\, \widehat{\mathfrak{G}}^k(M)\, \big) = \Big\{\, \mathrm{v}\in \widehat{\mathfrak{G}}^k(M) 
 \ : \ \iota^\star(\mathrm{v}) \in \frac{\delta\big(\Omega^k_0(M)\cap
   \dd\Omega^{k-1}_\mathrm{tc}(M)\big)}{\delta \dd
   \Omega^{k-1}_0(M)}\, \Big\}~.
\end{flalign}
\end{lem}
\begin{proof}
We show the inclusion ``$\supseteq$'' by evaluating the presymplectic structure
(\ref{eqn:presymp}) for any element $\mathrm{v}$ of the group on
 the right-hand side of (\ref{eqn:rad}) and any $\mathrm{w}\in \widehat{\mathfrak{G}}^k(M)$.
 Using $\iota^\star(\mathrm{v}) = [\delta \dd \rho]$ for some $\rho\in \Omega^{k-1}_\mathrm{tc}(M)$, we obtain
 \begin{flalign}
 \nn \widehat{\tau}(\mathrm{v},\mathrm{w}) &= \lambda^{-1}\,\ip{\iota^\star(\mathrm{v})}{G\big(\iota^\star(\mathrm{w}) \big)}
 \\[4pt] \nn &= \lambda^{-1}\,\ip{\delta
   \dd\rho}{G\big(\iota^\star(\mathrm{w}) \big)}\\[4pt] \nn
 &=\lambda^{-1}\,\ip{\rho}{G\big(\delta \dd \,
   \iota^\star(\mathrm{w}) \big)} \\[4pt] &=
  \lambda^{-1}\,\ip{\rho}{G\big(\square\, \iota^\star(\mathrm{w}) \big)}  \nn \\[4pt] 
 &=0~.
 \end{flalign}
We now show the inclusion ``$\subseteq$''. Let $\mathrm{v}$ be any element in the radical
 of $\widehat{\mathfrak{G}}^k(M)$, i.e. $0 =
 \widehat{\tau}(\mathrm{w},\mathrm{v}) =
 \lambda^{-1}\,\ip{\iota^\star(\mathrm{w})}{G(\iota^\star(\mathrm{v}))}$
 for all $\mathrm{w}\in \widehat{\mathfrak{G}}^k(M)$.
As $\iota^\star$ is surjective, this implies that $ \ip{\varphi}{G(\iota^\star(\mathrm{v}))}=0$ for all $\varphi\in\mathcal{V}^{k-1}(M)$,
from which we can deduce by similar arguments as in the proof of Proposition \ref{propo:vanishingsolution} that 
$\iota^\star(\mathrm{v}) = [\delta \dd \rho]$ for some $\rho\in \Omega^{k-1}_{\mathrm{tc}}(M)$
with $\dd \rho \in \Omega^{k}_0(M)$.
\end{proof}

We show that the radical $\mathrm{Rad}\big(\,
\widehat{\mathfrak{G}}^k(M)\, \big)$, and hence also the kernel of any
$\mathsf{PAb}$-morphism with source given by $\widehat{\mathfrak{G}}^k(M)$, is contained in the
 images of $\Curv^\star \circ \mathrm{q}^\star : \Co{m-k}{M}{\bbR}^\star \to \widehat{\mathfrak{G}}^k(M)$
 and $\Char^\star : \Co{k}{M}{\bbZ}^\star \to  \widehat{\mathfrak{G}}^k(M) $.
 To make this precise, similarly to \cite[Section 5]{Fewster:2014gxa} we may 
 equip the category $\mathsf{PAb}$ with the following monoidal structure $\oplus$:
 For two objects $(G,\sigma)$ and $(G^\prime,\sigma^\prime\, )$ in $\mathsf{PAb}$ we set
 $(G,\sigma)\oplus (G^\prime,\sigma^\prime\, ) := (G\oplus G^\prime, \sigma\oplus \sigma^\prime\, )$, 
 where $G\oplus G^\prime$ denotes the direct sum of Abelian groups and $\sigma\oplus\sigma^\prime$ is the presymplectic
 structure on $G\oplus G^\prime$ defined by $\sigma\oplus\sigma^\prime\big(g\oplus g^\prime, \tilde g\oplus \tilde{g}^\prime\,\big)
 := \sigma(g,\tilde g) + \sigma^\prime(g^\prime,\tilde{g}^\prime\, )$. For two $\mathsf{PAb}$-morphisms
 $\phi_i : (G_i,\sigma_i)\to (G_i^\prime,\sigma^\prime_i\, )$, $i=1,2$, the functor gives the direct sum
 $\phi_1 \oplus \phi_2 : (G_1\oplus G_2 , \sigma_1\oplus\sigma_2)\to 
 (G^\prime_1\oplus G^\prime_2,\sigma^\prime_1\oplus \sigma^\prime_2)$. 
 The identity object is the trivial presymplectic Abelian group.
 We define the covariant functor describing the direct sum of both topological 
 subtheories of $\widehat{\mathfrak{G}}^k(\,\cdot\,)$ by 
 \begin{flalign}
 \mathfrak{Charge}^k(\,\cdot\,) :=  
  \Co{m-k}{\,\cdot\,}{\bbR}^\star \oplus  \Co{k}{\,\cdot\,}{\bbZ}^\star \,:\, \mathsf{Loc}^m \ \longrightarrow \ \mathsf{PAb}~.
  \end{flalign}
There is an obvious natural transformation $\mathrm{top}^\star 
: \mathfrak{Charge}^k(\,\cdot\,)\Rightarrow \widehat{\mathfrak{G}}^k(\,\cdot\,)$
given for any object $M$ in $\mathsf{Loc}^m$ by
 \begin{flalign}
\mathrm{top}^\star \,:\, \mathfrak{Charge}^k(M) \ \longrightarrow \  \widehat{\mathfrak{G}}^k(M)~, \qquad
 \psi\oplus \phi \ \longmapsto \ \Curv^\star\big(\mathrm{q}^\star(\psi)\big)  + \Char^\star(\phi)~.
 \end{flalign}
 This natural transformation is injective by the following argument: Using the isomorphism explained in the paragraph before
 Remark \ref{rem:subtheoryonshell},  we can represent any $\psi \in  \Co{m-k}{M}{\bbR}^\star $
 by a compactly supported de Rham class $[\varphi]\in\Omega^{k}_{0,\,\dd}(M)/\dd\Omega^{k-1}_0(M)$. 
Applying $\iota^\star$ on the equation $\mathrm{top}^\star(\psi\oplus \phi) =0$ implies
$[\delta\varphi] =0$ in $\mathcal{V}^{k-1}(M)/\delta \dd\Omega^{k-1}_0(M)$, i.e.\ $\delta \varphi = \delta\dd\rho$ for
some $\rho\in \Omega^{k-1}_0(M)$, which after applying $\dd$ and using $\dd\varphi = 0$ leads to $\varphi = \dd\rho$, i.e.\
$[\varphi] =0$ and thus $\psi=0$. As $\Char^\star$ is injective, the condition $\mathrm{top}^\star(\psi\oplus \phi) =0$
implies $\psi \oplus \phi =0$ and so $\mathrm{top}^\star$ is injective.
\begin{lem}\label{lem:topradconnection}
The radical $\mathrm{Rad}\big(\,\widehat{\mathfrak{G}}^k(M) \,\big)$
is a subgroup of the image of $\mathfrak{Charge}^k(M)$ under $\mathrm{top}^\star : \mathfrak{Charge}^k(M)\to \widehat{\mathfrak{G}}^k(M)$.
In particular, the kernel of any $\mathsf{PAb}$-morphism with source given by $\widehat{\mathfrak{G}}^k(M)$ is a subgroup
of the image of $\mathfrak{Charge}^k(M)$ under $\mathrm{top}^\star : \mathfrak{Charge}^k(M)\to \widehat{\mathfrak{G}}^k(M)$.
\end{lem}
\begin{proof}
The second statement follows from the first one, since as we have argued above kernels of $\mathsf{PAb}$-morphisms are subgroups
of the radical of the source. To show the first statement,
let $\mathrm{v}\in \mathrm{Rad}\big(\,\widehat{\mathfrak{G}}^k(M)\,\big)$ and notice that by Lemma
\ref{lem:radical} there exists $\rho \in \Omega^{k-1}_{\mathrm{tc}}(M)$ with $\dd\rho$ of compact support such
that $\iota^\star(\mathrm{v}) = [\delta \dd\rho]$. As the element
$\mathrm{v} - \Curv^\star\big(\mathrm{q}^\star([\dd\rho])\big) \in \widehat{\mathfrak{G}}^{k}(M)$
lies in the kernel of $\iota^\star$, Remark \ref{rem:subtheoryonshell} implies that there exists
$\phi\in \Co{k}{M}{\bbZ}^\star$ such that $\mathrm{v} = \Curv^\star\big(\mathrm{q}^\star([\dd\rho])\big) + \Char^\star(\phi) = 
\mathrm{top}^\star\big([\dd\rho]\oplus \phi\big)$, and hence $\mathrm{v}$ lies 
in the image of $\mathfrak{Charge}^k(M)$ under $\mathrm{top}^\star : \mathfrak{Charge}^k(M)\to \widehat{\mathfrak{G}}^k(M)$.
\end{proof}
\begin{rem}
Notice that the converse of Lemma \ref{lem:topradconnection} is in general not true, i.e.\
the image of $\mathfrak{Charge}^k(M)$ under $\mathrm{top}^\star : \mathfrak{Charge}^k(M)\to \widehat{\mathfrak{G}}^k(M)$
is not necessarily a subgroup of the radical $\mathrm{Rad}\big(\, \widehat{\mathfrak{G}}^k(M)\, \big)$.
For example, for any object $M$ in $\mathsf{Loc}^m$ which has compact Cauchy surfaces (such as $M=\bbR\times \bbT^{m-1}$ 
equipped with the canonical Lorentzian metric), Lemma \ref{lem:radical} implies that the radical is the kernel of $\iota^\star$,
which by Remark \ref{rem:subtheoryonshell} is equal to the image of $\Char^\star$. If
$\Co{m-k}{M}{\bbR}^\star$ is non-trivial (as in the case $M= \bbR\times\bbT^{m-1}$ for any $1\leq k\leq m$),
then its image under $\Curv^\star\circ \mathrm{q}^\star$ is not contained in the radical.
\end{rem}
\sk

We can now give a characterization of the $\mathsf{Loc}^m$-morphisms which violate the locality axiom.
\begin{theo}\label{theo:injectivityproperty}
Let $f:M\to N$ be any $\mathsf{Loc}^m$-morphism. Then  the $C^\ast\mathsf{Alg}$-morphism
$\widehat{\mathfrak{A}}^k(f): \widehat{\mathfrak{A}}^k(M) \to \widehat{\mathfrak{A}}^k(N) $ is injective 
if and only if the $\mathsf{PAb}$-morphism 
$\mathfrak{Charge}^k(f) : \mathfrak{Charge}^k(M)\to \mathfrak{Charge}^k(N) $ 
is injective.
\end{theo}
\begin{proof}
We can simplify this problem by recalling from above that 
$\widehat{\mathfrak{A}}^k(f)$ is injective if and only if $\widehat{\mathfrak{G}}^k(f)$ is injective.
Furthermore, it is easier to prove the contraposition
``$\widehat{\mathfrak{G}}^k(f)$ not injective $\Leftrightarrow$ $\mathfrak{Charge}^k(f)$ not injective'',
which is equivalent to our theorem.
Our arguments will be based on the fact that $\mathrm{top}^\star : 
\mathfrak{Charge}^k(\,\cdot\,)\Rightarrow \widehat{\mathfrak{G}}^k(\,\cdot\,)$
is an {\it injective} natural transformation, so it is helpful to draw the corresponding commutative
diagram in the category $\mathsf{PAb}$ with exact vertical sequences:
\begin{flalign}\label{eqn:diagramlocalityproof}
\xymatrix{
\widehat{\mathfrak{G}}^k(M) \ \ar[rr]^-{\widehat{\mathfrak{G}}^k(f)}&&
\ \widehat{\mathfrak{G}}^k(N) \ \\
\ar[u]^-{\mathrm{top}^\star} \mathfrak{Charge}^k(M) \
\ar[rr]^-{\mathfrak{Charge}^k(f)} && \ \mathfrak{Charge}^k(N) \ar[u]_-{\mathrm{top}^\star}\\
0\ar[u] && 0 \ar[u]
}
\end{flalign}
Let us prove the direction ``$\Leftarrow$'': Assuming that $\mathfrak{Charge}^k(f)$ is not injective, the diagram 
(\ref{eqn:diagramlocalityproof}) implies that $\mathrm{top}^\star\circ\mathfrak{Charge}^k(f) =
 \widehat{\mathfrak{G}}^k(f)\circ \mathrm{top}^\star$  is not injective, and hence $\widehat{\mathfrak{G}}^k(f)$ is not injective
since $\mathrm{top}^\star$ is injective.
To prove the direction ``$\Rightarrow$'' let us assume that $\widehat{\mathfrak{G}}^k(f)$ is not injective. 
By Lemma \ref{lem:topradconnection} the kernel of $\widehat{\mathfrak{G}}^k(f)$ is a subgroup of the image
of $\mathfrak{Charge}^k(M)$ under $\mathrm{top}^\star : \mathfrak{Charge}^k(M)\to \widehat{\mathfrak{G}}^k(M)$,
hence $\widehat{\mathfrak{G}}^k(f)\circ \mathrm{top}^\star $ is not injective. The commutative diagram
(\ref{eqn:diagramlocalityproof}) then implies that $\mathrm{top}^\star\circ\mathfrak{Charge}^k(f)$ is not injective,
hence $\mathfrak{Charge}^k(f)$ is not injective since $\mathrm{top}^\star$ is injective.
\end{proof}

\begin{ex}\label{ex:localityviolation}
We provide explicit examples of $\mathsf{Loc}^m$-morphisms $f:M\to N$ which violate the locality axiom.
Let us take as $\mathsf{Loc}^m$-object $N=\bbR^m$, the $m$-dimensional oriented and time-oriented 
Minkowski spacetime. Choosing any Cauchy surface $\Sigma_N = \bbR^{m-1}$ in $N$,
we take the subset $\Sigma_M := \big(\bbR^{p}\setminus \{0\} \big) \times \bbR^{m-1-p}\subseteq \Sigma_N$, where
we have removed the origin $0$ of a $p$-dimensional subspace with $1\leq p\leq m-1$.
We take the Cauchy development of $\Sigma_M$ in $N$ (which we denote by $M$) 
and note that by \cite[Lemma A.5.9]{Bar:2007zz} $M$ is a causally compatible, open and globally hyperbolic
subset of $N$. The canonical inclusion provides us with a $\mathsf{Loc}^m$-morphism
$\iota_{N;M} : M\to N$. Using the diffeomorphism $\bbR^p \setminus \{0\} \simeq \bbR\times \mathbb{S}^{p-1}$ with the
$p{-}1$-sphere $\mathbb{S}^{p-1}$ (in our conventions $\mathbb{S}^0 := \{-1,+1\}$), we find that
$M\simeq \bbR^{m-p+1}\times \mathbb{S}^{p-1}$. Using the fact that the singular cohomology groups are homotopy invariant, we obtain
\begin{subequations}\label{eqn:chargegroupsexample}
\begin{flalign}
\mathfrak{Charge}^k(M) &\simeq \Co{m-k}{\mathbb{S}^{p-1}}{\bbR}^\star \oplus \Co{k}{\mathbb{S}^{p-1}}{\bbZ}^\star ~,\\[4pt]
\mathfrak{Charge}^k(N) &\simeq  \Co{m-k}{\mathrm{pt}}{\bbR}^\star \oplus \Co{k}{\mathrm{pt}}{\bbZ}^\star~,
\end{flalign}
\end{subequations}
where $\mathrm{pt}$ denotes a single point. By Theorem \ref{theo:injectivityproperty}, the $C^\ast\mathsf{Alg}$-morphism
$\widehat{\mathfrak{A}}^k(\iota_{N;M})$ is not injective if and only if  $\mathfrak{Charge}^k(\iota_{N;M})$ is not injective.
The following choices of $p$ lead to $\mathsf{Loc}^m$-morphisms $\iota_{N;M}: M\to N$ which violate the locality axiom:
\begin{itemize}
\item $k=1$: Since by assumption $m\geq 2$, the second isomorphism in (\ref{eqn:chargegroupsexample}) implies
$\mathfrak{Charge}^1(N) =0$. Since $1\leq p\leq m-1$, we have $\Co{m-1}{\mathbb{S}^{p-1}}{\bbR}^\star=0$ and hence 
the first isomorphism in (\ref{eqn:chargegroupsexample}) implies $\mathfrak{Charge}^1(M) \simeq \Co{1}{\mathbb{S}^{p-1}}{\bbZ}^\star$.
For $m\geq 3$ we choose $p=2$ and find that $\mathfrak{Charge}^1(M) \simeq \bbZ^\star\simeq \bbT$, hence 
$\mathfrak{Charge}^1(\iota_{N;M})$
is not injective (being a group homomorphism $\bbT\to 0$). The case $m=2$ is special and is discussed in detail below.

\item $2\leq k\leq m-1$: The second isomorphism in (\ref{eqn:chargegroupsexample}) implies $\mathfrak{Charge}^k(N) =0$.  Choosing $p=m-k+1$ (which is admissible since $2\leq p \leq m-1$),
the first isomorphism in (\ref{eqn:chargegroupsexample}) gives $\mathfrak{Charge}^k(M) \simeq \bbR \oplus \delta_{k,m-k}\, \bbT$,
where $\delta_{k,m-k}$ denotes the Kronecker delta. Hence
$\mathfrak{Charge}^k(\iota_{N;M})$ 
is not injective  (being a group homomorphism $ \bbR \oplus \delta_{k,m-k}\bbT\to 0$). Alternatively, if $2\leq k\leq m-2$ we may also choose $p=k+1$ and find via the
first isomorphism in (\ref{eqn:chargegroupsexample}) that $\mathfrak{Charge}^k(M) \simeq \delta_{k,m-k} \, \bbR \oplus \bbT$,
which also implies that $\mathfrak{Charge}^k(\iota_{N;M})$ is not injective.

\item $k=m$: The second isomorphism in (\ref{eqn:chargegroupsexample}) implies
$\mathfrak{Charge}^m(N) \simeq \bbR$. Choosing $p=1$ (which is admissible since $m\geq 2$)
we obtain $\mathfrak{Charge}^m(M) \simeq \bbR^2$, hence $\mathfrak{Charge}^m(\iota_{N;M})$ is not injective (being a group
homomorphism $\bbR^2\to\bbR$).
\end{itemize}
\end{ex}
\begin{cor}
Choose any $m\geq 2$ and $1\leq k\leq m$ such that $(m,k)\neq (2,1)$. Then the quantum field theory functor
$\widehat{\mathfrak{A}}^k(\,\cdot\,): \mathsf{Loc}^m\to C^\ast\mathsf{Alg}$ violates the locality axiom.
\end{cor}
\begin{proof}
This follows from the explicit examples of $\mathsf{Loc}^m$-morphisms given in Example \ref{ex:localityviolation}
and Theorem \ref{theo:injectivityproperty}.
\end{proof}

The case $m=2$ and $k=1$ is special. As any object $M$ in $\mathsf{Loc}^2$ is a two-dimensional globally hyperbolic
spacetime, there exists a one-dimensional Cauchy surface $\Sigma_M$ such that $M\simeq \bbR\times\Sigma_M$.
By the classification of one-manifolds (without boundary), $\Sigma_M$ is diffeomorphic to the disjoint union
of copies of $\bbR$ and $\bbT$, i.e.\ $\Sigma_M\simeq \coprod_{i=1}^{n_M}\, \bbR \sqcup \coprod_{j=1}^{c_M}\, \bbT$
(the natural numbers $n_M$ and $c_M$ are finite, since $M$ is assumed to be of finite-type).
By homotopy invariance, we have
\begin{flalign}
\mathfrak{Charge}^1(M) \simeq \Co{1}{\Sigma_M}{\bbR}^\star \oplus \Co{1}{\Sigma_M}{\bbZ}^\star\simeq 
\bbR^{c_M} \oplus  \bbT^{c_M}~.
\end{flalign}
As any $\mathsf{Loc}^m$-morphism $f:M\to N$ is in particular an embedding, 
the number of compact components in the Cauchy surfaces $\Sigma_M$ and $\Sigma_N$ cannot decrease,
i.e.\ $c_M \leq c_N$. As a consequence, $\mathfrak{Charge}^1(f)$ is injective and by Theorem \ref{theo:injectivityproperty}
so is $\widehat{\mathfrak{A}}^1(f)$.
\begin{propo}
The quantum field theory functor
$\widehat{\mathfrak{A}}^1(\,\cdot\,): \mathsf{Loc}^2\to C^\ast\mathsf{Alg}$ satisfies the locality axiom.
Thus it is a locally covariant quantum field theory in the sense of \cite{Brunetti:2001dx}.
\end{propo}

\section*{Acknowledgements}
We thank Christian B\"ar, Marco Benini, Claudio Dappiaggi, Chris Fewster, Klaus Fredenhagen, 
Florian Hanisch and Igor Khavkine for useful comments on
this work. A.S.\ would like to thank the Department of Mathematics of the University of Potsdam for hospitality during
a visit in May 2014, where parts of this research have been conducted. A.S.\ further thanks
the Erwin Schr\"odinger Institute in Vienna for support and hospitality during the Workshop 
``Algebraic Quantum Field Theory: Its Status and its Future'', where we received useful feedback and suggestions on this work
from the organizers and participants.
The work of C.B.\ is partially supported by the Sonderforschungsbereich ``Raum Zeit Materie'' funded 
by the Deutsche Forschungsgemeinschaft (DFG).
The work of A.S.\ is supported by a DFG Research Fellowship. 
The work of R.J.S.\ is supported in part by the Consolidated Grant ST/J000310/1 from the UK Science and Technology Facilities Council.


\appendix

\section{\label{app:frechet}Fr{\'e}chet-Lie group structures}

In this appendix we show how to equip the differential cohomology groups
$\DCo{k}{M}{\bbZ}$, as well as all other Abelian groups in the diagram (\ref{eqn:csdiagrm}),
 with the structure of an Abelian Fr{\'e}chet-Lie group such that 
 the morphisms in this diagram are smooth maps. 
Our notion of functional derivatives in Definition \ref{def:funcder}
then coincides with directional derivatives along tangent vectors corresponding 
to this Abelian Fr{\'e}chet-Lie group structure.
Furthermore, we show that the contravariant functor
$\DCo{k}{\,\cdot\,}{\bbZ} : \mathsf{Loc}^m\to \mathsf{Ab}$ can be promoted to a functor to the category of
Abelian Fr{\'e}chet-Lie groups. 
For the notions of Fr\'echet manifolds and Fr\'echet-Lie groups we refer to \cite{Hamilton:1982}.
\sk

Let $M$ be any smooth manifold that is of finite-type. The Abelian groups in the lower 
horizontal sequence in the diagram (\ref{eqn:csdiagrm}) are finitely generated discrete groups, hence we shall equip them with 
the discrete topology and therewith obtain zero-dimensional Abelian Fr{\'e}chet-Lie groups. All arrows 
in the lower horizontal sequence in (\ref{eqn:csdiagrm}) then become morphisms of Abelian Fr{\'e}chet-Lie groups.
Next, we consider the upper horizontal sequence in (\ref{eqn:csdiagrm}).
We endow the $\bbR$-vector space of differential $p$-forms $\Omega^p(M)$ with the natural $C^\infty$-topology, 
i.e.~the topology of uniform convergence together with all derivatives on any compact set $K \subset M$.
An elegant way to describe the $C^\infty$-topology is by choosing an auxiliary Riemannian metric $g$ on $M$ and a 
countable compact exhaustion $K_0 \subset K_1\subset \cdots \subset K_n \subset K_{n+1} \subset \cdots \subset M$, with $n \in \bbN$. 
We define the family of semi-norms
\begin{flalign}
\Vert \omega\Vert_{l,n} :=\max_{j=0,1,\dots, l} \ \max_{x\in K_n} \, \vert D^j \omega  (x)\vert 
\end{flalign}
for all $l,n\in \bbN$ and $\omega\in \Omega^p(M)$, where $D^j : 
\Omega^p(M) \to \Gamma^\infty\big(M,\bigwedge^p T^\ast M \otimes \bigvee^j T^\ast M\big)$ 
is the symmetrized covariant derivative corresponding to the Riemannian metric $g$ 
and $\vert\,\cdot\,\vert$ is the fibre metric on $ \bigwedge^p T^\ast M \otimes \bigvee^j T^\ast M$ induced by $g$.
The $C^\infty$-topology on $\Omega^p(M)$ is the Fr{\'e}chet topology induced 
by the family of semi-norms $\Vert\,\cdot\,\Vert_{l,n}$ with $l,n\in\bbN$. It is easy to check that
this topology does not depend on the choice of Riemannian metric $g$ and compact exhaustion $K_n$,
as for different choices of $g$ and $K_n$ the corresponding semi-norms can be estimated against each other.
\sk

The subspace of exact forms $\dd \Omega^{k-1}(M) \subseteq \Omega^k(M)$ is a closed subspace in the $C^\infty$-topology
on $\Omega^k(M)$, hence $\dd \Omega^{k-1}(M)$ is a Fr{\'e}chet space in its own right. Forgetting the multiplication by 
scalars, the Abelian group $\dd \Omega^{k-1}(M)$ (with respect to $+$) in the upper right corner of (\ref{eqn:csdiagrm}) 
is an Abelian Fr{\'e}chet-Lie group. Let us now describe the Abelian Fr{\'e}chet-Lie group structure
on $\Omega^{k-1}(M)/\Omega^{k-1}_\bbZ(M)$. For us it will be convenient to provide an 
explicit description by using charts. As model space we shall take the Fr{\'e}chet space
$\Omega^{k-1}(M)/\dd\Omega^{k-2}(M)$. To specify a topology on $\Omega^{k-1}(M)/\Omega^{k-1}_\bbZ(M)$ by 
defining a neighborhood basis around every point $[\eta]\in \Omega^{k-1}(M)/\Omega^{k-1}_\bbZ(M)$, notice that
 $\Omega^{k-1}(M)/\Omega^{k-1}_\bbZ(M)$ is the quotient of $ \Omega^{k-1}(M)/\dd\Omega^{k-2}(M)$ 
 by the finitely generated subgroup $H^{k-1}_\mathrm{free}(M;\bbZ)$.
Let $V_n \subseteq \Omega^{k-1}(M)/\dd\Omega^{k-2}(M)$, with $n\in \bbN$, be a countable
neighborhood basis of $0\in\Omega^{k-1}(M)/\dd\Omega^{k-2}(M) $, which consists of ``small open sets'' in the sense that
$V_n\cap H^{k-1}_\mathrm{free}(M;\bbZ) = \{0\}$ for all $n\in \bbN$. 
Thus the quotient map $q : \Omega^{k-1}(M)/\dd\Omega^{k-2}(M) \to \Omega^{k-1}(M)/\Omega_\bbZ^{k-1}(M)$ 
is injective when restricted to $V_n$. 
We may assume without loss of generality that $V_n$ is symmetric, i.e.\ $-V_n = V_n$.
We now define a neighborhood basis around every point $[\eta]\in \Omega^{k-1}(M)/\Omega^{k-1}_\bbZ(M)$ by setting
\begin{flalign}
U_{[\eta],n} := [\eta] + q\big[V_n\big] \ \subseteq \ \frac{\Omega^{k-1}(M)}{\Omega_\bbZ^{k-1}(M)} ~. 
\end{flalign}
The topology  induced by the basis $\big\{U_{[\eta],n} : 
[\eta] \in \Omega^{k-1}(M)/\Omega_\bbZ^{k-1}(M)~,~n\in\bbN\big\}$ makes the quotient
 $\Omega^{k-1}(M)/\Omega_\bbZ^{k-1}(M)$ into a Fr{\'e}chet manifold:
as charts around $[\eta]$ we may take the maps
\begin{flalign}
\psi_{[\eta],n} \,:\, U_{[\eta],n} \ \longrightarrow \ V_n \ \subseteq \ \frac{\Omega^{k-1}(M)}{\dd\Omega^{k-2}(M)}~, \qquad [\eta] 
+ q([\omega]) \ \longmapsto \ [\omega]~.
\end{flalign}
The change of coordinates is then given by affine transformations
$ V\to [\omega^\prime\, ] + V\,,~[\omega]\mapsto [\omega^\prime\, ] + [\omega]$ on open subsets 
$V\subseteq \Omega^{k-1}(M)/\dd\Omega^{k-2}(M)$, which are smooth maps.
Fixing any point $[\eta^\prime\, ]\in \Omega^{k-1}(M)/\Omega^{k-1}_\bbZ(M)$,
consider the map 
$[\eta^\prime\, ] + (\,\cdot\,) : \Omega^{k-1}(M)/\Omega^{k-1}_\bbZ(M) \to \Omega^{k-1}(M)/\Omega^{k-1}_\bbZ(M) $.
The inverse image of any open neighborhood $U_{[\eta],n}$ is given by
\begin{flalign}
\big([\eta^\prime\, ] + (\,\cdot\,) \big)^{-1}\big(U_{[\eta],n}\big) = [\eta]-[\eta^\prime\, ] + q\big[V_n\big] = U_{[\eta]-[\eta^\prime\, ],n}~,
\end{flalign}
which is open.
Hence $+$ is continuous. 
Analogously, we see that the group inverse is continuous
since the inverse image of any $U_{[\eta],n}$ is $U_{-[\eta],n}$. 
To see that the group operations are also smooth, notice that for any $U_{[\eta],n}$ 
the map $[\eta^{\prime}\, ] + (\,\cdot\,) : U_{[\eta] - [\eta^\prime\, ], n}\to U_{[\eta],n} $
induces the identity map $\id_{V_n} : V_n\to V_n$ in the charts $\psi_{[\eta]-[\eta^\prime\, ],n}$ and $\psi_{[\eta],n}$.
Similarly the inverse $-: U_{-[\eta],n} \to U_{[\eta],n}$ induces minus the identity map
 $-\id_{V_n} : V_n\to V_n$ in the charts $\psi_{-[\eta],n}$ and $\psi_{[\eta],n}$.
Hence $\Omega^{k-1}(M)/\Omega^{k-1}_\bbZ(M)$ is an Abelian Fr{\'e}chet-Lie group.
Since $\Omega^{k-1}(M)/\Omega^{k-1}_\bbZ(M)$ carries the quotient topology induced from the Fr\'echet space
 $\Omega^{k-1}(M)/\dd\Omega^{k-2}(M)$ and the exterior differential $\dd: \Omega^{k-1}(M) \to \dd\Omega^{k-1}(M)$
  is smooth, the same holds for the induced map $\dd: \Omega^{k-1}(M)/\Omega^{k-1}_\bbZ(M) \to \dd\Omega^{k-1}(M)$.
\sk

The Abelian Fr{\'e}chet-Lie group structure on $H^{k-1}(M;\bbR)/H^{k-1}_\mathrm{free}(M;\bbZ)$ is modeled
on the subspace $H^{k-1}(M;\bbR) \simeq \Omega^{k-1}_\dd(M)/\dd\Omega^{k-2}(M) 
\subseteq \Omega^{k-1}(M)/\dd\Omega^{k-2}(M) $ equipped with the subspace topology.
Explicitly, we define charts on 
$H^{k-1}(M;\bbR)/H^{k-1}_\mathrm{free}(M;\bbZ)$
by noticing that
\begin{flalign}
\frac{H^{k-1}(M;\bbR)}{H^{k-1}_\mathrm{free}(M;\bbZ)} \simeq \frac{\Omega^{k-1}_\dd(M)}{\Omega^{k-1}_\bbZ(M)} \ \subseteq \ \frac{\Omega^{k-1}(M)}{\Omega^{k-1}_\bbZ(M)}
\end{flalign}
and setting
\begin{flalign}
\tilde\psi_{[\eta],n} \, : \, \Big(\, U_{[\eta],n}\cap\frac{\Omega^{k-1}_\dd(M)}{\Omega^{k-1}_\bbZ(M)}\, \Big) \ 
& \longrightarrow \ \Big(\, V_n\cap \frac{\Omega_\dd^{k-1}(M)}{\dd\Omega^{k-2}(M)}\, \Big)~, \nonumber\\ \qquad [\eta] + q([\omega]) \ & \longmapsto \ [\omega]
\end{flalign}
for any $[\eta]\in \Omega^{k-1}_\dd(M)/\Omega^{k-1}_\bbZ(M)$ and $n\in\bbN$. 
This defines on  $H^{k-1}(M;\bbR)/H^{k-1}_\mathrm{free}(M;\bbZ)$ an Abelian Fr{\'e}chet-Lie group structure.
The smooth inclusion $\Omega^{k-1}_\dd(M) \hookrightarrow \Omega^{k-1}(M)$ of Fr\'echet spaces descends to 
a smooth inclusion $H^{k-1}(M;\bbR)/H^{k-1}_\mathrm{free}(M;\bbZ)\hookrightarrow \Omega^{k-1}(M)/\Omega^{k-1}_\bbZ(M)$ 
of Abelian Fr\'echet-Lie groups. 
Hence we have shown that all arrows in the upper horizontal sequence in (\ref{eqn:csdiagrm}) are morphisms
of Abelian Fr{\'e}chet-Lie groups.
\sk

From the Abelian Fr{\'e}chet-Lie group structure on the groups in the lower and upper
 horizontal sequence in  (\ref{eqn:csdiagrm}) we can derive an Abelian Fr{\'e}chet-Lie group 
 structure on the groups in the middle horizontal sequence.
Our strategy makes use of the vertical exact sequences in (\ref{eqn:csdiagrm}). As the construction 
is the same for all three vertical sequences, it is enough to discuss the example of the middle vertical sequence. 
As $\Char: \DCo{k}{M}{\bbZ} \to \Co{k}{M}{\bbZ}$ is an Abelian group homomorphism to a discrete Abelian Fr{\'e}chet-Lie group,
 the connected components of $\DCo{k}{M}{\bbZ}$ are precisely the fibers of the characteristic class map $\Char$.
Thus we shall take as model space for $\DCo{k}{M}{\bbZ}$ the same Fr{\'e}chet space $\Omega^{k-1}(M)/\dd\Omega^{k-2}(M)$ 
as for the group $\Omega^{k-1}(M)/\Omega^{k-1}_\bbZ(M)$ describing the
kernel of $\Char$.  
Explicitly, we take the topology on $\DCo{k}{M}{\bbZ}$ which is generated by the basis
\begin{flalign}
\widehat{U}_{h,n} := h + (\iota \circ q)\big[V_n\big] \ \subseteq \ \DCo{k}{M}{\bbZ} \ ,
\end{flalign}
for all $h\in \DCo{k}{M}{\bbZ} $ and $n\in \bbN$. The charts are given by
\begin{flalign}
\widehat{\psi}_{h,n} \,:\, \widehat{U}_{h,n} \ \longrightarrow \ V_n~, \qquad h + (\iota\circ q)([\omega]) \ \longmapsto \ [\omega]~,
\end{flalign}
and the proof that the change of coordinates is smooth is the same as above.
By construction of the topology on $\DCo{k}{M}{\bbZ}$, the group operations are smooth, 
since it suffices to consider them along the connected components and there the proof reduces
to the one above for the smoothness of the group operations on $\Omega^{k-1}(M)/\Omega^{k-1}_\bbZ(M)$. 
Again by construction, the topological trivialization $\iota: \Omega^{k-1}(M)/\Omega^{k-1}_\bbZ(M) \to \DCo{k}{M}{\bbZ}$ 
is a diffeomorphism onto the connected component of ${\bf 0} \in \DCo{k}{M}{\bbZ}$.
The characteristic class $\Char:\DCo{k}{M}{\bbZ} \to H^k(M;\bbZ)$ is smooth, since $H^k(M;\bbZ)$ carries the discrete topology.
Likewise, the inclusions $H^{k-1}(M;\bbR) /H^{k-1}_\mathrm{free}(M;\bbZ) \hookrightarrow H^{k-1}(M;\bbT)$ 
and $\dd\Omega^{k-1}(M) \hookrightarrow \Omega^k_\bbZ(M)$, as well as the projections
$H^{k-1}(M;\bbT) \to H^{k}_\mathrm{tor}(M;\bbZ)$ and $\Omega^k_\bbZ(M)\to H^{k}_\mathrm{free}(M;\bbZ)$, in the left 
and right vertical sequences in (\ref{eqn:csdiagrm}) are smooth maps.
\sk

It remains to show that the arrows in the middle horizontal sequence of (\ref{eqn:csdiagrm}) are smooth:
A basis for the topology on $\Omega^k_{\bbZ}(M)$ is given by sets of the form $\omega + V$, where $\omega \in \Omega^k_{\bbZ}(M)$ and
the sets $V$ are taken from an open neighborhood basis of $0\in \dd\Omega^{k-1}(M)$. 
The inverse image of an open set $\omega + V$ under $\Curv: \DCo{k}{M}{\bbZ}\to \Omega^k_\bbZ(M)$ is the union over all 
$h \in \Curv^{-1}(\omega)$ of the open sets $h + \iota[\dd^{-1}(V)]$, hence it is 
open and $\Curv$ is continuous. Furthermore, the curvature $\Curv$ is smooth with differential 
$D_h\Curv([\omega]) = \dd\omega$, where $h\in \DCo{k}{M}{\bbZ} $ and
 $[\omega] \in \Omega^{k-1}(M)/\dd\Omega^{k-2}(M) = T_h \DCo{k}{M}{\bbZ}$ is a tangent vector.
By a similar argument, the group homomorphism $\kappa: H^{k-1}(M;\bbT)\to \DCo{k}{M}{\bbZ}$ is 
smooth with differential given by the inclusion 
\begin{flalign}
T_{[\phi]}\Big(\, \frac{H^{k-1}(M;\bbR)}{H^{k-1}_\mathrm{free}(M;\bbZ)}\, \Big) = H^{k-1}(M;\bbR) \ \hookrightarrow \ \frac{\Omega^{k-1}(M)}{\dd\Omega^{k-2}(M)} = T_{\kappa([\phi])}\DCo{k}{M}{\bbZ} ~ ,
\end{flalign}
where $[\phi] \in H^{k-1}(M;\bbR) / H^{k-1}_\mathrm{free}(M;\bbZ)$.
\sk

With respect to the Abelian Fr{\'e}chet-Lie group structure developed above, the tangent space at a point $h\in \DCo{k}{M}{\bbZ}$
is given by the model space $\Omega^{k-1}(M)/\dd\Omega^{k-2}(M)$. The functional derivative
given in Definition \ref{def:funcder} is the directional derivative along tangent vectors.
\sk

It remains to show that the contravariant functor $\DCo{k}{\,\cdot\,}{\bbZ} : \mathsf{Loc}^m\to \mathsf{Ab}$ 
can be promoted to a functor with values in the category of Abelian Fr{\'e}chet-Lie groups. For any smooth map 
$f:M\to N$ the pull-back of differential forms $f^\ast: \Omega^{k-1}(N)\to \Omega^{k-1}(M)$ is 
smooth with respect to the Fr{\'e}chet space structure on differential forms. 
The same holds true for the induced map on the quotients $f^\ast: \Omega^{k-1}(N)/\dd\Omega^{k-2}(N) 
\to \Omega^{k-1}(M)/\dd\Omega^{k-2}(M)$.
Since the characteristic class is a natural transformation $\Char: \DCo{k}{\,\cdot\,}{\bbZ} \Rightarrow H^k(\,\cdot\,;\bbZ)$, 
the argument for smoothness of the pull-back directly carries over from the groups in the upper row of the 
diagram (\ref{eqn:csdiagrm}) to the middle row:
As above, by construction of the topology on the differential cohomology groups it suffices to consider 
$\DCo{k}{f}{\bbZ}:\DCo{k}{N}{\bbZ} \to \DCo{k}{M}{\bbZ}$ on the connected components, i.e.~along the 
fibers of the characteristic class.
Thus $\DCo{k}{f}{\bbZ}:\DCo{k}{N}{\bbZ} \to \DCo{k}{M}{\bbZ}$ is a smooth map with differential
\begin{align}
D_h\DCo{k}{f}{\bbZ} \,:\,  T_h\DCo{k}{N}{\bbZ} = \frac{\Omega^{k-1}(N)}{\dd\Omega^{k-2}(N)} 
& \ \longrightarrow \ T_{\DCo{k}{f}{\bbZ}(h)}\DCo{k}{M}{\bbZ} = \frac{\Omega^{k-1}(M)}{\dd\Omega^{k-2}(M)} \ , \nonumber \\
[\omega] 
& \ \longmapsto \ f^\ast\big([\omega]\big) \ . 
\end{align}
It follows that $\DCo{k}{\,\cdot\,}{\bbZ} : \mathsf{Loc}^m\to \mathsf{Ab}$ extends to a contravariant functor 
to the category of Abelian Fr\'echet-Lie groups, and all natural transformations in the definition of a differential 
cohomology theory are natural transformations of functors in this sense.

\subsection{Isomorphism types}

We shall now identify the isomorphism types of the Fr\'echet-Lie groups 
$\Omega^{k-1}(M) / \Omega^{k-1}_\bbZ(M)$ and $\widehat H^k(M;\bbZ)$ by splitting the rows in the diagram \eqref{eqn:csdiagrm}.
The lower row splits since $H^k_\mathrm{free}(M;\bbZ)$ is a free Abelian group and all groups in the lower row carry the discrete topology.  
By construction, all rows in \eqref{eqn:csdiagrm} are central extensions of Abelian Fr\'echet-Lie groups.
In particular, they define principal bundles over the groups in the right column. In the following we denote the $k$-th Betti number
of $M$ by $b_k$ with $k\in\bbN$; then all $b_k<\infty$ by the assumption that $M$ is of finite-type.
\sk

For the upper row, notice that $\dd\Omega^{k-1}(M)$ is contractible, hence the corresponding torus bundle is topologically trivial.
In fact, it is trivial as a central extension, i.e.~the Fr\'echet-Lie group $\Omega^{k-1}(M) / \Omega^{k-1}_\bbZ(M)$
is (non-canonically isomorphic to) the topological direct sum of the $b_{k-1}$-torus
 $H^{k-1}(M;\bbR)/H^{k-1}_\mathrm{free}(M;\bbZ)$ 
 with the Fr\'echet space $\dd\Omega^{k-1}(M)$:
Any choice of forms $\omega^1,\ldots,\omega^{b_{k-1}} \in \Omega_\dd^{k-1}(M)$ whose de Rham 
classes form a $\bbZ$-basis of $H^{k-1}_\mathrm{free}(M;\bbZ)$ provides us with (non-canonical) topological splittings
\begin{subequations}
\begin{flalign}
\Omega^{k-1}_\dd(M) 
&= \mathrm{span}_\bbR \big\{\omega^1,\ldots,\omega^{b_{k-1}}\big\} \oplus \dd\Omega^{k-2}(M) \ , \\[4pt]
\Omega^{k-1}_\bbZ(M) 
&= \mathrm{span}_\bbZ \big\{\omega^1,\ldots,\omega^{b_{k-1}}\big\} \oplus \dd\Omega^{k-2}(M) \ , \label{eq:OmegaZsplit}\\[4pt]
\Omega^{k-1}(M) 
&= \mathrm{span}_\bbR \big\{\omega^1,\ldots,\omega^{b_{k-1}}\big\} \oplus F^{k-1}(M) \ . \label{eq:Omegasplit}
\end{flalign}
\end{subequations}
Here $F^{k-1}(M) \subseteq \Omega^{k-1}(M)$ is a topological complement of the subspace spanned by the 
$k{-}1$-forms $\omega^1,\ldots,\omega^{b_{k-1}}$.
By a Hahn-Banach type argument, we may choose the complement $F^{k-1}(M)$ such that $\dd \Omega^{k-2}(M) \subseteq F^{k-1}(M)$:
Taking de Rham cohomology classes yields a continuous isomorphism 
\begin{flalign}
[\,\cdot\,]:\mathrm{span}_\bbR \big\{\omega^1,\ldots,\omega^{b_{k-1}}\big\} \to H^{k-1}(M;\bbR)~.
\end{flalign}
Thus we obtain a continuous projection $p:\Omega^{k-1}_\dd(M) \to \mathrm{span}_\bbR \big\{\omega^1,\ldots,\omega^{b_{k-1}}\big\}$ with kernel $\dd\Omega^{k-2}(M)$.
Denote by $\mathrm{pr}_j: \mathrm{span}_\bbR\big\{\omega^1,\ldots,\omega^{b_{k-1}} \big\} \to \bbR \omega^j$ the projection to the $j$-th component.
Then the continuous linear functionals $p_j := \mathrm{pr}_j \circ p: \Omega^{k-1}_\dd(M) \to \bbR \omega^j$, 
with $j \in \{1,\ldots,b_{k-1}\}$, have continuous extensions to $\Omega^{k-1}(M)$, and so does their direct sum
 $p= p_1 \oplus \cdots \oplus p_{b_{k-1}}:\Omega^{k-1}_\dd(M) \to \mathrm{span}_\bbR \big\{\omega^1,\ldots,\omega^{b_{k-1}}\big\}$. 
Then put $F^{k-1}(M) := \ker(p)$ to obtain a decomposition of $\Omega^{k-1}(M)$ as claimed.
By construction, the exterior differential induces a continuous isomorphism of Fr\'echet spaces $\dd:F^{k-1}(M) /\dd\Omega^{k-2}(M) \to \dd\Omega^{k-1}(M)$.
This yields the decomposition of Abelian Fr\'echet-Lie groups
\begin{flalign}\label{eq:Omegaquotientsplit}
\frac{\Omega^{k-1}(M)}{\Omega^{k-1}_\bbZ(M)} 
=
\frac{\mathrm{span}_\bbR \big\{\omega^1,\ldots,\omega^{b_{k-1}}\big\}}{\mathrm{span}_\bbZ \big\{\omega^1,\ldots,\omega^{b_{k-1}}\big\}} \oplus \frac{F^{k-1}(M)}{\dd\Omega^{k-2}(M)} \
\xrightarrow{[\, \cdot \,] \oplus \dd} \ 
\frac{H^{k-1}_{\phantom{\mathrm{free}}}(M;\bbR)}{H^{k-1}_\mathrm{free}(M;\bbZ)} \oplus \dd\Omega^{k-1}(M) \ .
\end{flalign}
Thus $\Omega^{k-1}(M)/\Omega^{k-1}_\bbZ(M)$ is the direct sum in the category of Abelian 
Fr\'echet-Lie groups of the $b_{k-1}$-torus $H^{k-1}(M;\bbR)/H^{k-1}_\mathrm{free}(M;\bbZ)$ and the
 additive group of the Fr\'echet space $\dd\Omega^{k-1}(M)$, as claimed.
\sk

For the middle row in the diagram (\ref{eqn:csdiagrm}), since the connected components of $\Omega^k_\bbZ(M)$ are contractible, the corresponding principal 
$H^{k-1}(M;\bbT)$-bundle $\Curv: \DCo{k}{M}{\bbZ} \to \Omega^k_\bbZ(M)$ is topologically trivial.
We can also split the middle exact sequence in the diagram (\ref{eqn:csdiagrm}) as a central extension:
Choose differential forms $\vartheta^{1},\ldots,\vartheta^{b_k} \in \Omega^k_\dd(M)$ whose de Rham 
classes form a $\bbZ$-module basis of $H^k_\mathrm{free}(M;\bbZ)$; this yields a splitting of $\Omega^k_\bbZ(M)$ analogous to the one in \eqref{eq:OmegaZsplit}. 
Thus we may write any form $\mu \in \Omega^k_\bbZ(M)$ as $\mu = \sum_{i=1}^{b_k}\, a_i \,\vartheta^i + \dd \nu$, 
where $a_i \in \bbZ$ and $\nu \in \Omega^{k-1}(M)$. 
Now choose elements $h_{\vartheta^i} \in \widehat H^k(M;\bbZ)$ with curvature $\Curv(h_{\vartheta^i})=\vartheta^i$,
for all $i=1,\dots,b_k$.
By the splitting \eqref{eq:Omegaquotientsplit} we may choose a Fr\'echet-Lie group homomorphism
 $\sigma^\prime : \dd\Omega^{k-1}(M) \to \Omega^{k-1}(M)/\Omega^{k-1}_\bbZ(M)$ such that
  $\dd \circ \sigma^\prime = \mathrm{id}_{\dd\Omega^{k-1}(M)}$.
Now we define a splitting of the middle row in the diagram (\ref{eqn:csdiagrm}) by setting
\begin{flalign}
\nn \sigma \,:\, \Omega^k_\bbZ(M) = \mathrm{span}_\bbZ\big\{\vartheta^1,\ldots,\vartheta^{b_k} \big\} \oplus \dd\Omega^{k-1}(M) 
& \ \longrightarrow \ \widehat H^k(M;\bbZ)~, \\
\sum_{i=1}^{b_k}\, a_i \, \vartheta^i + \dd\nu
& \ \longmapsto \ \sum_{i=1}^{b_k}\, a_i \, h_{\vartheta^i} + \iota\big(\sigma^\prime(\dd\nu)\big)~.
\end{flalign}
By construction, $\sigma$ is a homomorphism of Abelian Fr\'echet-Lie groups, i.e.~it is a smooth group homomorphism.
Moreover, for any form $\mu = \sum_{i=1}^{b_k}\, a_i \, \vartheta^i + \dd\nu \in \Omega^k_\bbZ(M)$
we have $\Curv(\sigma(\mu)) = \sum_{i=1}^{b_k} \, a_i \, \Curv(h_{\vartheta^i}) + \dd\nu = \mu$.
Thus $\sigma$ is a splitting of the middle horizontal sequence of Abelian Fr{\'e}chet-Lie groups
in the diagram (\ref{eqn:csdiagrm}), and we have obtained a (non-canonical) decomposition
\begin{flalign}
\widehat H^k(M;\bbZ) 
\simeq H^{k-1}(M;\bbT) \oplus \Omega^{k}_\bbZ(M) 
\end{flalign}
of Abelian Fr\'echet-Lie groups.


\end{document}